\documentclass[a4paper,UKenglish,cleveref, autoref, thm-restate]{lipics-v2021}
\pdfoutput=1 

\hideLIPIcs

\title{Verification of Population Protocols with Unordered Data}

\author{Steffen van Bergerem}{Humboldt-Universität zu Berlin, Germany}{steffen.van.bergerem@informatik.hu-berlin.de}{https://orcid.org/0000-0002-5212-8992}{This work was funded by the Deutsche Forschungsgemeinschaft (DFG, German Research Foundation) --- project number 431183758 (gefördert durch die Deutsche Forschungsgemeinschaft (DFG) --- Projektnummer 431183758).}
\author{Roland Guttenberg}{Technische Universität München, Germany}{guttenbe@in.tum.de}{https://orcid.org/0000-0001-6140-6707}{}
\author{Sandra Kiefer}{University of Oxford, United Kingdom}{sandra.kiefer@cs.ox.ac.uk}{https://orcid.org/0000-0003-4614-9444}{This research was supported by the Glasstone Benefaction, University of Oxford [Violette and Samuel Glasstone Research Fellowships in Science 2022] as well as Jesus College in Oxford, UK.}
\author{Corto Mascle}{LaBRI, Université de Bordeaux, France}{corto.mascle@labri.fr}{}{}
\author{Nicolas Waldburger}{IRISA, Université de Rennes, France}{nicolas.waldburger@irisa.fr}{https://orcid.org/0009-0002-7664-5828}{}
\author{Chana Weil-Kennedy}{IMDEA Software Institute, Spain}{chana.weilkennedy@imdea.org}{https://orcid.org/0000-0002-1351-8824}{This work was supported by the grant PID2022-138072OB-I00, funded by MCIN, FEDER, UE and partially supported by PRODIGY Project (TED2021-132464B-I00) funded by MCIN and the European Union NextGeneration.}

\authorrunning{S. van Bergerem, R. Guttenberg, S. Kiefer, C. Mascle, N. Waldburger, C. Weil-Kennedy} 

\Copyright{Steffen van Bergerem, Roland Guttenberg, Sandra Kiefer, Corto Mascle, Nicolas Waldburger, Chana Weil-Kennedy}



\acknowledgements{This project started at and has benefitted substantially from the research camp Autobóz 2023 in Kassel, Germany. We would like to thank the host, sponsors, and organisers of the research camp for bringing us together.}

\nolinenumbers

\EventEditors{Karl Bringmann, Martin Grohe, Gabriele Puppis, and Ola Svensson}
\EventNoEds{4}
\EventLongTitle{51st International Colloquium on Automata, Languages, and Programming (ICALP 2024)}
\EventShortTitle{ICALP 2024}
\EventAcronym{ICALP}
\EventYear{2024}
\EventDate{July 8--12, 2024}
\EventLocation{Tallinn, Estonia}
\EventLogo{}
\SeriesVolume{297}
\ArticleNo{151}

\usepackage[utf8]{inputenc}
\usepackage{xcolor}
\usepackage{hyperref}
\usepackage{stmaryrd}
\usepackage[notion, quotation, electronic]{knowledge}
\usepackage{todonotes}
\usepackage{mathtools}
\usepackage{amssymb}
\usepackage{amsmath,stackrel}
\usepackage{amsthm}
\usepackage{xspace}
\usepackage{tikz}
\usepackage{cite}
\usepackage{graphicx}
\usepackage{caption}
\usepackage{framed}
\usepackage{enumitem}

\definecolor{Green}{HTML}{45A229}
\definecolor{LightGreen}{HTML}{65C249}

\definecolor{Navy}{HTML}{2943A2}
\definecolor{LightNavy}{HTML}{4963C2}

\definecolor{LightGray}{HTML}{DDDDDD}

\definecolor{IoppudBlue}{HTML}{648FFF}
\definecolor{IoppudMagenta}{HTML}{DC267F}
\definecolor{IoppudOrange}{HTML}{FE6100}
\definecolor{IoppudPurple}{HTML}{785EF0}
\definecolor{IoppudYellow}{HTML}{FFB000}
\colorlet{ConfigBackgroundFill}{black!5}
\colorlet{ConfigBackgroundBorder}{black!20}
\colorlet{topstate}{IoppudBlue!60}
\colorlet{botstate}{IoppudOrange!60}

\usetikzlibrary{automata, shapes, positioning, arrows,arrows.meta,calc,decorations,decorations.markings, decorations.pathreplacing, math,shapes.geometric, petri,backgrounds,shadows}
\tikzset{AUT style/.style={>=angle 60,initial text= ,every edge/.append,every state/.style={minimum size=20,inner sep=2}}}
\tikzset{
  >={stealth'},-={stealth',ultra thick,scale=3}, square/.style={draw, regular polygon, regular polygon sides=4, thick, inner sep=1.5pt},
  triangle/.style={draw, regular polygon, regular polygon sides=3, thick, inner sep=1pt},
  ccircle/.style={draw, circle, inner sep=1.7pt, thick},
  sstar/.style={draw, star, inner sep=1pt, star point ratio=2, thick},
  crossi/.style={draw=white, cross out, line width=.6pt, inner sep=1.5pt},
  crosso/.style={draw, cross out, line width=2.0pt, inner sep=1.5pt},
  dropshadow/.style={drop shadow={opacity=.4, shadow xshift=.25ex, shadow yshift=-.25ex}},
  protocol state/.style={circle, draw, semithick, fill=white, dropshadow},
  big protocol state/.style={protocol state, minimum width=4em},
  medium protocol state/.style={protocol state, minimum width=3em},
  container/.style={rectangle split, rectangle split parts=3, inner sep=2em, draw, fill=ConfigBackgroundFill,
    text width=1em, line width=.15em, rounded corners=.3em},
  big box/.style={rectangle split, rectangle split parts=5, inner sep=.5em,
    draw, fill=white, text width=1.9em, line width=.1em, rounded corners=.2em},
  small box/.style={rectangle split, rectangle split parts=5, inner sep=.25em,
    draw=Datatype1, fill=white, text width=1.2em, line width=.1em, rounded corners=.1em},
  configuration/.style={ConfigBackgroundBorder, fill=ConfigBackgroundFill, rounded corners=1em},
}

\newenvironment{proofsketch}{\renewcommand{\proofname}{Proof sketch}\proof}{\endproof}

\newcommand{\nats}{\mathbb{N}}
\newcommand{\set}[1]{\{#1\}}
\newcommand{\pow}[1]{2^{#1}}

\newcommand{\size}[1]{\lvert#1\rvert}
\newcommand{\bigsize}[1]{\bigl\lvert#1\bigr\rvert}
\newcommand{\norm}[1]{||#1||}
\newcommand{\logarithm}[1]{\mathsf{log}(#1)}

\newcommand{\nset}[2]{[#1,#2]} \renewcommand{\phi}{\varphi}
\newcommand{\poly}{\mathtt{poly}}
\newcommand{\bigO}{\mathcal{O}}

\newcommand{\pspace}{\text{\sc{PSpace}}\xspace}

\newcommand{\nexpt}{\textsc{NExpTime}\xspace}
\newcommand{\conexpt}{\textsc{coNExpTime}\xspace}

\newcommand{\exps}{\textsc{ExpSpace}\xspace}
\newcommand{\nexps}{\textsc{NExpSpace}\xspace}

\newcommand{\prot}{\mathcal{P}} \newcommand{\trans}[2]{\xrightarrow{#2 #1}}
\newcommand{\adatum}{d} \newcommand{\datum}{\adatum}
\newcommand{\datvar}{\mathsf{d}} 

\newcommand{\Agentset}{\mathbb{A}}
\newcommand{\Dataset}{\mathbb{D}}
\knowledgenewrobustcmd{\dataof}{\cmdkl{\mathbf{dat}}}

\newcommand{\config}{\gamma}
\newcommand{\counting}[2]{#1^{\#}_{#2}}
\knowledgenewrobustcmd{\configset}{\cmdkl{\Gamma}}
\knowledgenewrobustcmd{\initconfigs}{\cmdkl{\Gamma_{\mathsf{init}}}}

\newcommand{\step}[3]{{\mathop{\xrightarrow{#1 #3}}}_{#2}}
\newcommand{\runto}{\xrightarrow{*}}
\newcommand{\run}{\rho}
\newcommand{\prefixrun}[2]{#1[\to #2]}
\newcommand{\suffixrun}[2]{#1[#2 \to]}
\newcommand{\trace}[2]{\mathtt{tr}^{#1}_{#2}}

\newcommand{\atrace}{{tr}}
\newcommand{\sometest}{\bowtie}
\newcommand{\undefsymb}{*}
\newcommand{\observedagents}[2]{\Agentset^{#2}_{#1, o}}

\newcommand{\leader}[1]{\mathsf{\ell}_{#1}}
\newcommand{\nonleader}[1]{\mathsf{f}_{#1}}
\newcommand{\dead}{\mathsf{dead}}

\newcommand{\poststarsymb}{\mathsf{Post}^*}
\newcommand{\prestarsymb}{\mathsf{Pre}^*}
\knowledgenewrobustcmd{\poststar}[1]{\cmdkl{\poststarsymb(#1)}}
\knowledgenewrobustcmd{\prestar}[1]{\cmdkl{\prestarsymb(#1)}}

\knowledgenewrobustcmd{\functionf}{\cmdkl{f}}
\knowledgenewrobustcmd{\functiong}{\cmdkl{g}}

\knowledgenewrobustcmd{\Output}[1]{\cmdkl{\mathsf{Out}}_{#1}}
\knowledgenewrobustcmd{\Stable}[1]{\cmdkl{\mathsf{Stable}}_{#1}}
\newcommand{\setcomplement}[1]{\overline{#1}}

\knowledgenewrobustcmd{\cubeapprox}[3]{\cmdkl{\lceil}{#2, #3}\cmdkl{\rceil}^{#1}}
\knowledgenewrobustcmd{\metaapprox}[3]{\cmdkl{\lceil}{#3}\cmdkl{\rceil}^{#1, #2}}

\knowledgenewrobustcmd{\cubeequiv}[1]{\cmdkl{\equiv}_{#1}}
\knowledgenewrobustcmd{\metaequiv}[2]{\cmdkl{\equiv}_{#1, #2}}

\newcommand{\abox}{\mathsf{b}}
\knowledgenewrobustcmd{\Boxes}[1]{\cmdkl{\mathbf{Boxes}}_{#1}}

\newcommand{\acontainer}{\mathsf{cont}}

\knowledgenewrobustcmd{\semanticsEP}[2]{\cmdkl{\llbracket} #1 \cmdkl{\rrbracket}_{#2}}
\newcommand{\setof}[1]{\semanticsEP{#1}{}}

\knowledgenewrobustcmd{\encodingsize}[1]{\cmdkl{|\langle} #1 \cmdkl{\rangle|}}
\knowledgenewrobustcmd{\lengthGRE}[1]{\cmdkl{|} #1 \cmdkl{|}}
\knowledgenewrobustcmd{\normGRE}[1]{\cmdkl{||} #1 \cmdkl{||}}

\newcommand{\bunchreach}{\mathcal{R}}
\newcommand{\shpair}{\mathcal{S}}

\DeclareMathOperator{\increment}{\mathsf{inc}}
\DeclareMathOperator{\decrement}{\mathsf{dec}}
\DeclareMathOperator{\zerotest}{\mathsf{Test}_0}
\newcommand{\Halt}{\mathsf{Halt}}
\DeclareMathOperator{\idle}{idle}
\DeclareMathOperator{\done}{done}
\DeclareMathOperator{\other}{other}
\DeclareMathOperator{\Prot}{\prot}
\DeclareMathOperator{\N}{\nats}
\newcommand{\Uniq}[0]{U}
\newcommand{\instrstate}[1]{q_{#1}}
\newcommand{\counterstate}[1]{\mathtt{#1}}
\newcommand{\countercontrol}[1]{\bar{\counterstate{#1}}}

\newcommand{\states}{Q}
\newcommand{\transitions}{\Delta}

\newcommand{\deff}{\coloneqq}
\newcommand{\setc}[2]{\ensuremath{\set{#1 \mid #2}}}
\newcommand{\bigset}[1]{\ensuremath{\bigl\{ #1 \bigr\}}}

\newcommand{\tcolors}{\mathcal{C}}
\newcommand{\tiles}{\mathcal{T}}
\newcommand{\atile}{t}

\newcommand{\twhite}{\mathsf{white}}
\newcommand{\leftcolor}[1]{\mathsf{left}(#1)}
\newcommand{\rightcolor}[1]{\mathsf{right}(#1)}
\newcommand{\topcolor}[1]{\mathsf{top}(#1)}
\newcommand{\bottomcolor}[1]{\mathsf{bottom}(#1)}
\newcommand{\tiling}{\tau}
\newcommand{\binarynumber}[1]{\overline{#1}^\mathtt{2}}

\newcommand{\dupstate}{q_{\mathsf{dup}}}
\newcommand{\cheatstate}{q_{\bot}}
\newcommand{\stateover}[1]{q_{\geq #1}}
\knowledgenewrobustcmd{\presenceformula}[1]{\cmdkl{\mathsf{Pres}(#1)}}
\knowledgenewrobustcmd{\absenceformula}[1]{\cmdkl{\mathsf{Abs}(#1)}}
\knowledgenewrobustcmd{\mandatoryformula}[1]{\cmdkl{\mathsf{Mandatory}(#1)}}
\knowledgenewrobustcmd{\encodingGRE}{\cmdkl{\mathcal{E}_{\mathsf{encoding}}}}
\newcommand{\sinkstate}{\mathsf{sink}}
\newcommand{\assignalgo}{\shortleftarrow}
\newcommand{\colorvar}{\mathtt{colour}}
\newcommand{\tilingofconfig}[1]{\tiling(#1)}
\newcommand{\statet}[1]{\mathtt{t}_{#1}}
\knowledgenewrobustcmd{\reductionGRE}{\cmdkl{\mathcal{E}}}
\knowledgenewrobustcmd{\GREinith}{\cmdkl{\mathcal{E}_{\mathsf{init}}^{(h)}}}
\knowledgenewrobustcmd{\GREinitv}{\cmdkl{\mathcal{E}_{\mathsf{init}}^{(v)}}}
\knowledgenewrobustcmd{\GREfinalh}{\cmdkl{\mathcal{E}_{\mathsf{final}}^{(h)}}}
\knowledgenewrobustcmd{\GREfinalv}{\cmdkl{\mathcal{E}_{\mathsf{final}}^{(v)}}}
\newcommand{\syncagent}{a_{\mathsf{syn}}}

\newcommand{\ie}{i.\,e.}

\newcommand{\eg}{e.\,g.}

\knowledgenewrobustcmd{\splittr}[2]{\cmdkl{\textsf{str}}_{#1}^{#2}}
\knowledgenewrobustcmd{\splitset}[1]{\cmdkl{\mathcal{S}}_{#1}}
\renewcommand{\epsilon}{\varepsilon}
\newcommand{\fbij}{\mathbf{f}}
\newcommand{\gbij}{\mathbf{g}}
\newcommand{\lbij}{\mathbf{l}}

\newcommand{\agentbunch}[1]{\mathsf{#1}}
 \definecolor{Blue Sapphire}{HTML}{002346} 
\definecolor{Gamboge}{HTML}{ee9b00}
\definecolor{Ruby Red}{HTML}{800000}

\IfKnowledgePaperModeTF{
}{
\knowledgestyle{intro notion}{color={Ruby Red}, emphasize}
	\knowledgestyle{notion}{color={Blue Sapphire}}
	\hypersetup{
		colorlinks=true,
		breaklinks=true,
		linkcolor={Blue Sapphire}, citecolor={Blue Sapphire}, filecolor={Blue Sapphire}, urlcolor={Blue Sapphire},
	}
}
\IfKnowledgeCompositionModeTF{
\knowledgeconfigure{anchor point color={Ruby Red}, anchor point shape=corner}
	\knowledgestyle{intro unknown}{color={Gamboge}, emphasize}
	\knowledgestyle{intro unknown cont}{color={Gamboge}, emphasize}
	\knowledgestyle{kl unknown}{color={Gamboge}}
	\knowledgestyle{kl unknown cont}{color={Gamboge}}
}{
}

\knowledge{notion}
| population protocol with unordered data
| Population Protocol with Unordered Data
| Population Protocols with Unordered Data
| PPUD
| population protocols with unordered data

\knowledge{notion}
| instruction agent

\knowledge{notion}
| violation
| violation detection
| Violation detection

\knowledge{notion}
| reservoir state

\knowledge{notion}
| covers
| cover
| covering

\knowledge{notion}
| immediate Observation PPUD
| IOPPUD
| Immediate Observation
| Immediate Observation PPUDs
| Immediate Observation Population Protocol with Unordered Data
| IOPPUDs
| immediate observation population protocol with unordered data
| immediate observation protocols
| immediate observation protocol
| protocol
| immediate-observation protocol
| immediate-observation protocols
| immediate observation

\knowledge{notion}
| $d$-agent
| $d$-agents
| $d'$-agents
| $d'$-agent

\knowledge{notion}
| split trace of $d$ in $\run$ at $i$
| split trace
| split traces

\knowledge{notion}
| appear
| appears
| appearing

\knowledge{notion}
| stabilise
| stabilises
| stabilising

\knowledge{notion}
| initial configurations
| initial configuration

\knowledge{notion}
| trajectory

\knowledge{notion}
| simple interval predicate
| simple interval predicates

\knowledge{notion}
| width@simple
| widths@simple

\knowledge{notion}
| height@simple
| heights@simple

\knowledge{notion}
| size@simple
| sizes@simple

\knowledge{notion}
| satisfied@simple

\knowledge{notion}
| interval predicate
| interval predicates
| Interval predicates

\knowledge{notion}
| width@pred
| widths@pred

\knowledge{notion}
| height@pred
| heights@pred

\knowledge{notion}
| size@pred
| sizes@pred

\knowledge{notion}
| norm@pred
| norms@pred

\knowledge{notion}
| satisfied@pred

\knowledge{notion}
| length@exp

\knowledge{notion}
| size@exp

\knowledge{notion}
| norm@exp

\knowledge{notion}
| active

\knowledge{notion}
| idle

\knowledge{notion}
| copies
| copy

\knowledge{notion}
| compatible

\knowledge{notion}
| norm@cubes

\knowledge{notion}
| data-explicitly
| data-explicit

\knowledge{notion}
| size@dataexplicit

\knowledge{notion}
| configuration
| configurations

\knowledge{notion}
| step
| steps

\knowledge{notion}
| run
| runs

\knowledge{notion}
| observe
| observes
| observing

\knowledge{notion}
| observed
| observations
| non-observed
| observation

\knowledge{notion}
| internally observed
| internally

\knowledge{notion}
| externally observed
| externally
| external observation
| external observations

\knowledge{notion}
| well-specified
| well-specification
| Well-specification

\knowledge{notion}
| correctly 
| correctness
| correct

\knowledge{notion}
| shadows
| shadow

\knowledge{notion}
| traces
| trace

\knowledge{notion}
| boxes
| box

\knowledge{notion}
| Containers
| containers
| container

\knowledge{notion}
| generalised reachability expressions
| generalised reachability expression
| GRE
| Generalised Reachability Expressions

\knowledge{notion}
| $2$-counter machine
| \(2\)-counter machine

\knowledge{notion}
| tiling problem

\knowledge{notion}
| tiling instance

\knowledge{notion}
| tiling

\knowledge{notion}
|colours@tile

\knowledge{notion}
| tiles
|tile

\knowledge{notion}
|horizontal verifier
| vertical verifier
| verifier@tiling

\knowledge{notion}
|main agent@tiling

\knowledge{notion}
| reader@tiling

\knowledge{notion}
|synchronisation test

\knowledge{notion}
| via

\knowledge{notion}
| fair
| fair run
| fair runs

\knowledge{notion}
| $b$-consensus
| $0$-consensus
| $1$-consensus
| consensus

\knowledge{notion}
| dissensus

\knowledge{notion}
| well-specification problem

\knowledge{notion}
| computes
| compute
| computed

\knowledge{notion}
| emptiness problem for GRE
| emptiness problems for GRE

\knowledge{notion}
| observed subtrace
| observed subtraces

\begin{document}

	\maketitle

  \begin{abstract}
    Population protocols are a well-studied model of distributed computation in which a group of anonymous finite-state agents communicates via pairwise interactions.
    Together they decide whether their initial configuration,
    \ie, the initial distribution of agents in the states, satisfies a property.
    As an extension in order to express properties of multisets over an infinite data domain, Blondin and Ladouceur (ICALP'23) introduced population protocols with unordered data (PPUD). In PPUD, each agent carries a fixed data value, and the interactions between agents depend on whether their data are equal or not.
    Blondin and Ladouceur also identified the interesting subclass of immediate observation PPUD (IOPPUD), where in every transition one of the two agents remains passive and does not move, and they characterised its expressive power.

    We study the decidability and complexity of formally verifying these protocols.
    The main verification problem for population protocols is well-specification,
    that is, checking whether the given PPUD computes some function.
    We show that well-specification is undecidable in general.
    By contrast, for IOPPUD, we exhibit a large yet natural class of problems, which includes well-specification among other classic problems, and establish that these problems are in \exps.
    We also provide a lower complexity bound, namely \conexpt-hardness.
  \end{abstract}

	\section{Introduction}
\label{sec:intro}

Population protocols (PP) model distributed computation and have received a lot of attention~\cite{AlistarhAEGR17, AlistarhG18, ElsasserR18, BlondinEJM21, EsparzaJRW21, CzernerGHE24} since their introduction in 2004 \cite{AngluinADFP04}.
In a PP, a collection of indistinguishable mobile agents with constant-size memory communicate via pairwise interactions.
When two agents meet, they exchange information about their states
and update their states accordingly.
The agents collectively compute whether their input configuration, \ie, the initial distribution of agents in each state, satisfies a certain predicate.
For a PP to compute a predicate, the protocol must be \emph{well-specified}, \ie,
for every initial configuration,
all "fair" runs starting in this configuration must converge to the same
answer.
It was shown that PP compute exactly the predicates of Presburger arithmetic~\cite{AngluinAER07}. Moreover, well-specification is known to be decidable but as hard as the reachability problem for Petri nets \cite{EsparzaGLM17}.
Note that deciding well-specification is a problem that concerns \emph{parameterised verification} in the sense of \cite{EsparzaKeepingSafe, 2015Bloem}, \ie, one must decide that something holds with respect to every value of the parameter. Here the parameter is the number of agents that are present -- the PP must converge to one answer for every initial configuration, no matter the number of agents.

\emph{Population protocols with unordered data} (PPUD) were introduced by Blondin and Ladouceur
as a means to compute predicates over arbitrarily large domains~\cite{BL23}.
In this setting, each agent holds a read-only datum from an infinite set $\Dataset$.
When interacting, agents may check (dis)equality of their data.
While PP can compute properties like
``there are more than 5 agents in state $q_1$'',
PPUD can express, \eg, ``there are more than 2 data with 5 agents each in state $q_1$''. In~\cite{BL23}, the authors construct a PPUD computing the absolute majority predicate, \ie, whether a datum is held by more than half of the agents.
They also characterise the expressive power of \emph{immediate observation PPUD} (IOPPUD), a subclass of interest in which interactions are restricted to observations. That is, in every interaction, one of the two agents is passive and does not change its state.
The decidability and complexity of the main verification question for PPUD, namely well-specification, is left open in Blondin's and Ladouceur's article \cite{BL23}. It is the subject of this paper.

\subparagraph*{Contributions}
We start by showing that well-specification is \emph{undecidable} for PPUD.
This follows from a reduction from $2$-counter machines; in fact, the presence of data  allows us to encode zero-tests.
Contrasting this, we show that deciding well-specification is in \exps for IOPPUD.
To this end, we define \emph{generalised reachability expressions} (GRE)
and establish that, for IOPPUD, deciding whether the set of configurations
that satisfy a given GRE is empty is in \exps.
This decidability result is powerful; indeed, this emptiness problem subsumes  classic  verification problems like reachability and coverability, as well as parameterised verification problems such as well-specification and correctness, where the latter asks whether a given protocol computes a given predicate.
Lastly, we exhibit a \conexpt\ lower bound for deciding emptiness of GRE for IOPPUD.

\subparagraph*{Related work}
For a recent survey of the research on verification of PP (without data), see \cite{Esparza21}. In particular, the well-specification problem for PP is known to be decidable, but as hard as Petri net reachability \cite{EsparzaGLM17} and therefore Ackermann-complete \cite{LerouxS19,Leroux21,CzerwinskiO21}.
In their seminal paper on the computational power of PP, Angluin, Aspnes, Eisenstat, and Ruppert also introduced five subclasses of PP that model one-way communication \cite{AngluinAER07}.
One of these is immediate observation population protocols (IOPP), which correspond to IOPPUD without data. The complexity of well-specification for all five subclasses is determined in \cite{EsparzaJRW21}. In particular, the paper shows that well-specification for IOPP is \pspace-complete. IOPP were modelled by immediate observation Petri nets, where classic parameterised problems can be decided in polynomial space.
The notion of generalised reachability expression was first phrased in this setting,
and one of the consequences is that
the emptiness problem of GRE for IOPP is \pspace-complete \cite{Weil-Kennedy23}.
Our result shows that adding data to the model as in \cite{BlondinEJM21} (and extending GRE naturally) pushes the emptiness problem between \conexpt\ and \exps.

While the introduction of data in the PP model happened recently \cite{BL23}, a similar approach has been studied in the related model of Petri nets, under the name of \emph{data nets}. In this setting, the classic problem of coverability (or control-state reachability) is decidable but non-primitive recursive \cite{LazicNORW07} and in fact $\mathbf{F}_{\omega^\omega}$-complete \cite{Rosa-VelardoF11}. While PPUD can be encoded into data nets, our results show that the problems that we study cannot be reduced to coverability.
Another related model is formed by broadcast networks of register automata (BNRA)\cite{GMW24}, an extension of reconfigurable broadcast networks (RBN) with data.
RBN subsume IOPP \cite{BW21}, and consequently BNRA subsume IOPPUD.
However,
the complexity of coverability in BNRA is known to be $\mathbf{F}_{\omega^\omega}$-complete, hence non-primitive recursive, and more complex problems quickly become undecidable \cite{GMW24}.
These hardness results contrast with the \exps membership.

\subparagraph*{Organisation} In Section~\ref{sec:definitions}, we introduce the models of PPUD and IOPPUD, the notion of GRE, and we state our main results.
We prove undecidability of well-specification for PPUD in \cref{sec:undecidability}. The next sections are dedicated to the study of IOPPUD. In \cref{sec:bounds-observed-agents}, we establish bounds on the number of observed agents. In Section~\ref{sec:equivalence-relation}, we introduce the technical notions of boxes and containers and use the bounds from the previous section to translate GRE into containers.
We present the complexity bounds for emptiness of GRE in Section~\ref{sec:complexity}.

	\section{Population Protocols and Main Results}
\label{sec:definitions}

We use the notation \([m,n] \coloneqq \set{\ell \in \nats \mid m \leq \ell \leq n}\) for \(m, n \in \nats\) and $[m, +\infty)\coloneqq \set{\ell \in \nats \mid m \leq \ell}$ .

\subsection{Population Protocols with Unordered Data}

\AP We fix an infinite \emph{data} domain $\Dataset$, an infinite set of \emph{agents} $\Agentset$, and a function $\intro*\dataof \colon \Agentset \to \Dataset$ such that $\dataof^{-1}(d)$ is infinite for all $d \in \Dataset$. For $d \in \Dataset$, a ""$d$-agent"" is an agent $a \in \Agentset$ with $\dataof(a) = d$.

\begin{definition}
	A ""population protocol with unordered data"" ("PPUD") is a tuple $(Q, \Delta, I, O)$ where $Q$ is a finite set of \emph{states},
  $\Delta \subseteq Q^2 \times \set{=, \neq} \times Q^2$ the set of \emph{transitions},
	$I \subseteq Q$ the set of \emph{initial states}, and
	$O \colon Q \to \set{\top, \bot}$ the \emph{output function}.
\end{definition}
\AP The size of a "PPUD" $\prot$, written $\size{\prot}$, is its number of states.
We fix a "PPUD" $(Q, \Delta, I, O)$.

\AP A ""configuration"" is a function $\config \colon \Agentset \to Q \cup \set{\undefsymb}$ such that $\config(a) \in Q$ only holds for finitely many agents $a \in \Agentset$ (the agents ""appearing"" in $\config$).
We denote by $\intro*\configset$ the set of all "configurations", and by $\intro*\initconfigs \coloneqq \set{\config \in \configset \mid \forall a \in \Agentset, \, \config(a) \notin Q \setminus I}$ the set of ""initial configurations"".
Given $\config \in \configset$ and $\adatum \in \Dataset$, we let $\counting{\config}{\adatum} \colon Q \to \nats$ be the function that maps each state $q$ to the number of "$d$-agents" in $q$ in $\config$.

\AP Given $\config, \config' \in \configset$, we write $\config \to \config'$,
and call it a ""step"" from $\config$ to $\config'$,
when there are states $q_1, q_2, q_3, q_4 \in Q$ and two distinct agents $a_1, a_2 \in \Agentset$
such that $((q_1,q_2),~\sometest,~(q_3,q_4)) \in \transitions$,
$\config(a_1) = q_1$, $\config(a_2) = q_2$, $\config'(a_1) = q_3$, $\config'(a_2) = q_4$,
$\config(a) = \config'(a)$ for all $a \in \Agentset \setminus \set{a_1, a_2}$,
and, additionally, if $\sometest$ is an equality (resp.\ disequality),
then $\dataof(a_1) = \dataof(a_2)$ (resp. $\dataof(a_1) \neq \dataof(a_2)$).
\AP A ""run"" $\run$ is a (finite or infinite) sequence of consecutive "steps" $\run \colon \config_1 \to \config_2\to \config_3 \to \dots$.
We write $\run \colon \config \runto \config'$ to denote that $\run$ is a finite "run" from $\config$ to $\config'$, and simply $\config \runto \config'$ to denote the existence of such a "run".
\AP For every $\config \in \Gamma$, let $\intro*\poststar{\config} \coloneqq \set{\config' \in \configset \mid \config \runto \config'}$ and $\intro*\prestar{\config} \coloneqq \set{\config' \in \configset \mid \config' \runto \config}$.
A "run" $\run$ ""covers"" a state $q \in Q$ if there is a configuration $\config$ in $\run$ such that $\config(a)=q$ for some agent $a$.

\AP In accordance with \cite{AngluinADFP04} and \cite{BL23}, we consider a "run" $\config_1 \to \config_2 \to \dots$ ""fair"" if it is infinite\footnote{One often considers that a finite "run" $\config_f \runto \config_\ell$ is fair when there is no $\config$ such that $\config_\ell \to \config$. In the following, we rule out this possibility by implicitly assuming that, for all $q_1, q_2 \in Q$ and $\sometest{} \in \set{=, \neq}$, it holds that $((q_1, q_2), \sometest, (q_1,q_2)) \in \Delta$, and ignoring the trivial cases of "runs" with at most one agent.}
and for every "configuration" $\config$ with $\bigsize{\setc{ i \in \nats }{ \config_i \runto \config}} = \infty$, it holds that $\bigsize{\setc{ i \in \nats }{ \config_i = \config}} = \infty$. That is, every infinitely often reachable configuration also occurs infinitely often along the run.
 For $b \in \set{\top, \bot}$,
a ""$b$-consensus"" is a configuration $\config$ in which, for all agents $a \in \Agentset$ "appearing" in $\config$, it holds that \(O(\config(a)) = b\).
A "fair run" $\run \colon \config_1 \to \config_2 \to \cdots$ ""stabilises"" to $b \in \set{\top, \bot}$ if there is an $n \in \nats$ such that for every $i \geq n$, $\config_i$ is a "$b$-consensus".
A protocol is ""well-specified"" if, for every "initial configuration" $\config_0 \in \initconfigs$, there is $b \in \set{\top, \bot}$ such that all "fair runs" starting in $\config_0$ "stabilise" to $b$. The ""well-specification problem"" for "PPUD" asks, given a "PPUD" $\prot$, whether $\prot$ is "well-specified".
\AP Given a "PPUD" $\prot$ and a function $\Pi \colon \initconfigs \to \set{\top,\bot}$, $\prot$ ""computes"" $\Pi$ if, for every \(\config_0 \in \initconfigs\),
every "fair run" of $\prot$ starting in $\config_0$ "stabilises" to $\Pi(\config_0)$.

\begin{example}
\label{ppud-example}
Consider the following "PPUD" $\prot$. Its set of states is $Q = \set{\leader{0}, \leader{1}, \nonleader{0}, \nonleader{1},\dead}$, with $I = \set{\leader{1}}$, $O(\leader{0}) = O(\nonleader{0}) =\top$, $O(\leader{1}) = O(\nonleader{1}) = O(\dead) = \bot$ and its transitions are:
\begin{align*}
& \forall b,b' \in \set{0,1}, \ (\leader{b},\leader{b'}) \mapsto (\leader{b \oplus b'},\nonleader{b \oplus b'}) \qquad & \forall b,b' \in \set{0,1}, \ (\leader{b},\nonleader{b'}) \mapsto (\leader{b},\nonleader{b}) \\
& \forall q,q' \in Q, \ (q,q') \mapsto_= (\dead, \dead) & \forall q \in Q, \ (q, \dead) \mapsto (\dead, \dead)
 \end{align*}
where $\mapsto_{=}$ denotes that the data of the agents must be equal, $\mapsto$ without subscript means no condition on data (or equivalently, the transition exists both for equality and disequality), and $\oplus$ denotes the XOR operator. $\prot$ is "well-specified" and "computes" the function $\Pi$ that is equal to $\top$ whenever there is an even number of appearing data and they all have exactly one corresponding agent. To see this, if there are two agents of equal datum, then all "fair runs" eventually have all agents on $\dead$ and "stabilise" to $\bot$. Otherwise, there will eventually be a single agent in $\set{\leader{0}, \leader{1}}$, and it will be on $\leader{b}$ if and only if the number of agents has parity $b$, in which case all other agents will eventually go to $\nonleader{b}$ and the run "stabilises" to $\bot$ if $b = 1$ (odd number of agents) and to $\top$ if $b=0$ (even number of agents).

A more interesting but also more complex example is the majority protocol described in \cite[Section 3]{BL23}; it computes whether a datum has the absolute majority, \ie, strictly more agents than all other data combined.
\end{example}

"Well-specification" is the fundamental verification problem for population protocols. However, as we will see in \cref{sec:undecidability}, this problem is  undecidable for PPUD.

\begin{restatable}{theorem}{PropositionVerificationUndecidable}
	\label{undec-well-spec}
The "well-specification problem" for "PPUD" is undecidable.
\end{restatable}

This motivates the study of the restricted class of \emph{immediate observation} "PPUD".

\subsection{Immediate Observation Protocols}
Immediate observation protocols \cite{AngluinAER07} are a restriction of population protocols where, when two agents interact, one of the two agents does not change its state. The restriction of the model with data to immediate observation was first considered in \cite{BL23}.

\begin{definition}
An ""immediate observation population protocol with unordered data"" (or \reintro{IOPPUD}) is a "PPUD" $\prot = (Q, \Delta, I, O)$ where every transition $\delta \in \Delta$ is of the form $(q_1, q_2, \sometest, q_1, q_3)$, with $q_1, q_2, q_3 \in Q$ and $\sometest{} \in \set{=, \neq}$, \ie, the first agent does not change its state.
\end{definition}

\AP For "IOPPUD", we denote a transition $(q_1, q_2, \sometest, q_1, q_3)$ by $q_2 \trans{q_1}{\sometest} q_3$. If we have a "step" $\config \to \config'$ with transition $q_2 \trans{q_1}{\sometest} q_3$ that involves agents $a,a_o \in \Agentset$ where $a$ is the agent moving from $q_2$ to $q_3$ and $a_o$ is the agent in $q_1$, we denote it by $\config \step{\sometest}{a}{a_o} \config'$.
We say that agent $a$ ""observes"" agent $a_o$, and call $a_o$ the \emph{observed} agent. Intuitively, $a$ ``observes'' $a_o$ and reacts, whereas  $a_o$ may not even know it has been observed.

\begin{example}
\label{example-ioppud}
Consider the following "IOPPUD" $\prot = (Q,\Delta, I, O)$, with $Q \deff \set{q_0,q_1,q_2,q_3}$, $I \deff \set{q_0, q_1}$, $O(q_3)=\top$, $O(q) = \bot$ for all $q \ne q_3$, and transitions in $\Delta$ as follows:
\begin{gather*}
q_0 \trans{q_1}{=} q_2 \quad  q_1 \trans{q_0}{=} q_2 \quad q_2 \trans{q_2}{\ne} q_3 \qquad  \forall q \in \set{q_1,q_2}, \forall \sometest{} \in \set{=, \ne}, q \trans{q_3}{\sometest} q_3
\end{gather*}
 This protocol is "well-specified": from $\config_0 \in \initconfigs$, all "fair runs" stabilise to $\top$ if two data have agents on both $q_0$ and $q_1$, and all "fair runs" stabilise to $\bot$ otherwise. Indeed, if there is a datum with agents on both $q_0$ and $q_1$, by fairness eventually  an agent with this datum is sent to $q_2$; if there are two such data, then eventually some agent covers $q_3$, and then all agents are sent to $q_3$ and the run "stabilises" to $\top$. Conversely, if it is not the case, then $q_3$ cannot be covered and all "fair runs" "stabilise" to $\bot$.
\end{example}

\AP Let $\run \colon \config_{start} \runto \config_{end}$ be a "run".
Agent $a_o$ is ""internally observed"" (resp. ""externally observed"") in $\run$ if $\run$ contains a step of the form $\config_1 \step{=}{a}{a_o} \config_2$ (resp. $\config_1 \step{\neq}{a}{a_o} \config_2$); it is ""observed"" if one of the two cases holds.
Similarly, a datum $\adatum$ is \emph{observed} in $\run$ if an agent $a$ with $\dataof(a) =\adatum$ is "observed" in $\run$; we define similarly a datum being internally or externally observed. 

While the set of functions that can be "computed" by "PPUD" remains an open question, it is known that "IOPPUD" exactly "compute" "interval predicates"~\cite{BL23}, defined as follows.
 Let $S$ be a finite set.
A ""simple interval predicate"" over $S$ is a formula $\psi$ of the form $\dot\exists \datvar_1, \ldots, \datvar_m,\, \bigwedge_{q \in S} \bigwedge_{j=1}^m \#(q,\datvar_j) \in [A_{q,j},B_{q,j}]$  where, for all $q \in S$ and $j \in \nset{1}{m}$, we have $A_{q,j} \in \nats$ and $B_{q,j} \in \nats \cup \set{+\infty}$. The dotted quantifiers quantify over \emph{pairwise distinct} data. Formally, given a protocol $\prot$ with set of states $Q$ such that $S \subseteq Q$ and
given $\config \in \configset$, the predicate $\psi$ is ""satisfied@@simple"" by $\config$
if there exist pairwise distinct data $d_1, \ldots, d_m \in \Dataset$ such that for all $q\in S$ and $j \in \nset{1}{m}$, it holds that $\counting{\config}{d_j}(q) \in [A_{q,j},B_{q,j}]$ (resp.\ $\counting{\config}{d_j}(q) \in [A_{q,j},B_{q,j})$ in the case that $B_{q,j} = +\infty$). An ""interval predicate"" over $S$ is a Boolean combination $\phi$ of "simple interval predicates" over $S$; we define that $\phi$ is ""satisfied@@pred"" by a "configuration" $\config$ if the "simple interval predicates" "satisfied@@simple" by $\config$ satisfy the Boolean combination.

\begin{theorem}[\cite{BL23}, Theorem 18 and Corollary 29]
\label{ioppud-compute-predicates}
	Given a finite set $I$, the functions "computed" by "IOPPUD" with set of initial states\footnote{This does not limit the number of states of said protocols, as their set of states $Q$ may be larger than $I$.} $I$ are exactly the "interval predicates" over $I$.
\end{theorem}

\begin{example} The protocol described in \cref{ppud-example} and the majority protocol of \cite[Section 3]{BL23} cannot be turned into "immediate observation protocols", as they "compute" functions that cannot be expressed as "interval predicates".
The "immediate observation protocol" from \cref{example-ioppud} "computes" the following "interval predicate", which is actually a "simple interval predicate":
 \[\dot{\exists} \datvar_1, \datvar_2, \, (\#(q_0,\datvar_1) \geq 1) \land (\#(q_1,\datvar_1) \geq 1) \land (\#(q_0,\datvar_2) \geq 1) \land (\#(q_1,\datvar_2) \geq 1).\]
\end{example}

\AP Given a "simple interval predicate" $\psi = \dot\exists \datvar_1, \ldots, \datvar_m,\, \bigwedge_{q \in I} \bigwedge_{j=1}^m \#(q,\datvar_j) \in [A_{q,j},B_{q,j}]$, we define its ""width@@simple"" as $m$, its ""height@@simple"" $h$ as the maximum of all finite $A_{q,j}$ and $B_{q,j}$, and its ""size@@simple"" as $\size{I} \cdot m \cdot \logarithm{h}$.  We also define the ""width@@pred"" (resp.\ ""height@@pred"") of an "interval predicate" as the maximum of the "widths@@simple" (resp.\ "heights@@simple") of its "simple interval predicates",
and its ""size@@pred"", measuring the space taken by its encoding, as the sum of their "sizes@@simple" plus its number of Boolean operators.

\begin{remark}
\label{rem:no-input-alpabet}
In \cite{BL23}, predicates refer to an input alphabet $\Sigma$, which is converted into initial states using an input mapping. For convenience, we have not included the input alphabet in our model, which is why we arbitrarily fix a set of initial states $I$ in \cref{ioppud-compute-predicates}.
\end{remark}

\subsection{Generalised Reachability Expressions}

We define a general class of specifications, called "generalised reachability expressions", which are formulas constructed using "interval predicates" as atoms and using union, complement, $\poststarsymb$, and $\prestarsymb$ as operators. This concept is inspired by~\cite[Section 2.4]{Weil-Kennedy23}, although our choice of atoms is more general and adapted to the data setting.

\begin{definition}
\label{def:GRE}
	Let $\prot = (Q, \Delta, I, O)$ be a "protocol". \\
	""Generalised Reachability Expressions"" ("GRE") over $\prot$ are produced by the grammar
	\[E ::= \phi \mid E \cup E \mid \setcomplement{E} \mid \poststar{E} \mid \prestar{E},\]
	where $\phi$ ranges over "interval predicates" over $Q$.

	Given a "GRE" $E$, we define the set of "configurations" defined by $E$, denoted $\intro*\semanticsEP{E}{\prot}$, as the set containing all "configurations" of $\prot$ that satisfy the formula,
	where the predicates are interpreted as above and the other operators are interpreted naturally (the overline denotes set complementation).
	This set is denoted $\setof{E}$ when $\prot$ is clear from context.
\end{definition}

\AP The ""length@@exp"" $\intro*\lengthGRE{E}$ of a "GRE" $E$ is its number of operators.
Letting $\phi_1, \ldots, \phi_k$ be the "interval predicates" used as atoms in $E$, the ""norm@@exp"" $\intro*\normGRE{E}$ of $E$ is the maximum of the "heights@@pred" and "widths@@pred" of the $\phi_i$.
Its ""size@@exp"" is the sum of the "sizes@@pred" of the $\phi_i$ plus $\lengthGRE{E}$. The ""emptiness problem for GRE"" asks, given as input a protocol $\prot$ and a "GRE" $E$ over $\prot$, whether $\semanticsEP{E}{\prot} = \emptyset$.
We will show in Section \ref{sec:complexity} that, for "IOPPUD", this problem is decidable.

\begin{restatable}{theorem}{TheoremEXPSPACE}
\label{thm:main-ioppud}
	The "emptiness problem for GRE" over "IOPPUD" is in \exps.
\end{restatable}

We now argue that this decidability result is powerful, as it implies decidability of many classic problems on "IOPPUD". We start with "well-specification". We use the notation $\dot \forall \datvar, \varphi$ as a short form for $\neg \dot \exists \datvar, \neg \varphi$.
Given a "PPUD" $\prot$ and $b \in \{\top,\bot\}$, let $\intro*\Output{b} \deff \dot \forall \datvar, \bigwedge_{q \notin O^{-1}(\set{b})} \#(q,\datvar) = 0$ be the "GRE" for "$b$-consensus" configurations; moreover, let $\intro*\Stable{b} \deff \setcomplement{\prestar{\setcomplement{\Output{b}}}}$ be the "GRE" for \emph{stable} "$b$-consensus", \ie, "configurations" from which all "runs" lead to a "$b$-consensus".

\begin{proposition} \label{well-spec-gre}
Let $\prot$ be a "PPUD", $E_{\textsf{ws}} \deff \initconfigs \cap \prestar{\setcomplement{\prestar{{\Stable{\top}}}}} \cap \prestar{\setcomplement{\prestar{{\Stable{\bot}}}}}$. \\
$\prot$ is "well-specified" if and only if $\semanticsEP{E_{\textsf{ws}}}{\prot} = \emptyset$.
\end{proposition}
\begin{proof}
First, $\initconfigs = \setof{\dot \forall \datvar, \, \bigwedge_{q \notin I} \#(q,d) =0}$ and $E_{\textsf{ws}}$ is indeed a "GRE" over $\prot$.
For every $\config \in \configset$, $\poststar{\config}$ is finite as all configurations reachable from $\config$ have the same number of agents.
Therefore, a "fair run" $\run$ that visits $\prestar{\mathcal{S}}$ infinitely often for $\mathcal{S} \subseteq \configset$ must visit $\mathcal{S}$ infinitely often.
Let $\config_0 \in \initconfigs$ and $b \in \set{\top, \bot}$;
it suffices to prove that there is a "fair run" from $\config_0$ that does \emph{not} "stabilise" to $b$ if and only if $\config_0 \in \setof{\prestar{\setcomplement{\prestar{{\Stable{b}}}}}}$. If, from $\config_0$, one can reach $\config \notin \setof{\prestar{\Stable{b}}}$, then one can build a "fair run" from $\config_0$ that first goes to $\config$, and then forever performs arbitrary steps in a fair way; since $\config \notin \setof{\prestar{\Stable{b}}}$, it will stay in $\setof{\prestar{\setcomplement{\Output{b}}}}$, so by fairness it visits $\setof{\setcomplement{\Output{b}}}$ infinitely often and does not "stabilise" to $b$.
Conversely, if there is a "fair run" that does not "stabilise" to $b$, then it never visits $\setof{\Stable{b}}$, hence by fairness it eventually stops visiting $\setof{\prestar{\Stable{b}}}$; this proves that it visits a configuration $\config \in \setof{\setcomplement{\prestar{\Stable{b}}}}$, and $\config$ is reachable from $\config_0$ hence $\config_0 \in \setof{\prestar{\setcomplement{\prestar{{\Stable{b}}}}}}$.
\end{proof}

Many other problems can be expressed as "emptiness problems for GRE"; we list a few.
\begin{itemize}
\item
The \emph{correctness} problem for "IOPPUD" asks,
given an "IOPPUD" $\prot =(Q,I,O,\Delta)$ and an "interval predicate" $\phi$ over $I$, whether $\prot$ "computes" $\phi$.
This can be equivalently phrased as $\setof{\phi^{-1}(b) \cap {\prestar{\setcomplement{\prestar{\Stable{b}}}}}} = \emptyset$ for all $b \in \set{\top, \bot}$, where $\phi^{-1}(b)$ is the set of "initial configurations" that $\phi$ maps to $b$. 
Note that the previous expression is a "GRE" because, in \cref{def:GRE}, we chose as atoms  "interval predicates" such as $\phi$.
\item
The \emph{set-reachability} problem (called cube-reachability in \cite{BGW22}) asks, given two sets of configurations $\mathcal{S}_1,\mathcal{S}_2$, whether $\mathcal{S}_2$ is reachable from $\mathcal{S}_1$; this typically expresses safety problems where $\mathcal{S}_2$ represents ``bad configurations'' that must not be reached.
If $\mathcal{S}_1 = \setof{E_1}$ and $\mathcal{S}_2 = \setof{E_2}$, then this amounts to checking whether $\setof{E_1 \cap \prestar{E_2}}$ is empty.
\item
The \emph{home-space problem} asks, given a protocol $\prot$ and a set of configurations $\mathcal{H}$, whether $\mathcal{H}$ can be reached from every configuration reachable from an initial configuration.
If $\mathcal{H}$ can be expressed as a "GRE" $E$, then it suffices to check whether $\setof{\poststar{\initconfigs}} \subseteq \setof{\prestar{E}}$.
This problem has been studied in Petri nets \cite{JL23}, but also in probabilistic settings, for example in \cite{BouyerMRSS16} for asynchronous shared-memory systems; indeed, in systems with uniform probabilistic schedulers where $\poststar{\config_0}$ is finite for every initial configuration $\config_0$, this problem is equivalent to asking whether the probability of reaching $\mathcal{H}$ is equal to $1$.
\end{itemize}
Theorem \ref{thm:main-ioppud} entails that, for "IOPPUD", all these problems are decidable and in \exps.

\section{Undecidability of Verification of Population Protocols with Unordered Data}
\label{sec:undecidability}

In this section, we establish that the most fundamental verification problem for "PPUD", \ie, the "well-specification problem", is already undecidable.

\PropositionVerificationUndecidable*

We proceed by reduction from the halting problem for $2$-counter machines with zero-tests, a famously undecidable problem \cite{Minsky67}.
Here, we give a proof sketch.
The detailed reduction can be found in \cref{app:undec}.

We fix a "$2$-counter machine" and build a protocol $\prot$ which is \emph{not} well-specified if and only if the counter machine halts. 
A ""$2$-counter machine"" performs increments, decrements and zero-tests on two counters.
The main difficulty are the zero-tests. Let us first recall how increments and decrements are simulated in many prior undecidability results for population protocols and Petri nets \cite{Jancar95, Rosa-VelardoF11}. The protocol has a control part \(Q_{CM} \deff \set{q_i \mid i \in \nset{1}{n}}\) where a single ""instruction agent"" evolves; this part has one state per instruction of the machine. Increments and decrements are simulated as follows: The instruction agent interacts with states of \(\set{R} \cup \set{\counterstate{x},\counterstate{y}}\), where \(R\) is a ""reservoir state"" and \(\counterstate{x}\) and \(\counterstate{y}\) are states in which the number of agents represents the value of the counters. For example an increment on counter \(x\) moves one agent from the reservoir \(R\) to \(\counterstate{x}\) and advances the instruction agent to the next instruction. The reservoir is hence implicitly assumed to start with arbitrarily many agents. 

The main difficulty is that one does not want to take the \(=0\) branch of a zero-test when the value of the counter is non-zero.
Actually, similar to \cite{Jancar95, Rosa-VelardoF11}, we will not prevent the existence of such runs.
Instead, our protocol will have ``violating'' runs which take the wrong branch of a zero-test, but our well-specification check will  consider only \emph{violation-free} runs. The correctness of the reduction is then established in two steps: The CM halts if and only if some \emph{violation-free} run to the halting state exists, and this is true if and only if our protocol is not well-specified. 
We establish the connection between non-well-specification and the existence of a \emph{violation-free} run in our protocol.

In the first place, we guarantee that every initial configuration of our protocol has a fair run "stabilising" to $\bot$, so that $\prot$ is not "well-specified" if and only if there exists a fair run which does not "stabilise" to \(\bot\)\footnote{This can be done with the addition of a fresh state that is the only initial state and that has internal transitions to all former initial states and an internal transition to a sink state that has output $\bot$.}. Second, we introduce ""violation detection"", a mechanism which guarantees that \emph{fair} runs which contain a violation stabilise to \(\bot\), hence preserving well-specification. To do so, we add a sink state \(\cheatstate\), which has output $\bot$ and is attracting, \ie, all other states have a transition to $\cheatstate$ available when observing that $\cheatstate$ is non-empty.
"Violation detection" then entails adding transitions into $\cheatstate$ that will be available infinitely often if the run (or its initial configuration) contained a violation.
By fairness, any run containing a violation will then eventually put an agent into \(\cheatstate\),
and hence, because \(\cheatstate\) is attracting, the run ends in a deadlock with all agents in \(\cheatstate\).
In particular, any fair run containing a violation will output \(\bot\) as claimed.
There are two types of violation detection.

First, we want to only mark those runs as violation-free that start in initial configurations where \(\Uniq \in \states\) has at most one agent of each datum. 
To do so, we make agents remember whether their initial state was \(\Uniq\) or not (by encoding it into the state space),
and, from every state, we add a transition to $\cheatstate$ such that this transition is enabled when an agent who started in $\Uniq$ observes another agent of same datum that also started in \(\Uniq\). 

Second, we want to detect violations which consist in falsely simulating a zero-test, as discussed above.
Here our technique shares some similarities with \cite{Rosa-VelardoF11}. Let $c \in \set{x,y}$ be a counter; for every zero-test of the counter machine, we add two types of transitions to the protocol. The first type simulates the $c \ne 0$ branch and can be taken by the "instruction agent" upon interacting with some agent on state $\counterstate{c}$; by contrast, the $c=0$ branch can always be taken. However, if it was taken with \(c \neq 0\), then  "violation detection" will eventually detect this. For this mechanism, we introduce a counter control state \(\countercontrol{c} \in \states\). At any point in time, $\countercontrol{c}$ contains one agent, similar to the instruction agent. The crux of our violation detection is that only agents which share the datum with the agent in $\countercontrol{c}$ will be allowed to move in and out of state $\counterstate{c}$, as illustrated in \cref{fig:decrement}.

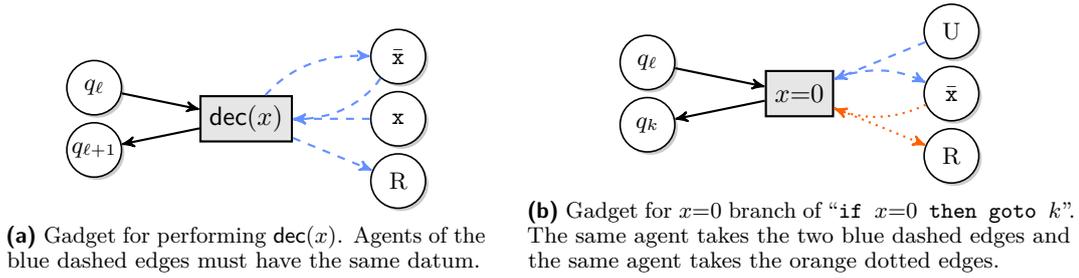
\begin{figure}
\centering
\begin{subfigure}{.45\textwidth}
  \centering
  \begin{tikzpicture}[-triangle 60,auto, thick, yscale =0.75]

  \tikzstyle{place}=[protocol state, minimum size=7.2mm, font = {\small}, inner sep = 0pt]
  \tikzstyle{transition}=[rectangle, thick, draw, fill=black!10, minimum size=6mm]

  \newcommand*{\distancesubx}{2cm}
  \newcommand*{\distancesuby}{1.1cm}
  \newcommand*{\distancelabel}{0.07cm}

  \node[place] (A1) at (0,1.5*\distancesuby) {\(\instrstate{\ell}\)};
  \node[place] (A0) at (0,0.5* \distancesuby) {\(\instrstate{\ell+1}\)};
  \node[transition] (B1) at (\distancesubx,\distancesuby) {\(\decrement(x)\)};
  \node[place] (C2) at (2*\distancesubx,2*\distancesuby) {\(\countercontrol{x}\)};
  \node[place] (C1) at (2*\distancesubx,\distancesuby) {\(\counterstate{x}\)};
  \node[place] (C0) at (2*\distancesubx,0) {\(\textup{R}\)};
  \path[->] (A1) edge[] (B1)
            (B1) edge[] (A0)
            (C2) edge[IoppudBlue, bend left, dashed] (B1)
            (B1) edge[IoppudBlue, bend left, dashed] (C2)
            (C1) edge[IoppudBlue, dashed] (B1)
            (B1) edge[IoppudBlue, dashed] (C0);
\end{tikzpicture}
   \caption{Gadget for performing \(\decrement(x)\). Agents of the blue dashed edges must have the same datum.}
  \label{fig:decrement}
\end{subfigure}
\hfill
\begin{subfigure}{.51\textwidth}
  \centering
  \begin{tikzpicture}[-triangle 60,auto, thick, yscale =0.75]

  \tikzstyle{place}=[protocol state, minimum size=7.2mm, font = {\small}, inner sep = 0pt]
  \tikzstyle{transition}=[rectangle, thick, draw, fill=black!10, minimum size=6mm]

  \newcommand*{\distancesubx}{2cm}
  \newcommand*{\distancesuby}{1.1cm}
  \newcommand*{\distancelabel}{0.07cm}

  \node[place] (A1) at (0,1.5*\distancesuby) {\(\instrstate{\ell}\)};
  \node[place] (A0) at (0,0.5* \distancesuby) {\(\instrstate{k}\)};
  \node[transition] (B1) at (\distancesubx,\distancesuby) {\(x{=}0\)};
  \node[place] (C2) at (2*\distancesubx,2*\distancesuby) {\(\textup{\Uniq}\)};
  \node[place] (C1) at (2*\distancesubx,\distancesuby) {\(\countercontrol{x}\)};
  \node[place] (C0) at (2*\distancesubx,0) {\(\textup{R}\)};
  \path[->] (A1) edge[] (B1)
            (B1) edge[] (A0)
            (C2) edge[IoppudBlue, dashed] (B1)
            (B1) edge[IoppudBlue, bend left, dashed] (C1)
            (C1) edge[IoppudOrange, bend left, dotted] (B1)
            (B1) edge[IoppudOrange, dotted] (C0);
\end{tikzpicture}
   \caption{Gadget for \(x{=}0\) branch of ``\(\texttt{if } x{=}0 \texttt{ then goto } k\)''. The same agent takes the two blue dashed edges and the same agent takes the orange dotted edges.}
  \label{fig:zerotest}
\end{subfigure}
\caption{For simplicity, we use Petri net notation: circles are states, rectangles are Petri net transitions.
To encode this into our protocols, we split each transition into pairwise interactions.}
\label{fig:countermachine}
\end{figure}

The \(=0\) branch of a zero-test is depicted in \cref{fig:zerotest}. It replaces the agent on $\countercontrol{c}$ with an agent with fresh datum from state \(\Uniq\).
Thus, when the $c=0$ branch is taken, any remaining agent in $\counterstate{c}$ is stuck in $\counterstate{c}$ as it will never again share datum with the agent in $\countercontrol{c}$. "Violation detection" then sends an agent in $\counterstate{c}$ to $\cheatstate$ upon observing an agent in $\countercontrol{c}$ with different datum.

Now that we have violation detection in place, it only remains to explain the connection to halting. The halting instruction $q_n$ in \(Q_{CM}\) is the only state with output \(\top\). Hence, any run not outputting \(\bot\) must contain an agent in the halting instruction at some point, and be violation-free by the above. That is, the counter machine reached the halting state without violations. Conversely, if the machine halts, one can build a finite run that puts an agent into the halting state without any violation occurring. The corresponding configuration is then a deadlock, and hence the extension to an infinite run (by staying there forever) is a fair run not outputting \(\bot\). This proves that "well-specification" is undecidable for "PPUD", which motivates restricting ourselves to "immediate observation" "PPUD".

	\section{An Analysis of Immediate Observation Protocols with Data}
\label{sec:bounds-observed-agents}

To obtain our complexity bounds on the "emptiness problem for GRE",
we first show some transformations on runs that allow us to bound the number of "observed" agents.
All runs that we consider in this section are finite, and we therefore write them as $\gamma_1 \to \dots \to \gamma_m$ or $\gamma_{start} \runto \gamma_{end}$.
In the rest of this section, we fix an "IOPPUD" $\prot = (Q, \Delta, I, O)$.

We introduce some notation for agents in runs.
Let $\run \colon \config_1 \runto \config_m$ and $d \in \Dataset$.
We let $\Agentset_{\rho}$
be the set of agents "appearing" in $\run$, and set $\Agentset^d_{\rho} \coloneqq \set{a \in \Agentset_{\rho} \mid \dataof(a) = d}$.
We let $\observedagents{\run}{d}$ be the set of agents with datum $d$ that are "observed" in $\run$,
\ie, the $a_o \in \Agentset^d_{\rho}$ such that there exists a "step" $\config \step{\sometest}{a}{a_o} \config'$ in $\run$.
For all $q_1, q_m \in Q$, we let $\Agentset^d_{\rho, q_1, q_m}$ be the set of agents with datum $d$ that start in $q_1$ and end in $q_m$,
\ie, the $a \in \Agentset^d_{\rho}$ such that $\config_1(a) = q_1$ and $ \config_m(a) = q_m$.
Moreover, we let $\Dataset_\run \coloneqq \set{d \in \Dataset \mid \Agentset^d_{\rho} \neq \emptyset}$
be the set of data appearing in $\run$.
We may omit $\rho$ in the subscript if the "run" is clear from the context.

\subsection{Bounds on the Number of Observed Agents per Datum}

Let $\run \colon \config_1 \to \cdots \to \config_{m}$ be a "run".
For $i \in [1,m]$, we call $\config_i \to \config_{i+1}$ the \emph{$i$-th step} in $\run$.
Let $\prefixrun{\run}{i}$ (resp.\ $\suffixrun{\run}{i}$) denote the prefix of $\run$ ending on its $i$-th "configuration" (resp.\ the suffix of $\run$ starting on its $i$-th "configuration").
Let $a,b \in \Agentset_{\run}$.
Agent $a$ is ""active"" in the $i$-th "step" if $\config_i \step{\sometest}{a}{a_o} \config_{i+1}$ for some agent $a_o$.
Otherwise, $a$ is ""idle"" in that step.
We say $b$ ""copies"" $a$ in $\run$ if after every
step $\config_i \step{\sometest}{a}{a_o} \config_{i+1}$ in $\run$
via some transition $t$,
there is a "step" $\config_{i+1} \step{\sometest}{b}{a_o} \config_{i+2}$
via $t$ and, additionally,
$b$ is idle in every "step" not immediately following an active "step" of $a$.

The following lemma allows us to add agents to a "run" that "copy" an agent of the same datum.

\begin{lemma}[Agents copycat]
	\label{lem:agents-copycat}
	Let  $\run \colon \config_{start} \runto \config_{end}$ be a "run". Let $a \in \Agentset_\run$ and $\Tilde{a} \in \Agentset \setminus \Agentset_\run$ with $\dataof(a)=\dataof(\Tilde{a})$.
	Then there exist "configurations" $\Tilde{\config}_{start}, \Tilde{\config}_{end}$ and a "run" $\Tilde{\run} \colon \Tilde{\config}_{start} \runto \Tilde{\config}_{end}$ such that:
	\begin{enumerate}[align=left, label = (\roman*), ref = \cref{lem:agents-copycat}-(\roman*)]
		\item \label{agentcopycat-item1}$\Agentset_{\tilde{\rho}} = \Agentset_\rho \uplus \set{\Tilde{a}}$, and for all $a' \in \Agentset_\rho$, $\tilde{\gamma}_{start}(a') = \gamma_{start}(a')$ and $\tilde{\gamma}_{end}(a') = \gamma_{end}(a')$;
		
		\item \label{agentcopycat-item2} $\Tilde{\config}_{start}(\Tilde{a}) = \config_{start}(a)$ and $\Tilde{\config}_{end}(\Tilde{a}) = \config_{end}(a)$;
		
		\item \label{agentcopycat-item3} $\Tilde{a}$ is not "observed" in $\Tilde{\run}$.
	\end{enumerate}
\end{lemma}

\begin{proof}
	We let $\Tilde{\config}_{start}$ be such that $\Tilde{\config}_{start}(\Tilde{a}) = \config_{start}(a)$ and $\Tilde{\config}_{start}(a') = \config_{start}(a')$ for all $a' \neq \Tilde{a}$.
We  construct $\Tilde{\run}$ by going through $\run$ step by step, making $\Tilde{a}$ "copy" $a$: whenever $\run$ takes a "step" $\step{\sometest}{a}{a_o}$, then we take this "step" followed by "step" $\step{\sometest}{\Tilde{a}}{a_o}$ to $\Tilde{\run}$.
	We can do so because $\dataof(\Tilde{a})=\dataof(a)$ and
	because agent $a_o \neq a$ has not moved and thus can be "observed" again. These are the only "steps" where $\Tilde{a}$ is involved, hence it is never "observed".
\end{proof}

The following result shows that, given a "run" $\run$, we can construct a new "run" with a small subset of the agents of $\Agentset_\run$ such that, for all $d \in \Dataset$ and all states $q_1$ and $q_2$, if there is a "$d$-agent" starting in $q_1$ and ending in $q_2$ in $\run$,
then this is also true in the new run.
The full proof can be found in \cref{app:bounds-observed-agents}.

\begin{restatable}[Agents core]{lemma}{LemmaAgentsCore}
	\label{lem:agents-core-lemma}
	Let $\run \colon \config_{start} \runto \config_{end}$ be a "run".
	Then there exist "configurations" $\config'_{start}$, $\config'_{end}$ and a "run" $\run' \colon \config'_{start} \runto \config'_{end}$ with $\Agentset_{\run'} \subseteq \Agentset_{\run}$ such that:
	\begin{enumerate}[align=left, label = (\roman*), ref = \cref{lem:agents-core-lemma}-(\roman*)]
		\item \label{agentscore-item1} for all $a \in \Agentset_{\run'}$, $\config'_{start}(a) = \config_{start}(a)$ and $\config'_{end}(a) = \config_{end}(a)$;
		
		\item \label{agentscore-item2} for all $d \in \Dataset$ and $q_s, q_e \in Q$, if $\Agentset^d_{\rho, q_s, q_e} \neq \emptyset$,
		then $\Agentset^d_{\rho', q_s, q_e} \neq \emptyset$;
		\item \label{agentscore-item3} for all $d \in \Dataset$, we have $|\Agentset^d_{\rho'}| \leq \size{Q}^3$.
	\end{enumerate}
\end{restatable}
\begin{proofsketch}
We adapt the bunch argument from the case of IO protocols without data \cite{EsparzaRW2019}.
Suppose there is $d\in \Dataset$ and $q_s, q_e \in Q$ such that $\size{\Agentset^d_{\run,q_s, q_e}} > \size{Q}$.
Let $\bunchreach$ be the set of states visited by agents of $\Agentset^d_{\run,q_s, q_e}$ in $\run$. 	
Notice that $|\bunchreach| \leq |Q|$.
We define a family $(a_q)_{q \in \bunchreach}$ of pairwise distinct agents such that reducing $\Agentset^d_{\run,q_s, q_e}$ in $\run$ to $(a_q)_{q \in \bunchreach}$ still yields a valid run.
	
	We iterate through $\bunchreach$ as follows. Let $q\in\bunchreach$ and
	let $f$ be the first moment $q$ is reached in $\run$,
	\ie, the minimal index such that there exists an $a\in \Agentset^d_{\run,q_s, q_e}$
	with $\gamma_f(a)=q$.
	Let $\ell$ be the last moment $q$ is occupied in $\run$,
	\ie, the maximal index such that there exists an $a\in \Agentset^d_{\run,q_s, q_e}$
	with $\gamma_\ell(a)=q$.
	Let $\alpha_q$ be the agent in $\Agentset^d_{\run,q_s, q_e}$ that reaches $q$ first,
	\ie, $\gamma_f(\alpha_q)=q$, and
	let $\beta_q$ be the agent in $\Agentset^d_{\run, q_s, q_e}$ that leaves
	$q$ last,
	\ie, $\gamma_l(\beta_q)=q$.
	Note that these agents do not have to be distinct.
We pick a fresh agent $a_q \notin \Agentset_\run$ with $\dataof(a_q) = d$ and modify $\run$ as follows. We let $a_q$ "copy" $\alpha_q$ in $\run[\to f]$, then $a_q$ stays "idle" until $\beta_q$ leaves $q$ (for the last time) and then $a_q$ "copies" $\beta_q$ in $\run[\ell \to]$. We do this for every $q \in \bunchreach$. 	

Then, for every "step" in which an $a_o$ in $\Agentset^d_{\run,q_s, q_e}$ is "observed" in state $q$, let $a_q$ be "observed" instead, \ie, replace steps $\step{\sometest}{a}{a_o}$ with $\step{\sometest}{a}{a_q}$.
Finally, remove all the agents of $\Agentset^d_{\run,q_s,q_e}$ from the run,
and identify (or substitute) each $a_q$ with a distinct agent in  $\Agentset^d_{\run,q_s,q_e}$, so that $(a_q)_{q \in \bunchreach} \subseteq \Agentset^d_{\run,q_s,q_e}$.
We do this for every $d\in \Dataset$ and $q_s, q_e \in Q$ such that $\size{\Agentset^d_{\run,q_s, q_e}} > \size{Q}$.
\end{proofsketch}

\begin{example}
\label{ex-bunch}
Consider the "run" $\run$ depicted in \cref{fig-bunch}. Applying \cref{lem:agents-core-lemma} on $\run$ yields a new "run" $\run'$ with $4$ agents instead of $5$. Indeed, let $d$ denote the datum of $\agentbunch{a}$, $\agentbunch{b}$ and $\agentbunch{c}$; we have $\size{\Agentset^d_{\run,q_1,q_1}} = \size{\set{\agentbunch{a},\agentbunch{b},\agentbunch{c}}} = 3$ whereas $\size{Q} = 2$. In $\run$, agents $\agentbunch{a}$ and $\agentbunch{b}$ successively go from $q_1$ to $q_2$ and back to $q_1$. In $\run'$, these two agents are replaced by a single agent (named $\agentbunch{b}$ again) who goes to $q_2$ on the first "step" and only leaves $q_2$ on the last step. In $\run'$, the new agent $\agentbunch{b}$ is "observed" by $\agentbunch{d}$ in the second step, and by $\agentbunch{e}$ in the penultimate step.
\end{example}

\begin{figure}
	\centering
	\begin{tikzpicture}[
  align=center,node distance=2cm,thick,inner sep=7pt,
  >={stealth'},-={stealth',ultra thick,scale=3} , xscale = 0.95]

  \colorlet{dataone}{IoppudBlue}
  \colorlet{datatwo}{IoppudMagenta}
  \colorlet{obscolor}{black}
  \tikzstyle{agent} = [draw, circle, minimum size = .8em, inner sep = 0pt, font = {\small}]
  \tikzstyle{agent one} = [agent, draw=dataone, fill=dataone, text=white]
  \tikzstyle{agent two} = [agent, draw=datatwo, fill=datatwo, text=white]
  \tikzstyle{moving} = [font=\small, inner sep=2pt]
  \tikzstyle{moving one} = [moving, text=dataone!80!black]
  \tikzstyle{moving two} = [moving, text=datatwo!80!black]
  \newcommand{\eqtestnode}{$\boldsymbol{=}$}
  \newcommand{\diseqtestnode}{$\boldsymbol{\neq}$}

  \newcommand{\drawconfig}[6]{ \draw[configuration, rounded corners=.5em] (#1-0.75,0) rectangle (#1+0.75,3);
    \node at (#1,0.75) [medium protocol state] (q2#1) {};
    \node at (#1,2.25) [medium protocol state] (q1#1) {};
    \node[agent one] at ($(#1, 0.75+1.5*#2-1.5) + 0.31*(-0.80901699437, -0.5877852522)$) (agenta#1) {$a$};
    \node[agent one] at ($(#1, 0.75+1.5*#3-1.5) + 0.31*(-0.80901699437, 0.58778525229)$) (agentb#1) {$b$};
    \node[agent one] at ($(#1, 0.75+ 1.5*#4-1.5) + 0.31*(0.3090, 0.9510)$) (agentc#1) {$c$};
    \node[agent two] at ($(#1, 0.75+1.5*#5-1.5) + 0.31*(1, 0)$) (agentd#1) {$d$};
    \node[agent two] at ($(#1-0.02, 0.75+1.5*#6-1.5) + 0.31*(0.30901699437, -0.95105651629)$) (agente#1) {$e$};
  }

  \node at (-1.25, 0.75) (q2label) {$q_2$};
  \node at (-1.25, 2.25) (q1label) {$q_1$};
  \drawconfig{0}{2}{2}{2}{1}{1}
  \drawconfig{2}{1}{2}{2}{1}{1}
  \drawconfig{4}{1}{2}{2}{2}{1}
  \drawconfig{6}{2}{2}{2}{2}{1}
  \drawconfig{8}{2}{1}{2}{2}{1}
  \drawconfig{10}{2}{1}{2}{2}{2}
  \drawconfig{12}{2}{2}{2}{2}{2}
  \draw [-stealth, very thick, color = dataone] (agenta0) -- (agenta2); \draw [-stealth, very thick, color = datatwo] (agentd2) -- (agentd4); \draw [-stealth, very thick, color = dataone] (agenta4) -- (agenta6); \draw [-stealth, very thick, color = dataone] (agentb6) -- (agentb8);  \draw [-stealth, very thick, color = datatwo] (agente8) -- (agente10); \draw [-stealth, very thick, color = dataone] (agentb10) -- (agentb12); \node[moving one] at ($(agenta0)!0.62!(agenta2) + (0,1.5ex)$) (0to2) {$a$};
  \node[moving two] at ($(agentd2)!0.35!(agentd4)+ (0,1.5ex)$) (2to4) {$d$};
  \node[moving one] at ($(agenta4)!0.62!(agenta6)+ (0,-1.5ex)$) (4to6) {$a$};
  \node[moving one] at ($(agentb6)!0.62!(agentb8)+ (0,1.5ex)$) (6to8) {$b$};
  \node[moving two] at ($(agente8)!0.46!(agente10)+ (0,1.5ex)$) (8to10) {$e$};
  \node[moving one] at ($(agentb10)!0.62!(agentb12)+ (0,1.5ex)$) (10to12) {$b$};

  \path[dashed, thick, color=obscolor]
  (0to2) edge[bend right=15] node[pos=.45, above right=-1em, color=obscolor] {\eqtestnode} (agentc0)
  (2to4) edge[bend right=15] node[pos=.4, above=-.5em, color=obscolor] {\diseqtestnode} (agenta2)
  (4to6) edge[bend left=15] node[pos=.2, below, color=obscolor] {\diseqtestnode} (agente4)
  (6to8) edge[bend right=15] node[above right=-1em, color=obscolor] {\diseqtestnode} (agentd6)
  (8to10) edge[bend right=10] node[pos=.4, above=-.5em, color=obscolor] {\diseqtestnode} (agentb8)
  (10to12) edge[bend right=46] node[pos=.2, above=-.5em, color=obscolor] {\eqtestnode} (agenta10);
\end{tikzpicture}
 	\caption{An example of a "run" with six steps on a protocol with two states $q_1,q_2$. $\agentbunch{a},\agentbunch{b},\agentbunch{c},\agentbunch{d},\agentbunch{e}$ denote agents; $\agentbunch{a},\agentbunch{b},\agentbunch{c}$ have the same datum and $\agentbunch{d},\agentbunch{e}$ have the same datum. Dashed lines are observations.}
	\label{fig-bunch}
\end{figure}
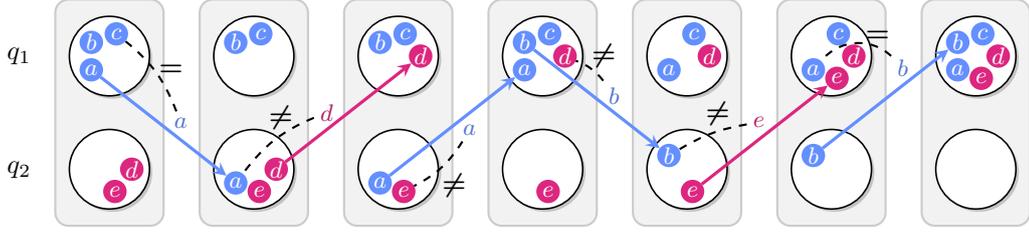

\subsection{Bounds on the Number of Observed Data}

Given a "run" and a datum $d$ appearing in it, we define the ""trace"" of $d$ in $\run$ as the function $\trace{d}{\run} \colon Q^2 \to \nats$ such that for all $q_1, q_2 \in Q$,
it holds that $\trace{d}{\run}(q_1, q_2) = \size{\Agentset^d_{\rho, q_1, q_2}}$.
For each pair of states $q_1, q_2$, the "trace" counts the number of "$d$-agents" starting in $q_1$ and ending in $q_2$.
For example, the "trace" of the "run" $\run$ of \cref{ex-bunch}, with $d$ the datum of agents $\agentbunch{a},\agentbunch{b}$ and $\agentbunch{c}$, is such that $\trace{d}{\run}(q_1, q_1) = 3$ and $\trace{d}{\run}(q, q') = 0$ for all $(q,q') \ne (q_1,q_1)$.
The "trace" is the information we need to copy data: if there is a datum $d$ with "trace" $\atrace$ in a "run", then we can add data to the "run" that mimic $d$ and have the same trace.
The following lemma echoes Lemma \ref{lem:agents-copycat}.

\begin{lemma}[Data copycat]
	\label{lem:data-copycat}
	Let  $\run \colon \config_{start} \runto \config_{end}$ be a "run". Let $d \in \Dataset_\run$ and $\Tilde{d} \in \Dataset \setminus \Dataset_\run$. Then there exist "configurations" $\Tilde{\config}_{start}$, $\Tilde{\config}_{end}$ and a "run" $\Tilde{\run} \colon \Tilde{\config}_{start} \runto \Tilde{\config}_{end}$ such that:
	\begin{enumerate}[align=left,label = (\roman*), ref = \cref{lem:data-copycat}-(\roman*)]
\item \label{datacopycat-item1} $\Agentset_{\tilde{\rho}} = \Agentset_\rho \uplus \Agentset^{\tilde{d}}_{\tilde{\rho}}$, and for all $a \in \Agentset_\rho$, $\tilde{\gamma}_{init}(a) = \gamma_{init}(a)$ and $\tilde{\gamma}_{end}(a) = \gamma_{end}(a)$,
		
		\item \label{datacopycat-item2}$\trace{\Tilde{d}}{\Tilde{\run}} = \trace{d}{\run}$ and $\trace{d'}{\Tilde{\run}} = \trace{d'}{\run}$ for all $d' \neq \Tilde{d}$,

		\item \label{datacopycat-item3} $\Tilde{d}$ is not "externally observed" in $\Tilde{\run}$.
\end{enumerate}
\end{lemma}

\begin{proof}
	For all $q_s,q_e \in Q$ and all $a \in \Agentset^d_{q_s, q_e}$,
	we add an agent $\Tilde{a}$ with datum $\Tilde{d}$ in $q_s$ at the start.
	We do this in a way similar to Lemma \ref{lem:agents-copycat}:
	after every "step" $\step{\neq}{a}{a_o}$ in $\run$, we insert a "step" $\step{\neq}{\Tilde{a}}{a_o}$, and
	after every "step" $\step{=}{a}{a_o}$ in $\run$, we insert a "step" $\step{=}{\Tilde{a}}{\Tilde{a}_o}$.
	We thus maintain the fact that each added agent $\Tilde{a}$ is in the same state as its counterpart $a$. In particular, they are in the same state at the end of the "run".
	This yields a "run" $\Tilde{\run}$ with $\trace{\Tilde{d}}{\Tilde{\run}} = \trace{d}{\run}$, and such that  for all $d' \neq \Tilde{d}$, $\trace{d'}{\Tilde{\run}} = \trace{d'}{\run}$.
	Since $\Tilde{d} \notin \Dataset_{\run}$, it is not "externally observed" in $\Tilde{\run}$.
\end{proof}

Like we showed for the agents, we show that we can reduce the number of data in a run.
We lift the proof strategy of \cref{lem:agents-core-lemma} from agents to data, exploiting the sets of data with equal traces.
The full proof is in \cref{app:bounds-observed-agents}.

\begin{restatable}[Data core]{lemma}{LemmaDataCore}
	\label{lem:data-core-lemma}
	Let $\run \colon \config_{start} \runto \config_{end}$ be a "run" and let $K$ be a number such that there are at most $K$  agents of each datum in $\run$.
	Then there exist "configurations" $\config'_{start}$, $\config'_{end}$, a "run" $\run' \colon \config'_{start} \runto \config'_{end}$, and a subset of data $\Dataset_{\run'} \subseteq \Dataset_{\run}$ such that:
	\begin{enumerate}[align=left,label = (\roman*), ref = \cref{lem:data-core-lemma}-(\roman*)]
		\item \label{datacore-item1} for all $d \in \Dataset_{\run'}$ and all agents $a$ of datum $d$, $\config_{start}(a) = \config'_{start}(a)$ and $\config_{end}(a) = \config'_{end}(a)$,
		
		\item \label{datacore-item2} for all $d \in \Dataset_{\run}$, there exists $d' \in \Dataset_{\run'}$ such that $\trace{d'}{\run'} = \trace{d}{\run}$,
		
		\item \label{datacore-item3} $\size{\Dataset_{\run'}} \leq (K+1)^{\size{Q}^3+\size{Q}^2}$.
	\end{enumerate}
\end{restatable}

\begin{proofsketch}
	We define the notion of "split trace".
	The "split trace" of a datum $d$ at the $i$-th "configuration" of a "run" $\run$ maps every triple of states $(q_1,q_2,q_3)$ to the number of "$d$-agents" that are in $q_1$ at the start of $\run$, then in $q_2$ in the $i$-th configuration, and finally in $q_3$ at the end.
	Since there are at most $K$ agents per datum, there are at most $(K+1)^{|Q^2|}$ possible "traces" and $M=(K+1)^{|Q^3|}$ possible "split traces".
	
	For every trace $\atrace$, if there are more than $M$ data that have trace $\atrace$ in $\run$, we  apply a similar argument to~\cref{lem:agents-core-lemma}:
	 we select one datum for each possible "split trace", and use it to cover all "external observations" of agents whose datum matches that "split trace". We remove the other data, and show that this is still a valid run.
	The bound on the total number of data comes from the number of "traces" and "split traces".
\end{proofsketch}

\begin{corollary}
	\label{cor:run-few-observed-data-and-agents}
  For every "run" $\run \colon \config_{start} \runto \config_{end}$, there exists a "run" $\Tilde{\run} \colon \config_{start} \runto \config_{end}$ such that for all $d \in \Dataset$, it holds that $|\Agentset^d_{\Tilde{\run},o}| \leq \size{Q}^3$ and that agents of at most $(\size{Q}^3+1)^{\size{Q}^3+\size{Q}^2}$ different data are "externally observed".
\end{corollary}

\begin{proof}
	We first apply~\cref{lem:agents-core-lemma} to $\run$ to obtain
	$\run^{(1)} \colon \config_{start}^{(1)} \runto \config_{end}^{(1)}$ over the same data
	such that for all $d \in \Dataset$, it holds that $|\Agentset^d_{\run^{(1)}}| \leq \size{Q}^3$.
	Then we apply \cref{lem:data-core-lemma} to obtain
	$\run^{(2)} \colon \config_{start}^{(2)} \runto \config_{end}^{(2)}$ with at most $(\size{Q}^3+1)^{\size{Q}^3+\size{Q}^2}$ data.
	By \ref{agentscore-item1} and \ref{datacore-item1}, the remaining agents have the same initial and final states in $\run$ and $\run^{(2)}$. It remains to put back the agents and data we removed, without increasing the number of "externally observed" data or "observed" agents per datum.
	
	By \ref{datacore-item2}, every "trace" of a datum in $\run^{(1)}$ appears as the "trace" of a datum in $\run^{(2)}$.
	Thus, it is possible to re-add data of $\Dataset_{\run^{(1)}}$ to $\Dataset_{\run^{(2)}}$ using repeated applications of~\cref{lem:data-copycat}.
	By \ref{datacopycat-item3}, this does not add any "external observation".
	So we obtain a "run" $\Tilde{\run}^{(1)}$ from $\config_{start}^{(1)}$ to $\config_{end}^{(1)}$ such that at most $(\size{Q}^3+1)^{\size{Q}^3+\size{Q}^2}$ data are "externally observed" by \ref{datacopycat-item3}.
	Recall that there are at most $\size{Q}^3$ agents per datum in $\config^{(1)}_{start}$ by \ref{agentscore-item3}; in particular there are at most $\size{Q}^3$ "observed" agents per datum in $\Tilde{\run}^{(1)}$.
	
	By \ref{agentscore-item2}, for each datum $d$ and states $q_s,q_e$,
	if there is a "$d$-agent" $a$ such that $\config_{start}(a)=q_{s}$ and $\config_{end}(a) = q_{e}$ then there is an agent $a'$ such that $\config^{(1)}_{start}(a')=q_{s}$ and $\config^{(1)}_{end}(a') = q_{e}$ in $\run$.
	Therefore, due to \ref{agentscore-item2}, we can apply \cref{lem:agents-copycat} repeatedly to add back the missing agents in $\Tilde{\run}^{(1)}$ and obtain a "run" $\Tilde{\run}$ from $\config_{start}$ to $\config_{end}$. By \ref{agentcopycat-item3}, this does not add any "observation".
	As a result, we obtain a "run" from $\config_{start}$ to $\config_{end}$ in which at most $(\size{Q}^3+1)^{\size{Q}^3+\size{Q}^2}$ data are "externally observed" and for all datum $d$, at most $\size{Q}^3$ "$d$-agents" are "observed".
\end{proof}
 
	\section{From Expressions to Containers}
\label{sec:equivalence-relation}

In this section, we define the technical notions of "boxes" and "containers", which are meant to represent sets of "configurations" defined by counting agents and data up to some thresholds.
In \cref{prop:bound-coefs-GRE}, we will prove that the set of configurations defined by a "generalised reachability expression" $E$ can be described as a union of "containers" whose thresholds are exponential in the "length@@exp" of $E$ and polynomial in its "norm@@exp". To do so, we will leverage the bounds on the number of observed agents from \cref{sec:bounds-observed-agents} to bound the description of the "GRE" $\poststar{F}$ with respect to the one of "GRE" $F$.  The key result of \cref{prop:bound-coefs-GRE} will be used in \cref{sec:complexity} to obtain the decidability of the "emptiness problem for GRE".

\subsection{Equivalence of Predicates and Containers}
In this subsection, we fix an "IOPPUD" $\prot = (Q, \Delta, I, O)$.

\AP Let $n, M \in \nats$.
An $n$-""box"" is a vector $\abox \colon Q \to \nset{0}{n}$.
Given a configuration $\config$ and a datum $d \in \Dataset$, we define \emph{the $n$-"box" of $d$ in $\config$} as $\cubeapprox{n}{\config}{d} \colon Q \to \nset{0}{n}$ such that for all $q\in Q$, $\intro*\cubeapprox{n}{\config}{d}(q)=\min \set{n, \counting{\config}{d}(q)}$; in words, the $n$-box of $d$ truncates the number of agents of $d$ if it exceeds $n$. We write $\intro*\Boxes{n}$ for the set of all $n$-"boxes".
We define the equivalence relation $\cubeequiv{n}$ over $\configset \times \Dataset$ by $(\config_1, d_1) \intro*\cubeequiv{n} (\config_2, d_2)$ whenever $\cubeapprox{n}{\config_1}{d_1} = \cubeapprox{n}{\config_2}{d_2}$.
An equivalence class of $\cubeequiv{n}$ is a set of the form $\set{(\config, d) \in \configset\times \Dataset \mid \cubeapprox{n}{\config}{d} = \abox}$ for $\abox \in \Boxes{n}$; we represent such an equivalence class for $\cubeequiv{n}$ by the associated $n$-"box" $\abox$.

\AP
To lift this concept to data, we count the number of data with the same $n$-"box" up to bound $M$.
The $(n,M)$-""container"" of a "configuration" $\config$ is the function $\intro*\metaapprox{n}{M}{\config} \colon \Boxes{n} \to \nset{0}{M}$ such that
$\metaapprox{n}{M}{\config}(\abox) = \min \bigset{M, \bigsize{\set{d \in \Dataset \mid \cubeapprox{n}{\config}{d} = \abox}}}$ for all $\abox \in \Boxes{n}$.
We define the equivalence relation $\metaequiv{n}{M}$ over $\configset$ by $\config_1 \intro*\metaequiv{n}{M} \config_2$ whenever $\metaapprox{n}{M}{\config_1} = \metaapprox{n}{M}{\config_2}$. An equivalence relation for $\metaequiv{n}{M}$ is the preimage of some $(n,M)$-"container" by the previously described function; we represent such an equivalence class by the associated $(n,M)$-"container". \cref{fig-boxes} illustrates the function mapping a given configuration to its container.

\setlength{\FrameSep}{.7em}
\begin{framed}
In all the following,
we use the terms $n$-"boxes" and $(n,M)$-"containers" to designate both the vectors and the equivalence classes of $\cubeequiv{n}$ and $\metaequiv{n}{M}$ that they represent.
For instance, we write \emph{union of $n$-boxes} for the union of the corresponding equivalence classes of $\cubeequiv{n}$.
\end{framed}
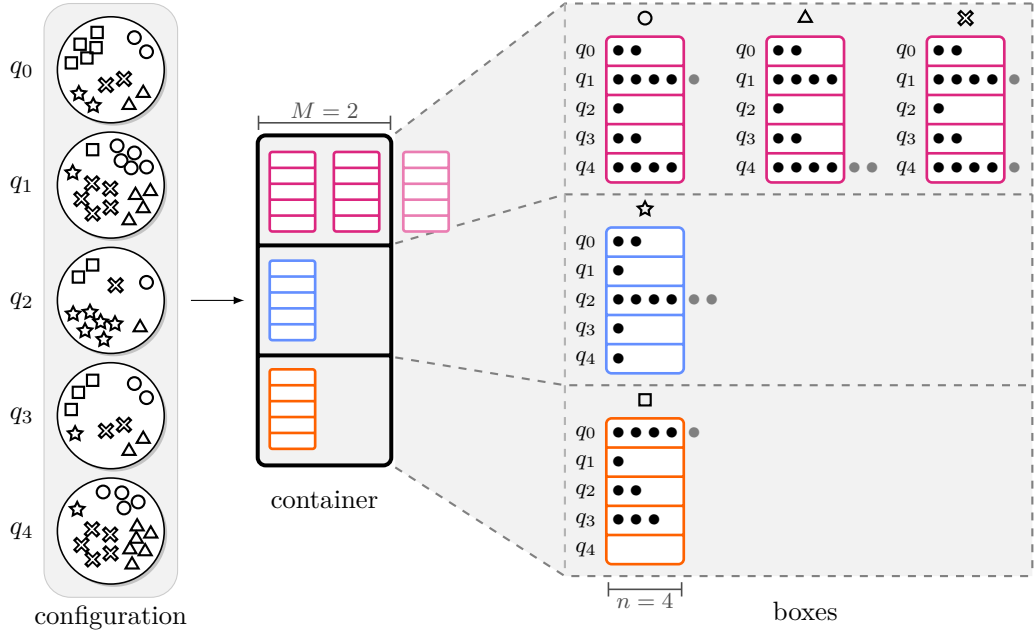
\begin{figure}
  \centering
  \begin{tikzpicture}
  \colorlet{Datatype1}{IoppudMagenta}
  \colorlet{Datatype2}{IoppudBlue}
  \colorlet{Datatype3}{IoppudOrange}

  \node[big protocol state, label={[label distance=.5em]left:{\(q_0\)}}] (q0) {};
  \foreach \i [count=\j] in {0,1,2,3} {
    \node[big protocol state, label={[label distance=.5em]left:{\(q_\j\)}}, below of=q\i, node distance=4.35em] (q\j) {};
  };
  \begin{pgfonlayer}{background}
    \draw[configuration] ($(q0) + (-2.5em,2.5em)$) rectangle ($(q4) + (2.5em,-2.5em)$);
  \end{pgfonlayer}
  \node[below of=q4, node distance=3.25em] {configuration};

\def\statenum{0}
  \foreach \i [count=\c] in {1,2} {
    \node[ccircle] at ([shift={(q\statenum)}]{27 * \c}:1.5em) (q\statenum{}ci\i) {};
  };
  \foreach \i [count=\c] in {1,2,3} {
    \node[square] at ([shift={(q\statenum)}]{30 * \c + 80}:1.5em) (q\statenum{}sq\i) {};
  };
  \foreach \i [count=\c] in {4,5} {
    \node[square] at ([shift={(q\statenum)}]{30 * \c + 93.5}:1.0em) (q\statenum{}sq\i) {};
  };
  \foreach \i [count=\c] in {1,2} {
    \node[sstar] at ([shift={(q\statenum)}]{27 * \c + 190}:1.5em) (q\statenum{}st\i) {};
  };
  \foreach \i [count=\c] in {1,2} {
    \node[triangle] at ([shift={(q\statenum)}]{27 * \c + 270}:1.5em) (q\statenum{}t\i) {};
  };
  \foreach \i in {1,2} {
    \node[crosso] at ([shift={(q\statenum)}]{72.5 * \i + 180}:.6em) (q\statenum{}cro\i) {};
    \node[crossi] at ([shift={(q\statenum)}]{72.5 * \i + 180}:.6em) (q\statenum{}cri\i) {};
  };

\def\statenum{1}
  \foreach \i [count=\c] in {1,2,3} {
    \node[ccircle] at ([shift={(q\statenum)}]{27 * \c}:1.5em) (q\statenum{}ci\i) {};
  };
  \foreach \i [count=\c] in {4,5} {
    \node[ccircle] at ([shift={(q\statenum)}]{27 * \c + 13.5}:1.0em) (q\statenum{}ci\i) {};
  };
  \foreach \i [count=\c] in {1} {
    \node[square] at ([shift={(q\statenum)}]{27 * \c + 90}:1.5em) (q\statenum{}sq\i) {};
  };
  \foreach \i [count=\c] in {1} {
    \node[sstar] at ([shift={(q\statenum)}]{27 * \c + 135}:1.5em) (q\statenum{}st\i) {};
  };
  \foreach \i [count=\c] in {1,2,3} {
    \node[triangle] at ([shift={(q\statenum)}]{27 * \c + 270}:1.5em) (q\statenum{}t\i) {};
  };
  \foreach \i [count=\c] in {4} {
    \node[triangle] at ([shift={(q\statenum)}]{27 * \c + 310.5}:1.0em) (q\statenum{}t\i) {};
  };
  \foreach \i in {1,2,3,4,5} {
    \node[crosso] at ([shift={($(q\statenum)+(-.5em,-.5em)$)}]{72.5 * \i + 180}:.6em) (q\statenum{}cro\i) {};
    \node[crossi] at ([shift={($(q\statenum)+(-.5em,-.5em)$)}]{72.5 * \i + 180}:.6em) (q\statenum{}cri\i) {};
  };

\def\statenum{2}
  \foreach \i [count=\c] in {1} {
    \node[ccircle] at ([shift={(q\statenum)}]{27 * \c}:1.5em) (q\statenum{}ci\i) {};
  };
  \foreach \i [count=\c] in {1,2} {
    \node[square] at ([shift={(q\statenum)}]{27 * \c + 90}:1.5em) (q\statenum{}sq\i) {};
  };
  \foreach \i [count=\c] in {1,2,3} {
    \node[sstar] at ([shift={(q\statenum)}]{30 * \c + 170}:1.5em) (q\statenum{}st\i) {};
  };
  \foreach \i [count=\c] in {4,5,6} {
    \node[sstar] at ([shift={(q\statenum)}]{35 * \c + 175}:.9em) (q\statenum{}st\i) {};
  };
  \foreach \i [count=\c] in {1} {
    \node[triangle] at ([shift={(q\statenum)}]{27 * \c + 290}:1.5em) (q\statenum{}t\i) {};
  };
  \foreach \i in {1} {
    \node[crosso] at ([shift={(q\statenum)}]{72.5 * \i}:.6em) (q\statenum{}cro\i) {};
    \node[crossi] at ([shift={(q\statenum)}]{72.5 * \i}:.6em) (q\statenum{}cri\i) {};
  };

\def\statenum{3}
  \foreach \i [count=\c] in {1,2} {
    \node[ccircle] at ([shift={(q\statenum)}]{27 * \c}:1.5em) (q\statenum{}ci\i) {};
  };
  \foreach \i [count=\c] in {1,2,3} {
    \node[square] at ([shift={(q\statenum)}]{27 * \c + 90}:1.5em) (q\statenum{}sq\i) {};
  };
  \foreach \i [count=\c] in {1} {
    \node[sstar] at ([shift={(q\statenum)}]{27 * \c + 180}:1.5em) (q\statenum{}st\i) {};
  };
  \foreach \i [count=\c] in {1,2} {
    \node[triangle] at ([shift={(q\statenum)}]{27 * \c + 270}:1.5em) (q\statenum{}t\i) {};
  };
  \foreach \i in {1,2} {
    \node[crosso] at ([shift={(q\statenum)}]{72.5 * \i + 180}:.6em) (q\statenum{}cro\i) {};
    \node[crossi] at ([shift={(q\statenum)}]{72.5 * \i + 180}:.6em) (q\statenum{}cri\i) {};
  };

\def\statenum{4}
  \foreach \i [count=\c] in {1,2,3} {
    \node[ccircle] at ([shift={(q\statenum)}]{27 * \c + 20}:1.5em) (q\statenum{}ci\i) {};
  };
  \foreach \i [count=\c] in {4} {
    \node[ccircle] at ([shift={(q\statenum)}]{27 * \c + 33.5}:1.0em) (q\statenum{}ci\i) {};
  };
  \foreach \i [count=\c] in {1} {
    \node[sstar] at ([shift={(q\statenum)}]{27 * \c + 120}:1.5em) (q\statenum{}st\i) {};
  };
  \foreach \i [count=\c] in {1,2,3} {
    \node[triangle] at ([shift={(q\statenum)}]{27 * \c + 275}:1.5em) (q\statenum{}t\i) {};
  };
  \foreach \i [count=\c] in {4,5,6} {
    \node[triangle] at ([shift={(q\statenum)}]{27 * \c + 288.5}:1.0em) (q\statenum{}t\i) {};
  };
  \foreach \i in {1,2,3,4,5} {
    \node[crosso] at ([shift={($(q\statenum)+(-.5em,-.5em)$)}]{72.5 * \i + 180}:.6em) (q\statenum{}cro\i) {};
    \node[crossi] at ([shift={($(q\statenum)+(-.5em,-.5em)$)}]{72.5 * \i + 180}:.6em) (q\statenum{}cri\i) {};
  };

\node[container]
    at ($(q2) + (8em, 0)$)
    (container)
    {\nodepart{two}
    \nodepart{three}};
  \node[below of=container, node distance=7.5em] {container};
\draw[|-|, semithick, black!70] ($(container) + (-2.5em, 6.7em)$)
    -- node[above=-.2em,font=\footnotesize, color=black!70] {$M=2$}
    ($(container) + (2.5em, 6.7em)$);

\node[small box, draw=Datatype1]
    at ($(container) + (-1.2em, 4.1em)$)
    {\nodepart{two}
    \nodepart{three}
    \nodepart{four}
    \nodepart{five}};
  \node[small box, draw=Datatype1]
    at ($(container) + (1.2em, 4.1em)$)
    {\nodepart{two}
    \nodepart{three}
    \nodepart{four}
    \nodepart{five}};
  \node[small box, draw=Datatype1!60]
    at ($(container) + (3.8em, 4.1em)$)
    {\nodepart{two}
    \nodepart{three}
    \nodepart{four}
    \nodepart{five}};
\node[small box, draw=Datatype2]
    at ($(container) + (-1.2em, 0em)$)
    {\nodepart{two}
    \nodepart{three}
    \nodepart{four}
    \nodepart{five}};
\node[small box, draw=Datatype3]
    at ($(container) + (-1.2em, -4.1em)$)
    {\nodepart{two}
    \nodepart{three}
    \nodepart{four}
    \nodepart{five}};

    \draw[-latex] ($(container) - (5em, 0em)$) --  ($(container) - (3em, 0em)$);

\def\boxname{circlebox}
  \node[big box, draw=Datatype1]
    at ($(container) + (12em,7.2em)$) (\boxname)
    {\nodepart{two}
    \nodepart{three}
    \nodepart{four}
    \nodepart{five}};
  \node[ccircle, above of=\boxname, node distance=3.45em] {};
  \node at ($(\boxname) + (-2.2em,2.2em)$) {\footnotesize\(q_0\)};
  \node[right] at ($(\boxname) + (-1.6em,2.2em)$)
  {\(\bullet\)\,\(\bullet\)};
  \node at ($(\boxname) + (-2.2em,1.1em)$) {\footnotesize\(q_1\)};
  \node[right] at ($(\boxname) + (-1.6em,1.1em)$)
  {\(\bullet\)\,\(\bullet\)\,\(\bullet\)\,\(\bullet\)\ \textcolor{gray}{\(\bullet\)}};
  \node at ($(\boxname) + (-2.2em,0)$) {\footnotesize\(q_2\)};
  \node[right] at ($(\boxname) + (-1.6em,0)$)
  {\(\bullet\)};
  \node at ($(\boxname) + (-2.2em,-1.1em)$) {\footnotesize\(q_3\)};
  \node[right] at ($(\boxname) + (-1.6em,-1.1em)$)
  {\(\bullet\)\,\(\bullet\)};
  \node at ($(\boxname) + (-2.2em,-2.2em)$) {\footnotesize\(q_4\)};
  \node[right] at ($(\boxname) + (-1.6em,-2.2em)$)
  {\(\bullet\)\,\(\bullet\)\,\(\bullet\)\,\(\bullet\)};

  \def\boxname{trianglebox}
  \node[big box, draw=Datatype1]
    at ($(circlebox) + (6em,0)$) (\boxname)
    {\nodepart{two}
    \nodepart{three}
    \nodepart{four}
    \nodepart{five}};
  \node[triangle, above of=\boxname, node distance=3.45em] {};
  \node at ($(\boxname) + (-2.2em,2.2em)$) {\footnotesize\(q_0\)};
  \node[right] at ($(\boxname) + (-1.6em,2.2em)$)
  {\(\bullet\)\,\(\bullet\)};
  \node at ($(\boxname) + (-2.2em,1.1em)$) {\footnotesize\(q_1\)};
  \node[right] at ($(\boxname) + (-1.6em,1.1em)$)
  {\(\bullet\)\,\(\bullet\)\,\(\bullet\)\,\(\bullet\)};
  \node at ($(\boxname) + (-2.2em,0)$) {\footnotesize\(q_2\)};
  \node[right] at ($(\boxname) + (-1.6em,0)$)
  {\(\bullet\)};
  \node at ($(\boxname) + (-2.2em,-1.1em)$) {\footnotesize\(q_3\)};
  \node[right] at ($(\boxname) + (-1.6em,-1.1em)$)
  {\(\bullet\)\,\(\bullet\)};
  \node at ($(\boxname) + (-2.2em,-2.2em)$) {\footnotesize\(q_4\)};
  \node[right] at ($(\boxname) + (-1.6em,-2.2em)$)
  {\(\bullet\)\,\(\bullet\)\,\(\bullet\)\,\(\bullet\)\ \textcolor{gray}{\(\bullet\)\,\(\bullet\)}};

  \def\boxname{crossbox}
  \node[big box, draw=Datatype1]
    at ($(trianglebox) + (6em,0)$) (\boxname)
    {\nodepart{two}
    \nodepart{three}
    \nodepart{four}
    \nodepart{five}};
  \node[crosso, above of=\boxname, node distance=3.45em] {};
  \node[crossi, above of=\boxname, node distance=3.45em] {};
  \node at ($(\boxname) + (-2.2em,2.2em)$) {\footnotesize\(q_0\)};
  \node[right] at ($(\boxname) + (-1.6em,2.2em)$)
  {\(\bullet\)\,\(\bullet\)};
  \node at ($(\boxname) + (-2.2em,1.1em)$) {\footnotesize\(q_1\)};
  \node[right] at ($(\boxname) + (-1.6em,1.1em)$)
  {\(\bullet\)\,\(\bullet\)\,\(\bullet\)\,\(\bullet\)\ \textcolor{gray}{\(\bullet\)}};
  \node at ($(\boxname) + (-2.2em,0)$) {\footnotesize\(q_2\)};
  \node[right] at ($(\boxname) + (-1.6em,0)$)
  {\(\bullet\)};
  \node at ($(\boxname) + (-2.2em,-1.1em)$) {\footnotesize\(q_3\)};
  \node[right] at ($(\boxname) + (-1.6em,-1.1em)$)
  {\(\bullet\)\,\(\bullet\)};
  \node at ($(\boxname) + (-2.2em,-2.2em)$) {\footnotesize\(q_4\)};
  \node[right] at ($(\boxname) + (-1.6em,-2.2em)$)
  {\(\bullet\)\,\(\bullet\)\,\(\bullet\)\,\(\bullet\)\ \textcolor{gray}{\(\bullet\)}};

\def\boxname{starbox}
  \node[big box, draw=Datatype2]
    at ($(container) + (12em,0)$) (\boxname)
    {\nodepart{two}
    \nodepart{three}
    \nodepart{four}
    \nodepart{five}};
  \node[sstar, above of=\boxname, node distance=3.45em] {};
  \node at ($(\boxname) + (-2.2em,2.2em)$) {\footnotesize\(q_0\)};
  \node[right] at ($(\boxname) + (-1.6em,2.2em)$)
  {\(\bullet\)\,\(\bullet\)};
  \node at ($(\boxname) + (-2.2em,1.1em)$) {\footnotesize\(q_1\)};
  \node[right] at ($(\boxname) + (-1.6em,1.1em)$)
  {\(\bullet\)};
  \node at ($(\boxname) + (-2.2em,0)$) {\footnotesize\(q_2\)};
  \node[right] at ($(\boxname) + (-1.6em,0)$)
  {\(\bullet\)\,\(\bullet\)\,\(\bullet\)\,\(\bullet\)\ \textcolor{gray}{\(\bullet\)\,\(\bullet\)}};
  \node at ($(\boxname) + (-2.2em,-1.1em)$) {\footnotesize\(q_3\)};
  \node[right] at ($(\boxname) + (-1.6em,-1.1em)$)
  {\(\bullet\)};
  \node at ($(\boxname) + (-2.2em,-2.2em)$) {\footnotesize\(q_4\)};
  \node[right] at ($(\boxname) + (-1.6em,-2.2em)$)
  {\(\bullet\)};

\def\boxname{squarebox}
  \node[big box, draw=Datatype3]
    at ($(container) + (12em,-7.2em)$) (\boxname)
    {\nodepart{two}
    \nodepart{three}
    \nodepart{four}
    \nodepart{five}};
  \node[square, above of=\boxname, node distance=3.45em] {};
  \node at ($(\boxname) + (-2.2em,2.2em)$) {\footnotesize\(q_0\)};
  \node[right] at ($(\boxname) + (-1.6em,2.2em)$)
  {\(\bullet\)\,\(\bullet\)\,\(\bullet\)\,\(\bullet\)\ \textcolor{gray}{\(\bullet\)}};
  \node at ($(\boxname) + (-2.2em,1.1em)$) {\footnotesize\(q_1\)};
  \node[right] at ($(\boxname) + (-1.6em,1.1em)$)
  {\(\bullet\)};
  \node at ($(\boxname) + (-2.2em,0)$) {\footnotesize\(q_2\)};
  \node[right] at ($(\boxname) + (-1.6em,0)$)
  {\(\bullet\)\,\(\bullet\)};
  \node at ($(\boxname) + (-2.2em,-1.1em)$) {\footnotesize\(q_3\)};
  \node[right] at ($(\boxname) + (-1.6em,-1.1em)$)
  {\(\bullet\)\,\(\bullet\)\,\(\bullet\)};
  \node at ($(\boxname) + (-2.2em,-2.2em)$) {\footnotesize\(q_4\)};
  \node[right] at ($(\boxname) + (-1.6em,-2.2em)$)
  {};
\draw[|-|, semithick, black!70] ($(\boxname) + (-1.4em, -3.7em)$)
    -- node[below=-.2em,font=\footnotesize, color=black!70] {$n=4$}
    ($(\boxname) + (1.4em, -3.7em)$);

\begin{pgfonlayer}{background}
    \fill[ConfigBackgroundFill] ($(container) + (2.65em,6.3em)$)
      -- ($(circlebox) + (-3em, 4em)$)
      -- ($(crossbox) + (2.5em, 4em)$)
      -- ($(squarebox) + (14.5em, -3.2em)$)
      -- ($(squarebox) + (-3em, -3.2em)$)
      -- ($(container) + (2.65em,-6.3em)$)
      -- ($(container) + (2.65em,6.3em)$);
    \draw[dashed, thick, black!50] ($(container) + (2.65em,6.3em)$) -- ($(circlebox) + (-3em, 4em)$)
      -- ($(crossbox) + (2.5em, 4em)$);
    \draw[dashed, thick, black!50] ($(container) + (2.65em,2.15em)$) -- ($(circlebox) + (-3em, -3.2em)$)
      -- ($(crossbox) + (2.5em, -3.2em)$);
    \draw[dashed, thick, black!50] ($(container) + (2.65em,-2.15em)$) -- ($(starbox) + (-3em, -3.2em)$)
      -- ($(starbox) + (14.5em, -3.2em)$);
    \draw[dashed, thick, black!50] ($(container) + (2.65em,-6.3em)$) -- ($(squarebox) + (-3em, -3.2em)$)
      -- ($(squarebox) + (14.5em, -3.2em)$);
    \draw[dashed, thick, black!30] ($(circlebox) + (-3em, 4em)$) -- ($(squarebox) + (-3em, -3.2em)$);
    \draw[dashed, thick, black!50] ($(crossbox) + (2.5em, 4em)$) -- ($(squarebox) + (14.5em, -3.2em)$);
  \end{pgfonlayer}
  \node at ($(squarebox) + (6em, -4.5em)$) {boxes};
\end{tikzpicture}
   \caption{
    How a configuration is mapped to a $(4,2)$-container. Here, the protocol has five states \(q_0, \dots, q_4\). Five distinct data appear in the configuration, and they are represented using symbols.
  }
  \label{fig-boxes}
\end{figure}

The partition of $\configset$ into $(n,M)$-"containers" becomes finer as $n$ and $M$ grow.

\begin{restatable}{lemma}{LemmaContainerMonotonicity}
  \label{lem:metacube-monotonicity}
  Let $n_1, n_2,M_1, M_2 \in \nats$.
  If $n_1 \leq n_2$ and $M_1 \leq M_2$, then every $(n_1, M_1)$-"container" is a union of $(n_2, M_2)$-"containers".
\end{restatable}

Algorithmically, we represent an $n$-"box" as a list of appearing states with associated numbers from $\nset{1}{n}$ encoded in binary.
Similarly, we represent an $(n,M)$-"container" as a list of appearing $n$-"boxes" with associated numbers from $\nset{1}{M}$ encoded in binary.

In fact, "interval predicates" exactly describe finite unions of "containers".

\begin{restatable}{proposition}{PropEquivContPred}
  \label{prop-equiv-containers-predicates}
  The sets of "configurations" defined by "interval predicates" of "height@@pred" at most $n$ and "width@@pred" at most $M$ are exactly the sets formed by unions of $(n,M)$-"containers".
\end{restatable}

\begin{proofsketch}
  For the translation from predicates to "containers", consider a "simple interval predicate" $\dot\exists \datvar_1, \ldots, \datvar_M,\, \bigwedge_{q \in Q} \bigwedge_{j=1}^M \#(q,\datvar_j) \in \nset{A_{q,j}}{B_{q,j}}$ of "height@@pred" $n$.
  This predicate cannot distinguish data mapped to the same $n$-"box", hence cannot distinguish configurations in the same equivalence class for $\metaequiv{n}{M}$, \ie,  $(n,M)$-"containers". The same directly extends to "interval predicates".

  For the other direction, we prove that a given $(n,M)$-container can be expressed as an "interval predicate" of height at most $n$ and width at most $M$. To do so, given a box $\abox \in \Boxes{n}$ and $m \leq M$, we define the "simple interval predicate" $\phi_{\abox, \geq m}$ expressing that at least $m$ data are mapped to box $\abox$. Formally,
  $\phi_{\abox, \geq m} \deff \dot\exists \datvar_1, \ldots, \datvar_{m}, \,\bigwedge_{q \in Q} \bigwedge_{j=1}^{m} \#(q,\datvar_j) \in [A_{q},B_{q}]$,
  where, for all \(q \in Q\), \(A_{q} \deff \abox(q)\),
 \(B_{q} \deff \abox(q)\) if \(\abox(q) < n\) and \(B_{q} \deff +\infty\) if \(\abox(q) = n\).
  This predicate has "height@@pred" at most $n$ and "width@@pred" at most $M$.
  A Boolean combination of such predicates allows us to express an $(n,M)$-container.
  The full proof can be found in \cref{app-equiv-containers-predicates}.
\end{proofsketch}

We therefore have two equivalent representations. Both are useful: "interval predicates" allow us to express properties more naturally, but "containers" are more convenient for the proofs in the remainder of this section.
While they are equally expressive, each can be much more succinct than the other, as stated below.
We refer to Appendix~\ref{app-equiv-containers-predicates} for details.

\begin{restatable}{remark}{ClaimSuccinctness}
	"Containers" can be exponentially more succinct than "interval predicates", while "interval predicates" can be doubly exponentially more succinct than unions of "containers".
\end{restatable}

\subsection{A Translation from Expressions to Containers}
\label{subsec:expressions-to-containers}

Based on the translation from "interval predicates" to "containers" from \cref{prop-equiv-containers-predicates},
we can now show that for all "generalised reachability expressions" $E$ over an "IOPPUD" $\prot$,
the set $\semanticsEP{E}{\prot}$ is a union of $(n,M)$-"containers" with $n$ and $M$ bounded in terms of $E$ and $\prot$.

\begin{restatable}{proposition}{boundCoefsGRE}
  \label{prop:bound-coefs-GRE}
  There is a polynomial function $\poly \colon \nats \to \nats$ such that for all "IOPPUD" $\prot$ and "GRE" $E$,
  the set $\semanticsEP{E}{\prot}$ is a union of $\Bigl(\norm{E} \cdot \bigl(\poly(\size{\prot})\bigr)^{\lengthGRE{E}}, \normGRE{E}^{\poly(\size{\prot}) \cdot \lengthGRE{E}^2}\Bigr)$-"containers".
\end{restatable}

The detailed proof of \cref{prop:bound-coefs-GRE} can be found in \cref{app-bound-coefs-GRE}. We prove the result by structural induction on $E$.
The base case, when $E$ is an "interval predicate", is provided by \cref{prop-equiv-containers-predicates}.
For the induction step, handling Boolean operators is straightforward;
the difficulty lies in operators $\prestarsymb$ and $\poststarsymb$.
This is handled by the following lemma, which relies on the bounds from \cref{sec:bounds-observed-agents}.

\AP Equivalence classes for fixed values of $n$ and $M$ do not behave well with respect to the reachability relation, in the sense that it can happen that $\config_{start} \runto \config_{end}$ and $\config_{start} \equiv_{n,M} \chi_{start}$, but there is no $\chi_{end} \equiv_{n,M} \config_{end}$ such that $\chi_{start} \runto \chi_{end}$. However, this will hold if we take some margin on the equivalence relation of configurations at the start; the following two functions express this margin.
For all $n, M \in \nats$, let $\functionf(n) \deff (n+\size{\prot}^3) \cdot \size{\prot}$ and $\functiong(n,M) \deff \bigl(M+(\size{\prot}^3+1)^{\size{\prot}^3+\size{\prot}^2}\bigr)(n+1)^{\size{\prot}}$.

The following lemma states that, if a set of "configurations" $C$ cannot distinguish $\metaequiv{n}{M}$-equivalent "configurations", then $\prestar{C}$ cannot distinguish $\metaequiv{\functionf(n)}{\functiong(n,M)}$-equivalent "configurations".
In other words, if $C$ is a union of $\metaequiv{n}{M}$-equivalence classes, (\ie, of $(n,M)$-"containers"), then $\prestar{C}$ is a union of $\metaequiv{\functionf(n)}{\functiong(n,M)}$-equivalence classes.

\begin{restatable}{lemma}{LemEqClassesCoReach}
  \label{lem:eq-classes-coreach}
  For all $n, M \in \nats$ and all "configurations" $\config_{start}, \config_{end}, \chi_{start} \in \configset$, if there is a run $\run \colon \config_{start} \runto \config_{end}$
  and $\config_{start} \metaequiv{\functionf(n)}{\functiong(n,M)} \chi_{start}$, then there is a "configuration" $\chi_{end} \in \configset$ with $\config_{end} \metaequiv{n}{M} \chi_{end}$ and a run $\pi \colon \chi_{start} \runto \chi_{end}$.
\end{restatable}
\begin{proofsketch}
	We first apply Corollary~\ref{cor:run-few-observed-data-and-agents} to $\run$, so that we can assume that $\run$ has a limited number of externally observed data and of observed agents per datum.

  In this proof sketch, we first handle the case with only one datum.
  Then, we explain how to generalise this.
  The full proof can be found in \cref{app-eq-classes-coreach}.

  Suppose that all agents in $\config_{start}$ and $\chi_{start}$ share a single datum $d$,
  and suppose $(\config_{start}, d) \cubeequiv{\functionf(n)} (\chi_{start}, d)$.
  Let $\Agentset_{\config}$ and $\Agentset_{\chi}$ be the agents in $\config_{start}$ and $\chi_{start}$, respectively.
For all $q, q' \in Q$, we set $\Agentset^{q\to}_{\config} \deff \setc{a \in \Agentset_{\config}}{\config_{start}(a) = q}$,
	$\Agentset^{q\to}_{\chi} \deff \setc{a \in \Agentset_{\chi}}{\chi_{start}(a) = q}$,
	$\Agentset^{\to q'}_{\config} \deff \setc{a \in \Agentset_{\config}}{\config_{end}(a) = q'}$,
	and $\Agentset^{q \to q'}_{\config} \deff \Agentset^{q\to}_{\config} \cap \Agentset^{\to q'}_{\config}$.

	Our aim is to assign to each agent in $\chi_{start}$ an agent in $\config_{start}$ to mimic. To do so, we construct a mapping $\nu \colon \Agentset_{\chi} \to \Agentset_{\config}$ such that
	\begin{enumerate}[label = (\Alph*)]
		\item
      \label{item:eq-classes-coreach-start}
      for all $a \in \Agentset_{\chi}$, we have $\chi_{start}(a) = \config_{start}\bigl(\nu(a)\bigr)$,

		\item
      \label{item:eq-classes-coreach-observed}
      for all $a' \in \Agentset_{\config}$ "observed" in $\run$, we have $\nu^{-1}(a')\neq \emptyset$, and

		\item
      \label{item:eq-classes-coreach-size}
      for all $q' \in Q$, we have $\size{\nu^{-1}(\Agentset^{\to q'}_{\config,d})} = \size{\Agentset^{\to q'}_{\config,d}}$,
		or $\size{\nu^{-1}(\Agentset^{\to q'}_{\config,d})} \geq n$ and $\size{\Agentset^{\to q'}_{\config,d}}\geq n$.
	\end{enumerate}

  We build $\nu$ separately on each set $\Agentset^{q\to}_{\chi}$ by defining, for each $q \in Q$,
  a mapping $\nu_q \colon \Agentset^{q\to}_{\chi} \to \Agentset^{q\to}_{\config}$.
		Let $q \in Q$. As $(\chi_{start}, d) \equiv_{\functionf(n)} (\config_{start}, d)$, either $\size{\Agentset^{q\to}_{\config}} = \size{\Agentset^{q\to}_{\chi}}$,
		or both $\size{\Agentset^{q\to}_{\config}}$ and $\size{\Agentset^{q\to}_{\chi}}$ are at least $\functionf(n)$.
If $\size{\Agentset^{q\to}_{\config}} = \size{\Agentset^{q\to}_{\chi}}$, we let $\nu_q$ form a bijection between $\Agentset^{q\to}_{\chi}$ and $\Agentset^{q\to}_{\config}$.
Consider now the second case, where $\size{\Agentset^{q\to}_{\config}}$ and $\size{\Agentset^{q\to}_{\chi}}$ are at least $\functionf(n)$.
  We aim at selecting, for every $q' \in Q$, a set $A_{q \to q'} \subseteq \Agentset^{q\to q'}_{\config}$ of agents that must be copied in $\pi$.
  If $\size{\Agentset^{q\to q'}_{\config}} \leq n$, then we let $A_{q \to q'} \deff \Agentset^{q\to q'}_{\config}$.
  Otherwise, we first put in $A_{q \to q'}$ all agents in $\Agentset^{q\to q'}_{\config}$ that are observed in $\run$, at most $\size{\prot}^3$ in total by \cref{cor:run-few-observed-data-and-agents}.
  If $\size{A_{q \to q'}} < n$, we add arbitrary agents from $\Agentset^{q\to q'}_{\config}$ to $A_{q \to q'}$ until $\size{A_{q \to q'}} \geq n$.
  Either way, we have selected $A_{q \to q'}$ of size at most $\size{\prot}^3+n$ for each $q'$, hence at most $\functionf(n)$ agents in total.
  For every $q'$, we have $\size{A_{q \to q'}} \leq \size{\Agentset^{q\to q'}_{\config}}$,
  and either $\size{A_{q \to q'}} = \size{\Agentset^{q\to q'}_{\config}}$ or the two sets have size more than $n$.

  We now build $\nu_q$ such that its image over $\Agentset^{q\to }_{\chi}$ is $\bigcup_{q' \in Q} A_{q\to q'}$.
  We build this in two steps.
  First, we assign to each $\bigcup_{q' \in Q} A_{q\to q'}$ one antecedent by $\nu_q$ in $\Agentset^{q\to }_{\chi}$.
  This is possible because $\size{\bigcup_{q' \in Q} A_{q\to q'}} \leq \functionf(n) \leq \size{\Agentset^{q\to }_{\chi}}$.
  We then identify some $q''$ such that $\size{A_{q\to q''}} > n$ and map all remaining agents of $\Agentset^{q\to }_{\chi}$ to an arbitrary agent in $A_{q\to q''}$.
  Such a $q''$ exists because $\size{\Agentset^{q\to}_{\config}}\geq \functionf(n) \geq n \cdot \size{\prot}$,
  so there is a $q'$ such that $\size{\Agentset^{q\to q'}_{\config}} \geq n$,
  and hence $\size{A_{q\to q'}} \geq n$ by construction.

  This concludes the construction of $\nu$.
  It remains to prove that $\nu$ fulfils \cref{item:eq-classes-coreach-start,item:eq-classes-coreach-observed,item:eq-classes-coreach-size}.
  \cref{item:eq-classes-coreach-start,item:eq-classes-coreach-observed} are immediate from the definition.
  We prove \cref{item:eq-classes-coreach-size}.
  Let $q' \in Q$. We distinguish two cases:
		\begin{itemize}
			\item if $\size{\Agentset^{q \to q'}_{\config}} < n$ for all $q \in Q$, then for all $q$,
        we have $\size{\nu^{-1}(\Agentset^{q \to q'}_{\config})} = \size{\nu^{-1}(A_{q \to q'})} = \size{A_{q \to q'}}= \size{\Agentset^{q \to q'}_{\config}}$,
        so $\size{\nu^{-1}(\Agentset^{\to q'}_{\config})} = \size{\Agentset^{\to q'}_{\config}}$;
			\item if $\size{\Agentset^{q \to q'}_{\config}} \geq n$ for some $q \in Q$, then $\size{A_{q \to q'}} \geq n$,
        so $\size{\nu^{-1}(\Agentset^{q \to q'}_{\config})} \geq n$, and thus both $\size{\nu^{-1}(\Agentset^{\to q'}_{\config})}$ and $\size{\Agentset^{\to q'}_{\config}}$ are at least $n$.
		\end{itemize}

	We construct a "run" $\pi$ from $\chi_{start}$ by copying $\run$ as follows. For each step of $\run$ where an agent $a$ performs some transition $t$, we make $\size{\nu^{-1}(a)}$ steps in $\pi$ so that all agents in $\nu^{-1}(a)$ perform transition $t$ one by one.
	If $a$ observed some agent $a'$, there is $a''$ in $\pi$ that can be observed because $\nu^{-1}(a') \ne \emptyset$: we made sure to map an agent to each "observed" agent in $\run$.

	For the general case with more data,
  we similarly construct two mappings $\mu$ and $\nu$. First we define $\mu$, which maps each datum $d$ of $\chi_{start}$ to one of $\config_{start}$ such that $(\config_{start}, \mu(d))\cubeequiv{\functionf(n)} (\chi_{start}, d)$.
  Then, for each datum $d$, $\nu$ maps each agent $a$ with datum $d$ of $\chi_{start}$ to one with datum $\mu(d)$ of $\config_{start}$.

  Once \(\mu\) and \(\nu\) are defined, we build a run from $\chi_{start}$ to a "configuration" $\chi_{end}$ in which each agent $a$ mimics the behaviour of $\nu(a)$ in $\run$.
  We make sure that agents (resp. data) "observed" in $\run$ have agents (resp. data) mapped to them, so that we can take the same transitions in $\run$ and $\pi$.
  The construction of $\nu$ ensures that, for all data $d$, we have $(\config_{end}, \mu(d)) \cubeequiv{\functionf(n)} (\chi_{end}, d)$. The construction of $\mu$ ensures that $\config_{end} \metaequiv{\functionf(n)}{\functiong(n,M)} \chi_{end}$.
\end{proofsketch}
 
	\section{Decidability and Complexity Bounds}
\label{sec:complexity}

\subsection{Decidability in Exponential Space}
In this section, we use the results on "GRE" from \cref{sec:equivalence-relation} to provide an \exps upper bound for the "emptiness problem for GRE".
In the following, we assume that the representation of a "GRE" $E$ takes $\size{E} + \log(\norm{E})$ space. 

We first prove that we can decide membership of a "configuration" (encoded in a naive way) in a "GRE" in \pspace.
A "configuration" is represented ""data-explicitly"" if it is represented as a list of vectors of $\nats^Q$, one vector for each datum.
The ""size@@dataexplicit"" of this representation is $k \cdot\size{\prot}\cdot log(m)$ where $k$
is the number of data and $m$
is the number of agents "appearing" in $\config$.

\begin{restatable}{proposition}{ConfigInGREPspace}
  \label{prop:config-in-GRE-pspace}
  The following problem is decidable in \pspace:
  given a "PPUD" $\mathcal{P}$, a "GRE" $E$, and a "configuration" $\config$ described "data-explicitly",
  decide if $\config \in \semanticsEP{E}{\prot}$.
\end{restatable}

The proof is given in Appendix~\ref{app-complexity-upper-bounds}.
It uses a relatively straightforward induction on $E$ to show that this problem can be decided in polynomial space using a recursive algorithm (with a polynomial whose degree does not depend on $E$).
For the case where $E = \poststar{F}$, we rely on the fact that the numbers of agents and data remain the same throughout a run; we therefore can guess the configuration $\config'$ such that $\config' \in \semanticsEP{F}{\prot}$ (which can be checked with a recursive call) and $\config \runto \config'$ (which can be checked by exploration of the graph containing configurations with as many agents and data as $\config$). The case $E = \prestar{F}$ is similar.

\cref{prop:config-in-GRE-pspace} allows us to check if a given configuration of a "PPUD" is in the set described by a "GRE"\footnote{This implies that the "emptiness problem for GRE" over "PPUD", while undecidable due to \cref{undec-well-spec}, is semi-decidable: one can simply enumerate all configurations and test membership for each of them.}.
In the case of "IOPPUD", \cref{prop:bound-coefs-GRE} allows us to search for a witness configuration within some bounded set, yielding decidability.

\TheoremEXPSPACE*

\begin{proof}
  Suppose $\semanticsEP{E}{\prot}$ is not empty.
  By \cref{prop:bound-coefs-GRE}, it contains an $(n,M)$-"container" $\acontainer$ with $n \deff \norm{E} \cdot \poly(\size{\prot})^{\size{E}}$ and $M \deff \norm{E}^{\poly(\size{\prot} ) \cdot \size{E}^2}$.
  We construct a configuration $\config \in \acontainer$ as follows.
  For each $n$-box $\abox$, we select $\acontainer(\abox)$ many data such that over all $n$-boxes, the selected data are pairwise distinct.
  Then, for each $n$-box $\abox$, each state $q \in Q$ of $\prot$, and each datum $d_{\abox}$ selected for $\abox$, we put $\abox(q)$ many agents with datum $d_{\abox}$ in $q$.
Note that the configuration $\config$ is in $\acontainer$, and the number of agents it contains it at most $n \cdot \size{\prot} \cdot \size{\Boxes{n}} \cdot M$. We have $\size{\Boxes{n}} = (n+1)^{\size{\prot}} = \norm{E} \cdot \poly(\size{\prot})^{\size{E}\size{\prot}}$. We assumed at the beginning of Section~\ref{sec:complexity} that the encoding of $E$ uses memory $\size{E} + \logarithm{\norm{E}}$.
  As a result, $n$, $M$, $\size{\prot}$ and $\size{\Boxes{n}}$ are all at most exponential in the size of the input.
  Therefore, if $\semanticsEP{E}{\prot}$ is not empty, then it contains a "configuration" with at most exponentially many agents.
  We can guess the "data-explicit" description of such a "configuration" in non-deterministic exponential space, and then check that the guessed configuration is in $\semanticsEP{E}{\prot}$ in exponential space by \cref{prop:config-in-GRE-pspace} (we apply the \pspace algorithm on an exponential input).
  As a result, deciding emptiness of $\semanticsEP{E}{\prot}$ is in \nexps, which is identical with \exps.
\end{proof}

\subsection{A Lower Complexity Bound}
We now provide the following lower complexity bound.
\begin{theorem}
\label{nexptime-hardness}
The "emptiness problem for GRE" over "IOPPUD" is \conexpt-hard.
\end{theorem}
\begin{proof}[Proof sketch]
We proceed by reduction from the problem of tiling an exponentially large grid,
a \nexpt-complete problem \cite{tiling_problems},
to the complement of the "emptiness problem for GRE".
The full proof can be found in \cref{app:lowerbound}.

\AP A \emph{tiling instance} is a tuple \((2^n, \tcolors, \tiles)\),
with \(n \geq 1\), \(\tcolors\) a finite set of \emph{colours} with special colour \(\twhite\),
and \(\tiles  = \set{t_1, \dots, t_m} \subseteq \tcolors^4\) a finite set of \emph{tiles}.
We can view a tile as a square whose four edges are coloured.
The \emph{tiling problem} asks whether there is a \emph{tiling}, that is,
a mapping \(\tiling \colon \nset{0}{2^n{-}1} \times \nset{0}{2^n{-}1} \to \tiles\)
such that the colours of neighbouring tiles match
and the borders of the grid are white.

Given a tiling instance \((2^n, \tcolors, \tiles)\),
we build an instance \((\prot, E)\) of the "emptiness problem for GRE".
In \(\prot\), witness tilings can be encoded in the configurations,
and we construct \((\prot, E)\) such that \(\setof{E}\) contains exactly the configurations
that correspond to a correctly encoded witness tiling.
More precisely, $\config \in \setof{E}$ when:
\begin{enumerate}[label=\textup{\textsf{(Cond}\textsf{\arabic*}\textsf{)}},ref=\textsf{(Cond}\textsf{\arabic*}\textsf{)}, leftmargin=50pt]
\item \label{ingr:datum-tile} for all $(i,j) \in \nset{0}{2^n-1}^2$, some datum encodes coordinates $(i,j)$ and a tile type;
\item \label{ingr:no-dup} for all $(i,j) \in \nset{0}{2^n-1}^2$, there is at most one datum encoding $(i,j)$; 
\item \label{ingr:horizontal-verif} the mapping $\nset{0}{2^n-1}^2 \to \tcolors$ defined by the data is a tiling.  
\end{enumerate}

The GRE $E$ will be of the form of a conjunction, \ie, a list of \emph{constraints} that the configuration must satisfy.  
Our first constraint is $\setcomplement{\prestar{\presenceformula{\cheatstate}}}$ where $\cheatstate$ is a special error state and $\presenceformula{q}$ is the GRE expressing that some agent is in $q$. This forbids, in $\setof{E}$, configurations from which $\cheatstate$ can be covered. 

\begin{figure}
\begin{tikzpicture}[auto, xscale = 0.6, yscale = 0.8]
\tikzset{every node/.style = {font = {\small}, align = center}}
\tikzset{every state/.style = {font = {\scriptsize}, inner sep = 1pt, minimum size = 15pt}}
\draw[rounded corners=2mm,dashed,fill=black!5] (-3,-1) -| (2.1,-4.3) -| (-3,-4.3) -- cycle;
\node[font = {\scriptsize}] at (-0.5, -1.2) {Horizontal coordinate} ;
\node[state] at (-0.5,-2) (b00h) {$h_1(0)$};
\node[state] at (1.2,-2) (b01h) {$h_1(1)$};
\node at (0.35, -2.75) {\Large $\vdots$};
\node[state] at (-0.5,-3.5) (bn0h) {$h_n(0)$};
\node[state] at (1.2,-3.5) (bn1h) {$h_n(1)$};
\node at (-2.2, -2) {bit $1$};
\node at (-2.2, -3.5) {bit $n$};

\node[align = center, font = {\small}] at (-0.45, 2) {\emph{Tiling gadget}}; 
\node[align = center, font = {\small}] at (5, 2) {\emph{Duplication gadget}}; 
\node[align = center, font = {\small}] at (13.5, 2.5) {\emph{Horizontal verifier}}; 
\node[align = center, font = {\small}] at (13.5, -4.2) {\emph{Synchronization gadget}}; 

\begin{scope}[yshift = -3.5cm]
\draw[rounded corners=2mm,dashed,fill=black!5] (-3,-1) -| (2.1,-4.3) -| (-3,-4.3) -- cycle;
\node[font = {\scriptsize}] at (-0.5, -1.2) {Vertical coordinate} ;
\node[state] at (-0.5,-2) (b00h) {$v_1(0)$};
\node[state] at (1.2,-2) (b01h) {$v_1(1)$};
\node at (0.35, -2.75) {\Large $\vdots$};
\node[state] at (-0.5,-3.5) (bn0h) {$v_n(0)$};
\node[state] at (1.2,-3.5) (bn1h) {$v_n(1)$};
\node at (-2.2, -2) {bit $1$};
\node at (-2.2, -3.5) {bit $n$};
\end{scope}

\begin{scope}[yshift = 3.2cm]
\draw[rounded corners=2mm,dashed,fill=black!5] (-3,-2) -| (2.1,-4) -| (-3,-4) -- cycle;
\node[align = center, font = {\scriptsize}] at (-0.5,-2.5) {Tile type choice};
\node[state, minimum size = 20pt] at (-2,-3.4)  {$\statet{1}$};
\node at (-0.45, -3.4) {$\dots$};
\node[state, inner sep = 0pt, minimum size = 20pt] at (1.1,-3.4) {$\statet{m}$};

\end{scope}

\begin{scope}[yshift = 0.5cm]
\draw[rounded corners=2mm,dashed,fill=IoppudBlue!30] (2.7,-3) -| (7.3,-7.5) -| (2.7,-7.5) -- cycle;
\node[state] at (3.5,-8) (entry1) {$D_a$};
\node[state] at (6.5, -8) (entry2) {$D_b$};
\node[state] at (5,-8) (sinkstate) {$\sinkstate$};
\node[state] at (3.5, -7) (duptrack1) {};
\node[state] at (6.5, -7) (duptrack2) {};
\node[state] at (5, -3.5) (qdup) {$\dupstate$};
\node[align =center, font = {\scriptsize}] at (5,-5.25) (blabla) {Two agents \\ test whether \\ their data \\ are dinstinct and \\ encode the same \\ coordinates};

\path[-stealth] 
(entry1) edge (duptrack1)
(entry2) edge (duptrack2) 
(entry1) edge (sinkstate)
(entry2) edge (sinkstate)
(duptrack1) edge[bend left = 50, dashed] (qdup)
(duptrack2) edge[bend right = 50, dashed] (qdup)
;

\draw[rounded corners=2mm,dashed,fill=IoppudMagenta!20] (2.7,-2.5) -| (7.3,0.7) -| (2.7,-2.5) -- cycle;

\node[state] at (5,0) (over3) {$\stateover{3}$};
\node[state] at (5,-1) (over2) {$\stateover{2}$};
\node[state] at (5,-2) (over1) {$\stateover{1}$};
\draw[-stealth, color = black] (3,-3) -- (over1);
\draw[-stealth, color = black] (4,-3) -- (over1);
\draw[-stealth, color = black] (5,-3) -- (over1);
\draw[-stealth, color = black] (6,-3) -- (over1);
\draw[-stealth, color = black] (7,-3) -- (over1);
\path[-stealth] 
(over1) edge[bend left= 30] node[left] {$\stateover{1}, =$} (over2)
(over1) edge[bend right= 30] node[right] {$\stateover{1}, \ne$} (over2)
(over2) edge[bend left= 30] node[left] {$\stateover{2}, =$} (over3)
(over2) edge[bend right= 30] node[right] {$\stateover{2},\ne$} (over3);
\end{scope}

\tikzset{every node/.style = {font = {\scriptsize}, align = center}}

\begin{scope}[yshift = -6.2cm, xshift = -0.5cm]
\node[state] (S) at (9,0) {$S^{(h)}$};
\node[state] (S2) at (12,0) {};
\node[state, minimum size = 15pt, inner sep = 0pt] (T1) at (14,1.2) {$T^{(h)}_1$};
\node[state, minimum size = 15pt, inner sep = 0pt] (T2) at (14,0) {$T^{(h)}_2$};
\node[font = {\tiny}] (Tsusp) at (14, -0.6) {$\dots$};
\node[state, minimum size = 15pt, inner sep = 0pt] (TT) at (14, -1.2) {$T^{(h)}_{m}$};
\node[state] (F2) at (16,0) {};
\node[state] (F) at (19,0) {$F^{(h)}$};

\path[-stealth]
(S) edge[dashed] node[above, align = center, font = {\scriptsize}] {first \\ sync. test} (S2)
(S2) edge node[above, sloped, rotate = 10] {$\statet{1}, =$} (T1)
(S2) edge node[above, sloped, yshift = -3pt] {$\statet{2}, =$} (T2)
(S2) edge node[below, sloped, rotate = -10] {$\statet{m}, =$} (TT)
(T1) edge (F2)
(T2) edge (F2)
(TT) edge (F2)
(F2) edge[dashed] node[above, align = center, font = {\scriptsize}] {second \\ sync. test} (F)
;
\end{scope}

\tikzset{every node/.style = {font = {\scriptsize}, align = center}}

\begin{scope}[xshift = 14.5cm, yshift = 1cm]
\node[state, inner sep = 1pt] at (-1,0) (start) {$q_i^{(h)}$};
\draw[-stealth] (-1,1) -- (start);
\node at (0.5,0.7) {$\colorvar \assignalgo \twhite$ \\ $x,y \assignalgo 0$};
\node[state] (getT) at (-1,-2) {};
\node[state] at (4,-2) (hsmallvsmallok) {};
\node[state] at (-6,-2) (hbigvsmallok) {};
\node[state] at (-1,-4) (hbigvbig) {};
\node[state, inner sep = 1pt] at (4,-4) (qf) {$q_f^{(h)}$};

\path[-stealth]
(start) edge node[left, xshift = 0.73cm, yshift = 0.05cm] 
{$T \assignalgo T_i$ \\ where $T_i^{(h)}$ \\ populated} (getT)
(getT) edge node[left, yshift = 0pt, xshift = 30pt] {if $x {=} 2^n{-}1$ \\ and $y {=} 2^n{-}1$ \hspace{0.2cm}} (hbigvbig)
(getT) edge node[below, xshift = 10pt] {if $x {<} 2^n{-}1$ \\ and $y {<} 2^n{-}1$, \\ check that \\ $\leftcolor{T} {=} \colorvar$} (hsmallvsmallok)
(getT) edge node[below, xshift = -10pt] {if $x {=} 2^n{-}1$ \\ and $y {<} 2^n{-}1$ \\ check that \\ $\leftcolor{T} {=} \colorvar$ \\and $\rightcolor{T} {=} \twhite$} (hbigvsmallok)
(hbigvbig) edge node[below] {check that \\$\rightcolor{T} {=} \twhite$} (qf)
(hbigvsmallok) edge[bend left = 15] node[above, xshift = -15pt] {$\colorvar \assignalgo \twhite$ \\ $x \assignalgo 0$ \\ $y \assignalgo y+1$} (start)
(hsmallvsmallok) edge[bend right = 15] node[above, xshift = 15pt] {$\colorvar \assignalgo \rightcolor{T}$ \\ $x \assignalgo x+1$} (start)
;

\end{scope}

\end{tikzpicture}
 \caption{Partial depiction of the protocol constructed in \cref{nexptime-hardness}.}
\label{fig:lowerbound-summary}
\end{figure}
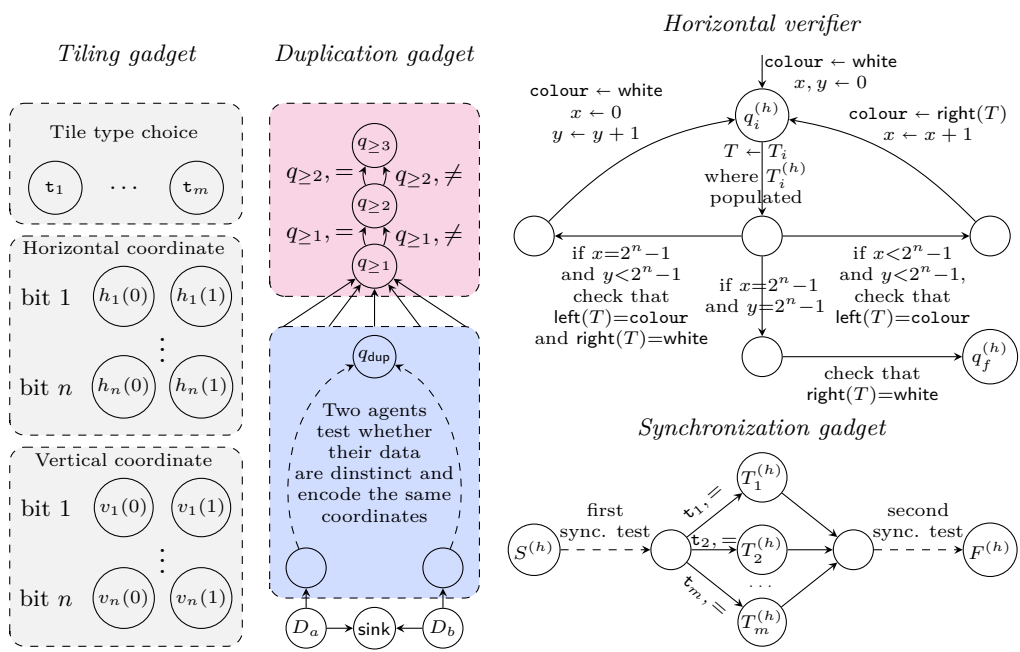

\ref{ingr:datum-tile} is obtained using the \emph{tiling gadget} in \cref{fig:lowerbound-summary}. 
States \(\statet{1}, \dots, \statet{m}\) represent the available tiles of $\tiles$, and coordinate states allow for a binary representation of the horizontal and vertical coordinates
of a square in the grid.
For a datum \(d\), the agents of datum \(d\) in the coordinate states encode the position of the square corresponding to \(d\),
and an agent of datum \(d\) in state \(\statet{i}\) indicates that the square in the grid corresponding to \(d\) should be coloured according to tile \(t_i\).  Configurations in $\setof{E}$ are not allowed to have two agents of same datum playing the same role; otherwise, one of them may observe the other and go to $\cheatstate$. In particular, each datum has at most $2n+1$ agents in the tiling gadget. 

To obtain \ref{ingr:no-dup}, we use a \emph{duplication gadget}, partially represented in \cref{fig:lowerbound-summary}. 
We enforce that any configuration in $\setof{E}$ has one agent of each datum in $D_a$, one in $D_b$ and none in the rest of the duplication gadget.
The blue part implements a test (depicted in \cref{fig:no_duplication} in \cref{app:lowerbound}) where two agents of distinct data, one from $D_a$ and one from $D_b$, may test that their data encode the same coordinates; if this is the case, they may go to $\dupstate$. If there are more than two agents in the blue part, this test is not reliable but $\stateover{3}$ can be covered. 
\cref{ingr:no-dup} can therefore be achieved by enforcing that configurations in $\setof{E}$ are \emph{not} in $\prestar{\presenceformula{\dupstate} \cap \setcomplement{\prestar{\presenceformula{\stateover{3}}}}}$.

Finally, we explain how \ref{ingr:horizontal-verif} is achieved; we describe only how the horizontal (left-right) borders are verified. We use a gadget, named \emph{horizontal verifier} in \cref{fig:lowerbound-summary}. In this gadget, a single agent, called \emph{verifier}, is in charge of verifying that colours of left-right borders match. The verifier uses $2n$ auxiliary agents to encode two variables $x,y \in \nset{0}{2^n-1}$ in binary. Again, transitions to $\cheatstate$ detect when two agents play the same role, so that there is only one verifier and so that variables $x$ and $y$ can be implemented faithfully. The initialisation  $x =y = 0$ is enforced as a constraint in $E$. We now sketch how the verifier reads the encoded tiling; to do that, it must synchronise with the datum encoding $(x,y)$.  

 This is done using the synchronisation gadget of \cref{fig:lowerbound-summary}. In $\setof{E}$, all agents in the synchronisation gadget are in $S^{(h)}$. Moreover, we add a constraint in $E$ so that $\config \in \setof{E}$ requires that there is a run from $\config$ where all agents in the synchronisation gadget end in $F^{(h)}$ and where the verifier ends in $q_f^{(h)}$.  
 The synchronisation tests guarantee that, whenever there is an agent in $T_i^{(h)}$, this agent's datum encodes square $(i,j)$ where $i$ is equal to the current value of $x$ and $j$ is equal to the current value of $y$. The synchronisation is challenging to design because the values of $x$ and $y$ may change throughout a run and only one bit can be tested at a time. However, as proved in \cref{sync-test} of \cref{app:lowerbound}, this can be achieved by having a first synchronisation test that checks equality of bits from most to least significant, and a second test that checks equality from least to most significant.
\end{proof}

\subsection{Discussion on Complexity Gaps}
We now discuss some complexity gaps left open by this paper.
First, there remains a complexity gap for the "emptiness problem for GRE", which is known to be between \conexpt (\cref{nexptime-hardness}) and \exps (\cref{thm:main-ioppud}). Closing the gap appears challenging. On one hand, if the problem is below \exps, then this probably requires developing new techniques. On the other hand, proving \exps-hardness does not seem easy. In particular, the synchronisation techniques from \cref{nexptime-hardness} assumes that each datum synchronises only once with the verifier. This synchronisation technique would not suitable for, \eg, multiple interactions between the head and the cells of a Turing machine.

Another, arguably more important open question is the exact complexity of "well-specification", which is only known to be between \pspace (model without data, \cite{EsparzaRW2019}) and \exps (\cref{thm:main-ioppud}). On the one hand, it is unclear whether relevant configurations can be stored in polynomial space:
\begin{claim}
\label{well-spec-exp}
The number of data to consider for "well-specification" may have to be exponential.
\end{claim}
The claim is formalised and proven in \cref{app:proof-well-spec-exp}. As a consequence, proving that the problem is in \pspace cannot be achieved with a procedure that explicitly stores configurations. 
On the other hand, in order to build a reduction from the tiling problem as in \cref{nexptime-hardness}, we need a new idea to enforce that \emph{at most} one datum encodes each tile. In \cref{nexptime-hardness}, we had states $\dupstate$ and $\stateover{3}$ and duplication meant being able to cover $\dupstate$ and, at the same, forbid that $\stateover{3}$ can ever be covered in the future. We do not know how to encode this constraint when working with an instance of well-specification. 
 
	\section{Conclusion}
	\label{sec:conclusion}
	We have studied the verification of "population protocols with unordered data" \cite{BL23}, an extension of population protocols where agents carry data from an infinite unordered set. We first proved that the "well-specification" problem is undecidable (\cref{undec-well-spec}), which then led us to consider the restriction to protocols with immediate observation. This subclass was defined in \cite{BL23}, where the authors proved that these protocols  compute exactly the "interval predicates". We defined a general class of problems on this model, which consists in deciding the existence of a configuration satisfying a so-called "generalised reachability expression"; this class of problems subsumes many classic problems, one of which is well-specification. Despite its generality, we showed the problem to be decidable in exponential space (\cref{thm:main-ioppud}); we also provided a \conexpt lower bound. A remaining open question is the exact complexity of "well-specification" for immediate observation population protocols with unordered data, which is located between \pspace (model without data, \cite{EsparzaRW2019}) and \exps (\cref{thm:main-ioppud}).

	\bibliography{biblio}

	\newpage
	\appendix

\section{Undecidability of Verification of Population Protocols with Unordered Data: Proof for Section \ref{sec:undecidability}}
\label{app:undec}

	\PropositionVerificationUndecidable*

A "\(2\)-counter machine" is a transition system using two counters \(x\) and \(y\).
  It consists of a list of instructions of the form: \(\increment(c)\), increment on counter \(c\); \(\decrement(c)\), decrement on counter \(c\); \(\zerotest(c,k)\), zero-test that moves to instruction $k$ if $c = 0$ (and to the next instruction otherwise); and \(\Halt\), halting instruction. The machine starts at the first instruction with both counters equal to \(0\). It is well known  that deciding whether a "$2$-counter machine" eventually reaches the $\Halt$ instruction is an undecidable problem.

	Let \(\mathsf{CM}\) be a "\(2\)-counter machine" with \(n\) instructions \(i_1, \dots, i_n\).

	Let \(Q_{CM}\coloneqq\{i_1,\dots, i_n, i_1', \dots, i_n'\}\). We define a "population protocol with unordered data" \(\Prot\) over infinite domain \(\Dataset\) as follows:

	\begin{enumerate}
		\item \(Q_{op} \coloneqq \{\idle, \increment, \decrement, \done, \zerotest, =0,>0\}\)
		\item \(Q_{main} \coloneqq Q_{CM} \cup \{\counterstate{x}, \counterstate{y}\} \cup \{\countercontrol{x}, \countercontrol{y}\} \times Q_{op} \cup \{R,\Uniq, \cheatstate\}\).
		\item \(Q \coloneqq Q_{main} \times \{R,\other\}\).
		\item \(\delta\) will be defined step by step later.
		\item \(I=\{(R,R),(\Uniq, \other), ((\countercontrol{x}, \idle), \other), ((\countercontrol{y}, \idle), \other), (i_1, \other)\}\).
		\item \(O(i_m, \other)=O(i_m, R)=\top\) if \(i_m=halt\), else \(O(q)=\bot\) for all other \(q \in Q\).
	\end{enumerate}

	The actual states of the protocol are \(Q\), but the second component is never updated; it simply remembers whether the initial state was \(R\) or not. We will hence mainly refer to \(Q_{main}\) and specify all except one transition in terms of \(Q_{main}\) only.

	Recall the role of each state:
	\begin{itemize}
		\item $R$ is a reservoir, which can contain many agents of all data.

		\item $\Uniq$ is a state in which all agents should have different data

		\item $Q_{CM}$ is the set of control states in which a single agent should evolve.

		\item for both $c \in \set{x,y}$, $\countercontrol{c} \times Q_{op}$ must contain a single agent at all times, with a datum that is not in $U$, which is in charge of synchronising with the agent in $Q_{CM}$ and apply the transitions of the machine to counter $\counterstate{c}$.

		\item $\counterstate{x}, \counterstate{y}$ represent the counters: all agents in each $\counterstate{c}$ should have the same datum as the current agent in $\countercontrol{c} \times Q_{op}$, and the number of agents in $\counterstate{c}$ represents the current counter value.
	\end{itemize}
We will ensure that those conditions are met using the "violation detection" mechanisms below.

	 We use the following notation: For \(p,p',q,q' \in Q\) we denote \(((p,p'), \sometest, (q,q')) \in \delta\) by \(p,p' \mapsto_{\sometest} q, q'\), or leave away \(\sometest\) if this transition is enabled in either case. The first transition we present allows us to guarantee that at most one agent per datum is not in $R$ initially.

	\begin{align}
		\tag{Input Violation} \label{tra:InputViolation}
		&(p, \other), (p', \other) \mapsto_{=} (\cheatstate, \other), (\cheatstate, \other) & p, p' \in Q_{main}
	\end{align}

	This transition detects that a "violation" occurred; multiple agents of the same datum started in \(\Uniq\) or in the control states. It overwrites whatever other transition would be defined between the states in \(Q_{main}\).
	It is important that it is also impossible for one of the counter control agents in $\set{\countercontrol{x}, \countercontrol{y}}\times Q_{op}$ to initially have a datum which is still in the pool \(\Uniq\), otherwise one cheat, \ie\ an incorrectly performed zero-test, would be undetectable at the start.

	We continue by specifying the other violation transitions; namely if two agents are in \(Q_{CM}\), if two agents are in \(\countercontrol{x}\), or if an agent in \(\countercontrol{x}\) meets an agent in \(\counterstate{x}\) of different datum:

	\begin{align}
		\tag{Counter Colour Violation} \label{tra:CounterColorViolation}
		&(\countercontrol{c}, b), \counterstate{c} \mapsto_{\neq} \cheatstate, \cheatstate & b \in Q_{op}, c \in \{x,y\} \\
		\tag{Counter Control Violation} \label{tra:ControlStateViolation1}
		&(\countercontrol{c}, b), (\countercontrol{c}, b') \mapsto \cheatstate, \cheatstate & b, b' \in Q_{op}, c \in \{x,y\} \\
		\tag{Control State Violation} \label{tra:ControlStateViolation2}
		&q, q' \mapsto \cheatstate, \cheatstate & q,q' \in Q_{CM} \\
		\tag{Convert To Sink} \label{tra:ConvertToSink}
		&\cheatstate, q \mapsto \cheatstate, \cheatstate & q \in Q_{main}
	\end{align}

	The \eqref{tra:ConvertToSink} transition informs other agents about violations.

	Next we explain the actual counter machine simulation. Increments and Decrements are performed via a sequence of three transitions each as follows.

	\begin{align}
		\tag{Start Increment c} \label{tra:StartIncreC}
		& i_m, (\countercontrol{c}, \idle) \mapsto i_m', (\countercontrol{c}, \increment) & c\in \{x,y\}, i_m=\increment(c) \\
		\tag{Increment c} \label{tra:IncreC}
		& (\countercontrol{c}, \increment), R \mapsto_{=} (\countercontrol{c}, \done), \counterstate{c} & c \in \{x,y\}\\
		\tag{End Operation on c} \label{tra:EndOperationC}
		& (\countercontrol{c}, \done), i_m' \mapsto (\countercontrol{c}, \idle), i_{m+1} & c \in \{x,y\}, m \in \{1,\dots, n\} \\
		\tag{Start Decrement c} \label{tra:StartDecreC}
		& i_m, (\countercontrol{c}, \idle) \mapsto i_m', (\countercontrol{c}, \decrement) & c\in \{x,y\}, i_m=\decrement(c) \\
		\tag{Decrement c} \label{tra:DecreC}
		& (\countercontrol{c}, \decrement), \counterstate{c} \mapsto_{=} (\countercontrol{c}, \done), R & c\in \{x,y\}
	\end{align}

	To simulate zero test instructions, the counter control agent in \(\countercontrol{c}\) is informed about the instruction; afterwards they either meet an agent in \(\counterstate{c}\) of their own datum, in which case they decide \(>0\), or an agent in \(\Uniq\), in which case this new agent takes their place and the counter is assumed to be \(0\).

	\begin{align}
		\tag{Zerotest Start} \label{tra:decrementC}
		& i_m, (\countercontrol{c}, \idle) \mapsto i_m', (\countercontrol{c}, \zerotest) & i_m=\zerotest(c,k), c\in \{x,y\}, k\in \N \\
		\tag{c=0}
		& (\countercontrol{c}, \zerotest), \Uniq \mapsto R, (\countercontrol{c}, =0) & c\in \{x,y\} \\
		\tag{c>0}
		& (\countercontrol{c}, \zerotest), \counterstate{c} \mapsto (\countercontrol{c}, >0), \counterstate{c} & c\in \{x,y\} \\
		\tag{End Zerotest =0}
		& i_m', (\countercontrol{c}, =0) \mapsto i_k, (\countercontrol{c}, \idle) & i_m=\zerotest(c,k), c\in \{x,y\}, k \in \N \\
		\tag{End Zerotest >0}
		& i_m', (\countercontrol{c}, >0) \mapsto i_{m+1}, (\countercontrol{c}, \idle) & i_m=\zerotest(c,k), c \in \{x,y\}, k \in \N
	\end{align}

	\textbf{Correctness}: First assume that the counter machine \(\mathsf{CM}\) does halt.

	Let \(k\) be the number of steps \(\mathsf{CM}\) requires to halt. Let \(\config_0\) be an initial configuration with \(\config_0(d_i,R)\geq k\) and \(\config_0(d_i, \Uniq)=1\) for at least \(k\) different data \(d_1, \dots, d_k\). Moreover, assume that in \(\config_0\) there is no input violation, \ie\ all transitions with ``Violation'' as part of the name are disabled. In particular, \(\config_0(\Dataset, i_1)=\config_0(\Dataset, \countercontrol{x})=\config_0(\Dataset, \countercontrol{y})=1\), \ie\ we do not have a violation in the counter or the control states.

	We perform the following transition sequence \(\sigma\): Using the \(k\) data above, we correctly simulate the counter machine until it halts. Observe that any simulating transition takes at most \(1\) agent out of \(R\) and at most one new datum from \(\Uniq\), hence we do not run out of agents in \(R, \Uniq\). Call the reached configuration \(\config\). This configuration is clearly a deadlock, since we do not have any violations to detect, and the control agent in \(Q_{CM}\) does not start any new operation. However, some agent is in the halt instruction, \ie\ we indeed have a fair (in fact terminal) run which does not output \(\bot\).

	Now assume that \(\mathsf{CM}\) does \emph{not halt}. We start with three important observations, which will be used to prove that violation transitions can only be disabled by occurring:

	\begin{enumerate}
		\item \(\config(Q_{CM},\Dataset)\), \(\config(\{\countercontrol{x}\} \times Q_{op}, \Dataset)\) and \(\config(\{\countercontrol{y}\} \times Q_{op}, \Dataset)\) are preserved by all transitions except transitions moving an agent to the sink state.
		\item Agents who enter \(\{\countercontrol{x}, \countercontrol{y}\} \times Q_{op}\) are always from \(\Uniq\), a previously unused datum.
		\item Agents who enter or leave \(\counterstate{x},\counterstate{y}\) are always the same datum as the corresponding \(\countercontrol{x}, \countercontrol{y}\).
	\end{enumerate}

	We have to prove that every initial configuration has a unique output. Let \(\config_0\) be any initial configuration. We claim that every fair run \(\pi=(\config_0,\config_1, \dots)\) "stabilises" to output \(\bot\). Assume that in \(\pi\) no violation occurs, \ie \(\config_m(\cheatstate)=0\) for all \(m \in \N\), otherwise \(\pi\) has output \(\bot\) via transition \eqref{tra:ConvertToSink}, since agents cannot leave the sink.

	Initially all agents not starting in \(R\) have a different datum, otherwise \eqref{tra:InputViolation} would eventually occur. Furthermore, by observation 1, \(Q_{CM}, (\countercontrol{x}, \idle)\) and \((\countercontrol{y},\idle)\) all start with one agent each, otherwise the corresponding violation transition would eventually occur by fairness.  We claim that in \(\pi\) the counter machine is simulated faithfully. Assume the opposite. The only way to not simulate the counter machine correctly is by performing a zero-test wrong. That is, at some \(\config_m\) the goto \(k\) part of a \(\zerotest(c,k)\) instruction was applied with \(\config_m(\counterstate{c},\adatum)>0\) for some \(\adatum \in \Dataset\).

	However, by observation 2, the unique agent with a state in the set \(\{\countercontrol{c}\} \times Q_{op}\) will always have a data \(\adatum' \neq \adatum\) for the rest of the run. By observation 3, the agent in state \(\counterstate{c}\) with data \(\adatum\) can hence never leave, \ie\ we have \(\config_l(\counterstate{c},\adatum)>0\) for all \(l \geq m\). This implies that \eqref{tra:CounterColorViolation} is always enabled, and would eventually occur, contradiction.

	Hence in \(\pi\) the counter machine is faithfully simulated. Since \(\mathsf{CM}\) does not halt, we have \(\config_l(i_k, \Dataset)=0\) for all \(l \in \N\) and halt instructions \(i_k\). Since this is the only state with output \(\top\), every agent always has output \(\bot\), and the run \(\pi\) hence has output \(\bot\) as required, concluding the proof.

As a final remark, this reduction in fact also establishes that the following problem is undecidable: given two interval predicates \(\varphi_1, \varphi_2\) and a "PPUD" \(\prot\), decide whether the "GRE" \(\poststar{\varphi_1} \cap \varphi_2\) is empty. To see this, instead of using "violation detection", use \(\varphi_1\) to encode the restrictions on "initial configurations" and \(\varphi_2\) to encode that at the end, no agent in \(\counterstate{x}\) or \(\counterstate{y}\) is supposed to have a different datum than the corresponding agent in \(\countercontrol{x}\) or \(\countercontrol{y}\).

 	\section{An Analysis of Immediate Observation Protocols with Data: Proofs for Section~\ref{sec:bounds-observed-agents}}
\label{app:bounds-observed-agents}

We recall a few notations: 
Let $\run : \config_1 \to \config_2 \to \cdots \to \config_{m}$ be a "run".
For $i \in [1,m]$, let $\prefixrun{\run}{i}$ (resp.\ $\suffixrun{\run}{i}$) denote the prefix of $\run$ ending on its $i$-th "configuration" (resp.\ the suffix of $\run$ starting on its $i$-th configuration).

We let $\Agentset_{\rho}$
be the set of agents "appearing" in $\run$, and set $\Agentset^d_{\rho} \coloneqq \set{a \in \Agentset_{\rho} \mid \dataof(a) = d}$. 
We let $\observedagents{\run}{d}$ be the set of agents with datum $d$ that are "observed" in $\run$, 
\ie, the $a_o \in \Agentset^d_{\rho}$ such that there exists a "step" $\config \step{\sometest}{a}{a_o} \config'$ in $\run$.
For all $q_1, q_m \in Q$, we let $\Agentset^d_{\rho, q_1, q_m}$ be the set of agents with datum $d$ that start in $q_1$ and end in $q_m$,
\ie, the $a \in \Agentset^d_{\rho}$ such that $\config_1(a) = q_1$ and $ \config_m(a) = q_m$. 

\LemmaAgentsCore*

\begin{proof}
	Let $\run : \config_1 \to \config_2 \to \cdots \to \config_{m}$.
	Suppose there are $d\in \Dataset$ and $q_s, q_e \in Q$ such that $\size{\Agentset^d_{\run,q_s, q_e}} > \size{Q}$ (otherwise we can set $\run' \coloneqq \run$).
	Let $\bunchreach$ be the set of states visited by agents of $\Agentset^d_{\run,q_s, q_e}$ during $\run$. Notice that $|\bunchreach| \leq |Q|$. We are going to define a family $(a_q)_{q \in \bunchreach}$ of pairwise distinct agents and corresponding trajectories and identify them with agents in $\Agentset^d_{\run,q_s, q_e}$ such that reducing $\Agentset^d_{\run,q_s, q_e}$ in $\run$ to $(a_q)_{q \in \bunchreach}$ still yields a valid run.
	Repeating the operation for every $\Agentset^d_{\run,q_s, q_e}$ will yield a "run" fulfilling the conditions of the lemma.
	
	We iterate through $\bunchreach$ as follows. Let $q$ be a state in $\bunchreach$ and
	let $f$ be the first moment $q$ is reached in $\run$,
	\ie, the minimal index such that there exists an $a\in \Agentset^d_{\run,q_s, q_e}$
	with $\gamma_f(a)=q$.
	Let $\ell$ be the last moment $q$ is occupied in $\run$,
	\ie, the maximal index such that there exists an $a\in \Agentset^d_{\run,q_s, q_e}$
	with $\gamma_\ell(a)=q$.
	Let $\alpha_q$ be an agent in $\Agentset^d_{\run,q_s, q_e}$ that reaches $q$ first,
	\ie, $\gamma_f(\alpha_q)=q$, and
	let $\beta_q$ be an agent in $\Agentset^d_{\run, q_s, q_e}$ that leaves
	$q$ last,
	\ie, $\gamma_l(\beta_q)=q$.
	Note that these agents do not have to be distinct.
	
	We pick a fresh agent $a_q \notin \Agentset_\run$ with $\dataof(a_q) = d$ and modify $\run$ as follows. We let $a_q$ copy $\alpha_q$ in $\run[\to f]$, then $a_q$ stays idle until $\beta_q$ leaves $q$ (for the last time) and then $a_q$ copies $\beta_q$ in $\run[\ell \to]$. We do this for every $q \in \bunchreach$. To see that this still yields a valid "run" $\run''$, we can apply Lemma \ref{lem:agents-copycat} stepwise.
	
	In $\run''$, for every "step" in which an $a_o$ in $\Agentset^d_{\run,q_s, q_e}$ is observed in state $q$, let $a_q$ be observed instead, \ie, replace steps $\step{\sometest}{a}{a_o}$ with $\step{\sometest}{a}{a_q}$.
	Now remove all steps involving agents in $\Agentset^d_{\run'',q_s,q_e} \setminus \{a_q \mid q \in \bunchreach\}$.
	To see that this yields a valid run, note that we only need to verify that every "step" is valid.
	This is the case, since, whenever an agent in $\Agentset^d_{\rho'',q_s, q_e}$ is observed in state $q$ in $\rho''$, by construction of the trajectory for $a_q$, the agent $a_q$ is also in state $q$.
	To render $\Agentset^d_{\rho'',q_s,q_e}$ a subset of $\Agentset^d_{\rho,q_s,q_e}$, we identify (or substitute) each $a_q$ with a distinct agent in $\Agentset^d_{\rho,q_s,q_e}$.
	
	By applying this transformation to all data in $\Dataset_\run$ and all pairs of states $q_s, q_e$ for which $\size{\Agentset^d_{\run,q_s, q_e}} > \size{Q}$,
	we obtain a "run" $\run'$ in which, for all $d$, $q_s$, $q_e$,
	at most $\size{Q}$ agents go from $q_s$ to $q_e$ with datum $d$.
	In total, for each $d$, there are most $\size{Q}^3$ $d$-agents.
	\ref{agentscore-item1} is guaranteed by the fact that all agents of $\Agentset^d_{\run',q_s, q_e}$ are from $\Agentset^d_{\run,q_s, q_e}$.
	\ref{agentscore-item2} follows from the fact that for all $d$, $q_s$, $q_e$, either $|\Agentset^d_{\run,q_s, q_e}|\leq \size{Q}$ and $\Agentset^d_{\run',q_s, q_e} = \Agentset^d_{\run,q_s, q_e}$ or $|\Agentset^d_{\run,q_s, q_e}|> \size{Q}$ and $|\Agentset^d_{\run',q_s, q_e}| = |Q|$ as we have replaced those agents by the agents $a_q$.
\end{proof}

\LemmaDataCore*

\begin{proof}
	Let $\run \colon \config_1 \to \config_2 \to \cdots \to \config_{m}$.
  We proceed similarly as in the proof of \cref{lem:agents-core-lemma}, lifting the proof from agents to data.

  For a datum $d$ and \AP$i \in \nset{1}{m}$, we let \(\intro* \splittr{\run}{d}(i) \colon Q^3 \to \nats\),
  called the ""split trace of $d$ in $\run$ at $i$"",
  be such that for all $q_1, q_2, q_3 \in Q$, we have
  \[\splittr{\run}{d}(i)(q_1, q_2,q_3) \deff
  \bigsize{\Agentset^d_{\prefixrun{\run}{i},q_1,q_2} \cap \Agentset^d_{\suffixrun{\run}{i},q_2,q_3}}.\]
  Thus, for each triple of states $(q_1, q_2,q_3)$, the value of \(\splittr{\run}{d}(i)(q_1, q_2,q_3)\) is the number of $d$-agents that were in $q_1$ at the start of $\run$, in $q_2$ at the $i$th configuration, and in $q_3$ at the end.

  Let \(\shpair \deff \setc{\splittr{\run}{d}(i)}{d \in \Dataset_\run, i \in \nset{1}{m}}\).
  As there are no more than $K$ agents of each datum in $\run$,
  it holds that \(\splittr{\run}{d}(i)(q_1, q_2, q_3) \in \nset{0}{K}\) for all data \(d\),
  for \(i \in \nset{1}{m}\), and states \(q_1, q_2, q_3 \in Q\).
  Hence, $\size{\shpair} \leq M \deff (K+1)^{\size{Q}^3}$.

  For $\atrace \colon Q^2 \to \nset{0}{K}$,
  let $\Dataset_\run^\atrace \deff \set{d \in \Dataset_{\run} \mid \trace{d}{\run} = \atrace}$
  be the data in $\run$ with "trace" $\atrace$.
  If $\size{\Dataset_\run^\atrace} \leq M$ for all $\atrace \colon Q^2 \to \nset{0}{K}$,
  then \(\size{\Dataset_\run} \leq M \cdot (K+1)^{\size{Q}^2} = (K+1)^{\size{Q}^3+\size{Q}^2}\),
  so \(\rho' \deff \rho\) already fulfils the requirements of the lemma.

  Hence, in the following, let $\atrace \colon Q^2 \to \nset{0}{K}$ be a "trace" with $\size{\Dataset_\run^\atrace} > M$.
We define $\shpair_\atrace = \set{\splittr{\run}{d}(i) \mid d \in \Dataset_\run^\atrace, i\in\nset{1}{m}}$ the set of "split traces" corresponding to some datum $d$ of "run"  $\run$ with "trace" $\atrace$ at some point $i$.
  Note that \(\shpair_\atrace \subseteq \shpair\), and thus, $\size{\shpair_\atrace}\leq M$.

  For each "split trace" $str \in \shpair_\atrace$, we let $f_{str}$ be the minimal index such that there exists a datum $\delta_{str}$ in $\Dataset_\run^\atrace$ such that $\splittr{\run}{\delta_{str}}(f_{str}) = str$.
  Similarly, for each $str \in \shpair_\atrace$, we let $\ell_{str}$ be the maximal index such that there exists a datum $\epsilon_{str}$ in $\Dataset_\run^\atrace$ such that $\splittr{\run}{\epsilon_{str}}(\ell_{str}) = str$.
  Note that those data do not have to be distinct.
  Also note that, by definition, we have $f_{str} \leq \ell_{str}$ for all $str \in \shpair_\atrace$.
As $\splittr{\run}{\delta_{str}}(f_{str}) = str = \splittr{\run}{\epsilon_{str}}(\ell_{str})$,
  we can define a bijection $\gbij_{str} \colon \Agentset_\run^{\delta_{str}} \to \Agentset_\run^{\epsilon_{str}}$
  between $\delta_{str}$-agents and $\epsilon_{str}$-agents appearing in $\run$ such that
  $\config_{1}(a) = \config_{1}(\gbij_{str}(a))$, $\config_{m}(a) = \config_{m}(\gbij_{str}(a))$, and $\config_{f_{str}} (a) = \config_{\ell_{str}} (\gbij_{str}(a))$
  for all \(\delta_{str}\)-agents $a \in \Agentset_\run^{\delta_{str}}$.

  We pick $\size{\shpair_\atrace} \leq M$ arbitrary, but pairwise distinct, data $(\eta_{str})_{str \in \shpair_\atrace}$ in $\Dataset_{\run}^{\atrace}$, one for each "split trace" reached by a datum of $\Dataset_{\run}^{\atrace}$ in $\run$.
For each $str \in \shpair_\atrace$, since $\delta_{str}, \eta_{str} \in \Dataset_{\run}^{\atrace}$ have the same "trace" $\atrace$ in $\run$,
  we can define a bijection $\fbij_{str} \colon \Agentset_\run^{\eta_{str}} \to \Agentset_\run^{\delta_{str}}$ between $\eta_{str}$-agents and $\delta_{str}$-agents appearing in $\run$
  so that $\config_{1}(a) = \config_{1}(\fbij_{str}(a))$ and $\config_{m}(a) = \config_{m}(\fbij_{str}(a))$ for all $\eta_{str}$-agents $a \in \Agentset_\run^{\eta_{str}}$.

  The function $\lbij_{str} \deff \gbij_{str} \circ \fbij_{str} \colon \Agentset_\run^{\eta_{str}} \to \Agentset_\run^{\epsilon_{str}}$
  is then a bijection between $\eta_{str}$-agents and $\epsilon_{str}$-agents
  such that $\config_{1}(a) = \config_{1}(\lbij_{str}(a))$ and $\config_{m}(a) = \config_{m}(\lbij_{str}(a))$ for all $\eta_{str}$-agents $a \in \Agentset_\run^{\eta_{str}}$.

	Intuitively, we will replace $\Dataset_{\run}^{\atrace}$ with a subset $\set{\eta_{str} \mid str \in \shpair_\atrace}$. 
	Agents of each $\eta_{str}$ will be in charge of allowing all "external observations"  on agents of data $d \in \Dataset_{\run}^\atrace$ in carried when that datum matches "split trace" $str$.
	
	To do so, agents of $\eta_{str}$ start by mimicking the moves of $\delta_{str}$ (the first datum of $\Dataset_{\run}^\atrace$ to reach $str$) until $f_{str}$. This is done by making each agent copy its image by $\fbij_{str}$.
	Then, agents of $\eta_{str}$ remain idle and allow all aforementioned "external observations" to be carried out. 
	Finally, when point $\ell_{str}$ is reached in $\run$, the agents of $\eta_{str}$ copy their images by $\lbij_{str}$ until the end. 

	This is formalised by the following claim.
	
  \begin{claim}
    There is a "run"  $\tilde{\run} \colon \tilde{\config}_1 \runto \tilde{\config}_2 \runto \cdots \runto \tilde{\config}_m$
    with \(\Dataset_{\tilde{\run}} = (\Dataset_{\run} \setminus \Dataset_{\run}^{\atrace}) \cup \setc{\eta_{str}}{str \in \shpair_\atrace}\)
    such that for all $i \in \nset{1}{m}$, the following holds.
    \begin{enumerate}
      \item
        \label{claim:data-core-keep-agents}
        For all $d \in \Dataset_{\tilde{\run}}$, we have $\Agentset_{\tilde{\run}}^d = \Agentset_\run^d$.
      \item
        \label{claim:data-core-keep-agent-trajectories}
        For all $d \in \Dataset_{\run} \setminus \Dataset_{\run}^{\atrace}$ and $a \in \Agentset_\run^d$, we have $\tilde{\config}_i(a) = \config_i(a)$.
      \item
        \label{claim:data-core-copy-agent-trajectories}
        For all $str \in \shpair_\atrace$ and all $a \in \Agentset_\run^{\eta_{str}}$,
      \begin{itemize}
        \item if $i \leq f_{str}$, then $\tilde{\config}_{i}(a) = \config_i(\fbij_{str}(a))$,

        \item if $f_{str}\leq i\leq \ell_{str}$, then $\tilde{\config}_{i}(a) = \config_{f_{str}}(\fbij_{str}(a)) = \config_{\ell_{str}} (\textbf{l}_{str}(a))$, and

        \item if $i \geq \ell_{str}$ then $\tilde{\config}_{i}(a) = \config_i(\textbf{l}_{str}(a))$.
      \end{itemize}
    \end{enumerate}
  \end{claim}

  \begin{claimproof}
    We inductively define configurations \(\tilde{\config}_i\) and runs
    \(\tilde{\config}_1 \runto \cdots \runto \tilde{\config}_i\) for all \(i \in \nset{1}{m}\)
    that satisfy \cref{claim:data-core-keep-agents,claim:data-core-keep-agent-trajectories,claim:data-core-copy-agent-trajectories}.
    For all $a \in \Agentset$, we set $\tilde{\config}_1(a) \deff \config_{1}(a)$ if $\dataof(a) \notin \Dataset_{\run}^{\atrace}$ or if $\dataof(a) = \delta_{str}$ for some $str \in \shpair_\atrace$,
    and we set $\tilde{\config}_1(a) \deff *$ otherwise.
Now assume we have defined $\tilde{\run}$ up to $\tilde{\config}_i$ with $i<m$.
    We continue with the "run" up to $\tilde{\config}_{i+1}$.
Let $\config_{i} \step{\sometest}{a}{a_o} \config_{i+1}$ be the "step" at hand, and let $q \trans{q_o}{\sometest} p$ be the used transition.
\begin{itemize}
      \item
        If $\dataof(a_o) \notin \Dataset_{\run}^{\atrace}$,
        by \cref{claim:data-core-keep-agents,claim:data-core-keep-agent-trajectories},
        we have $\tilde{\config}_{i}(a_o) = \config_i(a_o) = q_o$,
        and we set $\tilde{a}_o \deff a_o$.

      \item
        If $d_o \deff \dataof(a_o) \in \Dataset_{\run}^{\atrace}$, then let $str_o \deff \splittr{\run}{d_o}(i)$.
        Since \(f_{str_o}\) and \(\ell_{str_o}\) were defined as the minimal and maximal indices where \(str_o\) occurs,
        we have $f_{str_o} \leq i \leq \ell_{str_o}$.
        Moreover, since $\config_i(a_o) = q_o$, we must have $str_o(q_{1,o}, q_o, q_{m,o}) >0$ for some $q_{1,o}, q_{m,o} \in Q$.
        As a result, there exists a $\delta_{str_o}$-agent $a'_o$ such that $\config_{f_{str_o}}(a'_o) = q_o$.
        We let \(\tilde{a}_o \deff \fbij_{str_o}^{-1}(a'_o)\).
        Then, by \cref{claim:data-core-copy-agent-trajectories}, we have
        $\tilde{\config}_{i}(\tilde{a}_o) = \config_{f_{str_o}}(a'_o) = q_o$.
    \end{itemize}
In both cases, we have $\tilde{\config}_{i}(\tilde{a}_o) = q_o$.
    We now define the "steps" $\tilde{\config}_i \runto \tilde{\config}_{i+1}$ based on another case distinction.
\begin{itemize}
      \item If $\dataof(a) \notin \Dataset_{\run}^{\atrace}$,
        then we use the "step" $\tilde{\config}_i \step{\sometest}{a}{\tilde{a}_o} \tilde{\config}_{i+1}$ based on the transition $q \trans{q_o}{\sometest} p$.
        If $\dataof(a_o) \notin \Dataset_{\run}^{\atrace}$, then $\tilde{a}_o = a_o$,
        and this transition can be taken because $\tilde{\config}_i(a) = \config_i(a)$ and $\tilde{\config}_i(a_o) = \config_i(a_o)$.
        Otherwise, if $\dataof(a_o) \in \Dataset_{\run}^{\atrace}$,
        then $\dataof(a_o) \neq \dataof(a)$, so $\sometest$ is $\neq$.
        Further, $\dataof(\tilde{a}_o) \in \set{\eta_{str} \mid str \in \shpair_\atrace}$, so $\dataof(\tilde{a}_o) \neq \dataof(a)$.
        Also, we defined $\tilde{a}_o$ so that $\tilde{\config}_i(\tilde{a}_o) = \config_i(a_o)$.
        Then the transition can be taken because $\tilde{\config}_i(a) = \config_i(a)$ and $\tilde{\config}_i(\tilde{a}_o) = \config_i(a_o)$.

      \item If $\dataof(a) \in \Dataset_{\run}^{\atrace} \setminus \set{\delta_{str}, \epsilon_{str} \mid str\in \shpair_\atrace}$, then we ignore the "step" and set $\tilde{\config}_{i+1} \deff \tilde{\config}_{i}$.

      \item If $\dataof(a) \in \set{\delta_{str}, \epsilon_{str} \mid str\in \shpair_\atrace}$, then, for all $str \in \shpair_\atrace$ we sequentially apply the following.
        \begin{itemize}
          \item If $i+1 \leq f_{str}$ and $\dataof(a) = \delta_{str}$, then we move the agent $\tilde{a} \deff \fbij^{-1}_{str}(a)$.
            By \cref{claim:data-core-copy-agent-trajectories},
            we know that $\tilde{\config}_i(\tilde{a}) = \config_i(\fbij_{str}(\tilde{a})) = \config_i(a) = q$.
            Furthermore, we have $\tilde{\config}_i(\tilde{a}_o) = \config_i(a_o) = q_o$.
            Finally, if $\sometest$ is $=$, then it holds that $\dataof(a_o) = \delta_{str}$,
            and thus \(\dataof(\tilde{a}) = \eta_{str} = \dataof(\tilde{a}_o)\).
            If $\sometest$ is $\neq$, then by definition of $f_{str}$, as $i< f_{str}$, at this point, no datum of "trace" $\atrace$ matches "split trace" $str$.
            That is, for all data $d \in \Dataset_\run^\atrace$, we have $\splittr{\run}{d}(i) \neq str$.
            Thus, for $d_o \deff \dataof(a_o)$, we have $d_0 \not\in \Dataset_\run^\atrace$ or $str_o \deff \splittr{\run}{d_o}(i) \neq str$.
            In the first case, we have $\tilde{a}_o = a_o$, so $\dataof(\tilde{a}_o) \not\in \Dataset_\run^\atrace$, but $\dataof(\tilde{a}) = \eta_{str} \in \Dataset_\run^\atrace$.
            In the second case, we have \(\dataof(\tilde{a}_o) = \eta_{str_o} \neq \eta_{str} = \dataof(\tilde{a})\),
            since the data $(\eta_{str'})_{str' \in \shpair_\atrace}$ are pairwise distinct.
            In both cases, we have $\dataof(\tilde{a}_o) \neq \eta_{str}$.
         
            Hence, in all cases, we can let the agent $\tilde{a}$ take transition $q \trans{q_o}{\sometest} p$.

          \item Similarly, if $i \geq \ell_{str}$ and $\dataof(a) = \epsilon_{str}$,
            then we move the agent $\tilde{a} \deff \lbij^{-1}_{str}(a)$.
            By \cref{claim:data-core-copy-agent-trajectories},
            we know that $\tilde{\config}_i(\tilde{a}) = \tilde{\config}_i(\lbij^{-1}_{str}(a)) = \config_i(a) = q$.
            Furthermore, we have $\tilde{\config}_i(\tilde{a}_o) = \config_i(a_o) = q_o$.
            Finally, if $\sometest$ is $=$, then it holds that $\dataof(a_o) = \epsilon_{str}$,
            and thus \(\dataof(\tilde{a}) = \eta_{str} = \dataof(\tilde{a}_o)\).
            If $\sometest$ is $\neq$ then by definition of $\ell_{str}$, as $i\geq \ell_{str}$, $\dataof(a_o)$ cannot have "trace" $\atrace$ and match "split trace" $str$, as otherwise it would still match it at "step" $i+1> \ell_{str}$. 
            Therefore, analogously to the case $i+1 \leq f_{str}$ and $\dataof(a) = \delta_{str}$,
            for $d_o \deff \dataof(a_o)$, we have $d_0 \not\in \Dataset_\run^\atrace$ or $str_o \deff \splittr{\run}{d_o}(i) \neq str$.
            Again, this implies that $\dataof(\tilde{a}_o) \neq \eta_{str}$.
            In all cases, we can let the agent $\tilde{a}$ take transition $q \trans{q_o}{\sometest} p$.
        \end{itemize}

        Note that for each $str$, at most one of the two cases applies.
        Furthermore, all the moving agents $\tilde{a}$ are distinct,
        as they have different data (since the $\eta_{str}$ are distinct).
        Moreover, $\tilde{a}_o$ does not move at any point.
        Therefore, as all those steps are enabled in $\tilde{\config}_i$,
        they can all be taken sequentially to get to $\tilde{\config}_{i+1}$.
    \end{itemize}
It is straightforward to show that the induction hypothesis is maintained in all these cases.
    Since the moves of agents $a$ with $\dataof(a) \notin \Dataset_{\run}^{\atrace}$ do not change,
    clearly, $\tilde{\config}_{i+1}(a) = \config_{i+1}(a)$.
For agents $a$ with $\dataof(a) = \eta_{str}$ for some $str \in \shpair_\atrace$,
    note that we make $a$ follow the same steps as $\fbij(a)$ until point $f_{str}$,
    then stay idle until $\ell_{str}$,
    and then follow the same steps as $\lbij(a)$.

    This concludes our induction.
  \end{claimproof}

  The following claim is a direct consequence of the previous one,
  as all remaining data have preserved their initial and final configuration, thus their traces.
  Moreover, we have only deleted data with "trace" $\atrace$.
  Since there is a datum with "trace" $\atrace$ in $\run$,
  the set $\shpair_\atrace$ is not empty,
  and thus there are data in $\run'$ with "trace" $\atrace$, \ie, the datum $\eta_{str}$ for every \(str \in \shpair_\atrace\).
  Hence, as all traces appearing in $\run$ are represented in $\run'$.
  
  \begin{claim}
    For all $\atrace \colon Q^2 \to [0,K]$, there is a "run" $\tilde{\run} \colon \tilde{\config}_1 \runto \tilde{\config}_m$ over the set of agents appearing in $\run$ with datum in $(\Dataset_{\run} \setminus \Dataset_{\run}^{\atrace}) \cup \set{\eta_{str} \mid str \in \shpair_\atrace}$ such that
    \begin{itemize}
      \item for all $d \in \Dataset_{\run'}$ and all $d$-agent $a$, we have $\config_1(a) = \tilde{\config}_1(a)$ and $\config_m(a) = \tilde{\config}_m(a)$,
      
      \item for all $d \in \Dataset_{\run}$, there exists $d'$ such that $\trace{d'}{\run'} = \trace{d}{\run}$, and
      
      \item there are at most $M$ data $d'$ such that $\trace{\run'}{d'} = \atrace$.
    \end{itemize}
  \end{claim}

  The statement of \cref{lem:data-core-lemma} follows by iteratively applying the claim
  for each "trace" $\atrace \colon Q^2 \to \nset{0}{K}$ with $\size{\Dataset_\run^\atrace} > M$;
  as argued above, once we have $\size{\Dataset_\run^\atrace} \leq M$ for all $\atrace \colon Q^2 \to \nset{0}{K}$,
  then \(\size{\Dataset_\run} \leq M \cdot (K+1)^{\size{Q}^2} = (K+1)^{\size{Q}^3+\size{Q}^2}\),
  so \(\rho' \deff \rho\) fulfils the requirements of the lemma.
\end{proof}
 	\section{From Expressions to Containers: Proofs for Section~\ref{sec:equivalence-relation}}
\label{app-GRE-to-boxes}

\LemmaContainerMonotonicity*

\begin{proof}
  First, we show that every $n_1$-"box" is a union of $n_2$-"boxes".
  Let $\config, \chi \in \configset$ and $d, d' \in \Dataset$ with $\cubeapprox{n_2}{\config}{d} = \cubeapprox{n_2}{\chi}{d'}$.
  For every state $q$, there are either at least $n_2$ agents with datum $d$ in $\config$ and at least $n_2$ agents with datum $d'$ in $\chi$,
  in which case there are at least $n_1$ agents in both, or the two numbers are the same.

  As a result, $\cubeapprox{n_1}{\config}{d} = \cubeapprox{n_1}{\chi}{d'}$.
  Hence, the partition of $\configset \times \Dataset$ induced by $n_2$-"boxes" is at least as fine as the one induced by $n_1$-"boxes".

  Now let $\config, \chi \in \configset$ such that $\metaapprox{n_2}{M_2}{\config} = \metaapprox{n_2}{M_2}{\chi}$.
  We show that $\metaapprox{n_1}{M_1}{\config} = \metaapprox{n_1}{M_1}{\chi}$.
  Let $\abox \in \Boxes{n_1}$. Then $\abox$ is a union of $n_2$-"boxes" $\abox_1, \ldots, \abox_k$.
  If one of the two "configurations" has less than $M_1$ data mapped to $\abox$, then, as $M_2\geq M_1$, it has less than $M_2$ data mapped to each $\abox_i$.
  As a consequence, the other "configuration" has the same number of data mapped to each $\abox_i$, and thus the same number of data mapped to $\abox$.

  This shows that $(n_2, M_2)$-"containers" form a partition of $\configset$ that is at least as fine as $(n_1, M_1)$-"containers", concluding our proof.
\end{proof}

\subsection{Proof of \cref{prop-equiv-containers-predicates} and Comparison Between Sizes of Representations}
\label{app-equiv-containers-predicates}

\PropEquivContPred*

We start with the translation from "interval predicates" to "containers".
First, in \cref{lem-pred-to-cont-simple}, we show that a "simple interval predicate" of "height@@pred" $n$ and "width@@pred" $M$ cannot distinguish configurations that are equivalent with respect to $\equiv_{n,M}$.
Then, in \cref{lem-pred-to-cont}, we use this fact to prove that an "interval predicate" of "height@@pred" $n$ and "width@@pred" $M$ can be translated into a union of $(n,M)$-"containers".

\begin{lemma}
  \label{lem-pred-to-cont-simple}
  Let $n,M \in \nats$, and let $\config, \chi$ be configurations such that $\config \equiv_{n,M} \chi$.
  Furthermore, let $\psi$ be a "simple interval predicate" of "height@@pred" at most $n$ and "width@@pred" at most $M$.
  Then $\config$ satisfies $\psi$ if and only if $\chi$ does.
\end{lemma}

\begin{proof}
  Let $\psi = \dot\exists \datvar_1, \ldots, \datvar_M, \, \bigwedge_{q \in Q} \bigwedge_{j=1}^M \#(q,\datvar_j) \in \nset{A_{q,j}}{B_{q,j}}$ be a "simple interval predicate" of "height@@pred" $n$ and "width@@pred" $M$.

  Suppose $\config$ satisfies $\psi$.
  Then there are pairwise distinct data $\adatum_1, \ldots, \adatum_M$ such that for all $q \in Q$ and $j \in \nset{1}{M}$, it holds that $\counting{\config}{\adatum_j}(q) \in \nset{A_{q,j}}{B_{q,j}}$.

  Let $\abox$ be an $n$-"box", and let $\adatum_{\abox,1}, \ldots, \adatum_{\abox,k_\abox}$ be the pairwise distinct data among $(\adatum_i)_{1\leq i\leq M}$ with $\cubeapprox{n}{\config}{\adatum_i} = \abox$.
  As $k_\abox \leq M$ and $\config \metaequiv{n}{M} \chi$, there are pairwise distinct data $d'_{\abox,1}, \ldots, d'_{\abox,k_\abox}$ such that $\cubeapprox{n}{\chi}{d'_{\abox,j}} = \abox$ for all $j \in \nset{1}{k_\abox}$.
  By doing this for every $n$-box, we obtain pairwise distinct data $\adatum'_{1}, \ldots, \adatum'_{M}$ such that $(\config, \adatum_i) \cubeequiv{n} (\chi, \adatum'_i)$ for all $i \in \nset{1}{M}$.
  Moreover, for all $i \in \nset{1}{M}$, since $A_{q,i}, B_{q,i} \in \nset{0}{n} \cup \set{+\infty}$,
  and since $\counting{\config}{\adatum_i}(q) \in \nset{A_{q,i}}{B_{q,i}}$ and $(\config, \adatum_i) \cubeequiv{n} (\chi, \adatum'_i)$,
  we have $\counting{\chi}{\adatum'_i}(q) \in \nset{A_{q,i}}{B_{q,i}}$.

  This shows that $\chi$ satisfies $\psi$.
  The other direction follows by symmetry.
\end{proof}

\begin{lemma}
  \label{lem-pred-to-cont}
  Let $\phi$ be an "interval predicate" of "height@@pred" at most $n$ and "width@@pred" at most $M$.
  Then $\semanticsEP{\phi}{\prot}$ is a union of $(n,M)$-"containers".
\end{lemma}

\begin{proof}
  Let $\config, \chi$ be "configurations" such that $\config \metaequiv{n}{M} \chi$.
  By definition, $\phi$ is a Boolean combination of "simple interval predicates" $\psi_1, \ldots, \psi_p$ for some $p \in \nats$.
  Furthermore, by \cref{lem-pred-to-cont-simple}, for every $i \in \nset{1}{p}$, the predicate $\psi_i$ is satisfied by $\config$ if and only if it is satisfied by $\chi$.
  Thus, $\phi$ is satisfied by $\config$ if and only if it is satisfied by $\chi$.

  As a result, we obtain that each equivalence class of $\metaequiv{n}{M}$ (\ie, each $(n,M)$-"container") is either fully contained in $\semanticsEP{\phi}{\prot}$ or disjoint from it.
  Since $(n, M)$-"containers" form a partition of the set of "configurations", this implies the statement of \cref{lem-pred-to-cont}.
\end{proof}

The next result takes care of the other direction of the proof of \cref{prop-equiv-containers-predicates}.
That is, we show that any finite union of $(n,M)$-"containers" can be expressed as an "interval predicate" of "height@@pred" $n$ and "width@@pred" $M$.
Combined with the first direction, this shows that the two formalisms are equally expressive.

\begin{lemma}
  Let $\acontainer_1, \ldots, \acontainer_k$ be sets of "configurations" such that for all $i$, $\acontainer_i$ is an $(n_i,M_i)$-"container".
  Then there is an "interval predicate" of "width@@pred" $\max_i n_i$ and "height@@pred" $\max_i M_i$ that defines the set $\bigcup_{i=1}^k \acontainer_i$.
\end{lemma}

\begin{proof}
  As "interval predicates" are closed under disjunction, we only need to show that every $(n, M)$-"container" can be expressed as an "interval predicate" of "width@@pred" $n$ and "height@@pred" $M$.

  Let $\acontainer$ be an $(n,M)$-"container". We construct an "interval predicate" expressing the same set of "configurations".
  Let $\abox$ be an $n$-"box", and let $m \in \nats$.
  We set
  \[\phi_{\abox, \geq m} \deff \dot\exists \datvar_1, \ldots, \datvar_{m}, \, \bigwedge_{q \in Q} \bigwedge_{j=1}^{m} \#(q,\datvar_j) \in \nset{A_{q,j}}{B_{q,j}}\]
  where for all $q \in Q$ and $j \in \nset{1}{m}$, we set $A_{q,j} \deff \abox(q)$, and we set $B_{q,j} \deff \abox(q)$ if $\abox(q)<n$ and $B_{q,j} \deff +\infty$ otherwise.
  This "simple interval predicate" expresses that there are at least $m$ pairwise distinct data that match the $n$-"box" $\abox$.
Let
  \[\psi_{\acontainer} \deff \bigwedge_{\substack{\abox \in \Boxes{n}\\ \acontainer(\abox)<M}}\!\! \bigl( \phi_{\abox, \geq \acontainer(\abox)} \land \neg \phi_{\abox, \geq \acontainer(\abox)+1} \bigr)\ \land\ \bigwedge_{\substack{\abox \in \Boxes{n}\\ \acontainer(\abox)=M}}\!\! \phi_{\abox, \geq M}.\]

  This "interval predicate" expresses that for every $n$-"box" $\abox$, the number of data mapped to $\abox$ matches the corresponding number in $\acontainer$, or that the number of data mapped to $\abox$ is at least $M$ if the corresponding number in $\acontainer$ is $M$.
  Hence, the "interval predicate" $\psi_{\acontainer}$ is satisfied by a "configuration" $\config$ if and only if $\config \in \acontainer$.
\end{proof}

\ClaimSuccinctness*

\begin{proof}
  To see this, suppose we want to express that there are exactly $2^k-1$ data that all have one agent in $q_1$ and no agents in $q_0$.
  This can be expressed with a $(2,2^k)$-"container" (whose binary encoding uses $\bigO(k)$ bits for each number).
  Meanwhile, this cannot be expressed by a union of $(n,M)$-"containers" with $M<2^k$, as they cannot distinguish a "configuration" with $2^k-1$ such data from one with more than $2^k-1$.
  By the proposition above, this means that an "interval predicate" for this set requires "width@@pred" $2^k$,
  and thus its encoding must be of size $\Omega(2^k)$ (as we must have at least this many data variables).

  Conversely, consider the set of configurations containing at least one agent in state $q$.
  This is expressible by a trivial "interval predicate", but it corresponds to the union of all $(1,1)$-"containers" $\acontainer$ such that $\acontainer(\abox)>0$ for some "box" $\abox$ with $\abox(q)>0$.
  One can show that there are $\pow{\pow{\size{Q}}} - \pow{\pow{\size{Q}-1}}$ such "containers".
\end{proof}

\subsection{Proof of Lemma~\ref{lem:eq-classes-coreach}}
\label{app-eq-classes-coreach}

Recall that we chose $\functionf \colon \nats \to \nats$ and $\functiong \colon \nats^2 \to \nats$ to be the functions with
$\intro*\functionf(n) \deff (n+\size{\prot}^3) \cdot \size{\prot}$ and $\intro*\functiong(n,M) \deff \bigl(M+(\size{\prot}^3+1)^{\size{\prot}^3+\size{\prot}^2}\bigr)(n+1)^{\size{\prot}}$.
for all $n, M \in \nats$

\LemEqClassesCoReach*

\begin{proof}
  By \cref{cor:run-few-observed-data-and-agents}, we can assume that $\run$ has at most $\size{\prot}^3$ "observed" agents per datum and at most $(\size{\prot}^3+1)^{\size{\prot}^3+\size{\prot}^2}$ "externally observed" data.
Let $\Agentset_{\config}$ and $\Agentset_{\chi}$ be the sets of agents appearing in $\config_{start}$ and $\chi_{start}$, respectively.
  Let $\Dataset_{\config}$ and $\Dataset_{\chi}$ be their sets of data.

  We are going to define maps $\mu \colon \Dataset_{\chi} \to \Dataset_{\config}$ and $\nu \colon \Agentset_{\chi} \to \Agentset_{\config}$.
  Intuitively, $\mu$ (resp. $\nu$) will map each datum in $\Dataset_{\chi}$ (resp. agent in $\Agentset_{\chi}$) to a datum in $\config_{start}$ (resp. agent in $\Agentset_{\config}$) that it should mimic.
We will ensure that data and agents are mapped to a counterpart in $\config_{start}$ that is ``compatible'' in the sense that they can mimic that counterpart while keeping the same initial and final configurations (up to $n$-approximation for the data).
We will also ensure that $\mu$ maps a datum to each "externally observed" datum in $\run$, so that the former can fulfil the same roles in the run from $\chi_{start}$.
  For the same reason, for each datum $d \in \Dataset_{\chi}$, all "observed" agents in $\run$ with datum $\mu(d)$ must have an agent of $\chi_{start}$ with datum $d$ mapped to them.
  We will be able to satisfy these conditions as $\config_{start}$ and $\chi_{start}$ are equivalent up to sufficient bounds.

  The proof is split in three parts. We first prove two claims establishing the existence of mappings $\mu$ and $\nu$ with the desired properties.
  The rest of the proof is dedicated to the construction of the "run" $\pi \colon \chi_{start} \xrightarrow{*} \chi_{end}$.

  For all $\abox \in \Boxes{\functionf(n)}$ and $\abox' \in \Boxes{n}$,
  we set $\Dataset^{\abox\to}_{\config} \deff \setc{d \in \Dataset_{\config}}{\cubeapprox{\functionf(n)}{\config_{start}}{d} = \abox}$,
  $\Dataset^{\abox\to}_{\chi} \deff \setc{d \in \Dataset_{\chi}}{\cubeapprox{\functionf(n)}{\chi_{start}}{d} = \abox}$,
  $\Dataset^{\to \abox'}_{\config} \deff \setc{d \in \Dataset_{\config}}{\cubeapprox{n}{\config_{end}}{d} = \abox'}$,
  and $\Dataset^{\abox \to \abox'}_{\config} \deff \Dataset^{\abox\to}_{\config} \cap \Dataset^{\to \abox'}_{\config}$.

  \begin{claim}
    \label{claim:equiv-mu}
    There exists a mapping $\mu \colon \Dataset_{\chi} \to \Dataset_{\config}$ such that
    \begin{enumerate}
      \item\label{item:equiv-mu-preserves-equiv} for all $d \in \Dataset_{\chi}$, we have $(\chi_{start}, d) \equiv_{\functionf(n)} (\config_{start}, \mu(d))$,

      \item\label{item:equiv-mu-observed-have-antecedent} for all data $d' \in \Dataset_{\config}$ "externally observed" in $\run$, $\mu^{-1}(d') \neq \emptyset$, and

      \item\label{item:equiv-mu-number-classes} for all $\abox' \in \Boxes{n}$, we have $\size{\mu^{-1}(\Dataset^{\to\abox'}_{\config})} = \size{\Dataset^{\to\abox'}_{\config}}$, or $\size{\Dataset^{\to\abox'}_{\config}} \geq M$ and $\size{\mu^{-1}(\Dataset^{\to\abox'}_{\config})} \geq M$.
    \end{enumerate}
  \end{claim}

  \begin{claimproof}
    As $\config_{start} \metaequiv{\functionf(n)}{\functiong(n,M)} \chi_{start}$, we know that for all $\abox \in \Boxes{\functionf(n)}$,
    it holds that $\size{\Dataset^{\abox\to}_{\config}} = \size{\Dataset^{\abox\to}_{\chi}}$,
    or both $\size{\Dataset^{\abox\to}_{\config}}$ and $\size{\Dataset^{\abox\to}_{\chi}}$ are at least $\functiong(n,M)$.

    In the first case, we let $\mu$ map a datum of $\Dataset^{\abox\to}_{\chi}$ to each datum $\Dataset^{\abox\to}_{\config}$ to form a bijection.
    Otherwise, in the second case, for each $\abox' \in \Boxes{n}$, if $\size{\Dataset^{\abox \to \abox'}_{\config}} \leq M$, then we map a datum of $\Dataset^{\abox\to}_{\chi}$ to each datum of $\Dataset^{\abox \to \abox'}_{\config}$.
    If $\size{\Dataset^{\abox \to \abox'}_{\config}} > M$, then we select $M$ data from it, we add all those which are "externally observed" in $\run$, and we map a datum of $\Dataset^{\abox\to}_{\chi}$ to each one of them.
As there are at most $(\size{\prot}^3+1)^{\size{\prot}^3+\size{\prot}^2}$ "externally observed" data in $\run$, we have selected at most $M + ((\size{\prot}^3+1)^{\size{\prot}^3+\size{\prot}^2})$ data per $n$-"box" $\abox'$.
    Further, as $\size{\Boxes{n}} \leq (n+1)^{\size{\prot}}$, we have selected at most $\functiong(n,M)$ data in total.
    Hence, we can indeed map a datum of $\Dataset^{\abox\to}_{\chi}$ to each one of them.
As $\size{\Dataset^{\abox\to}_{\config}} \geq \functiong(n,M) \geq M \cdot \size{\Boxes{n}}$, there exists $\abox' \in \Boxes{n}$ such that $\size{\Dataset^{\abox \to \abox'}_{\config}} \geq M$.
    We pick any $d \in \Dataset^{\abox \to \abox'}_{\config}$ and map all remaining data of $\Dataset^{\abox\to}_{\chi}$ to $d$.

    This concludes the construction of $\mu$.
    It remains to prove that $\mu$ fulfils the requirements of the claim.
The first two items follow directly from the definition of $\mu$.
    Now let $\abox' \in \Boxes{n}$.
    \begin{itemize}
      \item If $\size{\Dataset^{\abox \to\abox'}_{\config}} \leq M$ for all $\abox \in \Boxes{\functionf(n)}$, then, by definition of $\mu$,
        $\size{\Dataset^{\abox \to\abox'}_{\config}} = \size{\mu^{-1}(\Dataset^{\abox \to\abox'}_{\config})}$ for all $\abox \in \Boxes{\functionf(n)}$.
        As $\Dataset^{\to\abox'}_{\config} = \bigsqcup_{\abox \in \Boxes{\functionf(n)}} \Dataset^{\abox \to\abox'}_{\config}$,
        we obtain $\size{\Dataset^{\to\abox'}_{\config}} = \size{\mu^{-1}(\Dataset^{\to\abox'}_{\config})}$.

      \item If $\size{\Dataset^{\abox \to\abox'}_{\config}} \geq M$ for some $\abox \in \Boxes{\functionf(n)}$, then $\size{\mu^{-1}(\Dataset^{\abox \to\abox'}_{\config})} \geq M$ by definition of $\mu$.
        As $\Dataset^{\abox \to\abox'}_{\config} \subseteq \Dataset^{ \to\abox'}_{\config}$, we obtain that both $\size{\Dataset^{ \to\abox'}_{\config}}$ and $\size{\mu^{-1}(\Dataset^{\abox \to\abox'}_{\config})}$ are at least $M$.
    \end{itemize}
    This concludes the proof of the claim.
  \end{claimproof}

  We now use similar arguments as above to define the mapping $\nu$.
  The proof is almost identical, but we provide it in full to avoid any confusion.
For every $d \in \Dataset_{\chi}$, we let $\Agentset_{\config, d}$ (resp. $\Agentset_{\chi, d}$) be the set of agents in $\Agentset_{\config}$ (resp. $\Agentset_{\chi}$) with datum $d$.
  Moreover, for all $q, q' \in Q$, we set $\Agentset^{q\to}_{\config, \mu(d)} \deff \setc{a \in \Agentset_{\config,\mu(d)}}{\config_{start}(a) = q}$,
  $\Agentset^{q\to}_{\chi,d} \deff \setc{a \in \Agentset_{\chi,d}}{\chi_{start}(a) = q}$,
  $\Agentset^{\to q'}_{\config,\mu(d)} \deff \setc{a \in \Agentset_{\config,\mu(d)}}{\config_{end}(a) = q'}$,
  and $\Agentset^{q \to q'}_{\config, \mu(d)} \deff \Agentset^{q\to}_{\config, \mu(d)} \cap \Agentset^{\to q'}_{\config, \mu(d)}$.

  \begin{claim}
    \label{claim:equiv-nu}
    Let $\mu$ be the mapping constructed in the previous claim.
    There exists a mapping $\nu \colon \Agentset_{\chi} \to \Agentset_{\config}$ such that, for all $d \in \Dataset_{\chi}$, it holds that
    \begin{enumerate}[label = (\Alph*), leftmargin = 0.5cm]
      \item\label{item:equiv-nu-preserves-state} for all $a \in \Agentset_{\chi,d}$, we have $\chi_{start}(a) = \config_{start}\bigl(\nu(a)\bigr)$,

      \item\label{item:equiv-nu-observed-antecedent} for all "observed" $a' \in \Agentset_{\config,\mu(d)}$, we have $\nu^{-1}(a')\neq \emptyset$, and

      \item\label{item:equiv-nu-number-states} for all $q' \in Q$, we have $\size{\nu^{-1}(\Agentset^{\to q'}_{\config,\mu(d)})} = \size{\Agentset^{\to q'}_{\config,\mu(d)}}$,
        or $\size{\nu^{-1}(\Agentset^{\to q'}_{\config,\mu(d)})} \geq n$ and $\size{\Agentset^{\to q'}_{\config,\mu(d)}}\geq n$.
    \end{enumerate}
  \end{claim}

  \begin{claimproof}
    We proceed datum by datum.
    That is, for each $d \in \Dataset_{\chi}$, we define the map $\nu$ over agents in $\chi_{start}$ with datum $d$.

    Let $d \in \Dataset_{\chi}$.
    By \cref{item:equiv-mu-preserves-equiv} of \cref{claim:equiv-mu}, we have $(\chi_{start}, d) \equiv_{\functionf(n)} (\config_{start}, \mu(d))$.
    Thus, we know that for all $q \in Q$, it holds that $\size{\Agentset^{q\to}_{\config,\mu(d)}} = \size{\Agentset^{q\to}_{\chi,d}}$,
    or both $\size{\Agentset^{q\to}_{\config,\mu(d)}}$ and $\size{\Agentset^{q\to}_{\chi,d}}$ are at least $\functionf(n)$.
In the first case, we let $\nu$ form a bijection between $\Agentset^{q\to}_{\chi,d}$ and $\Agentset^{q\to}_{\config,\mu(d)}$.
Otherwise, in the second case, for each $q' \in Q$, if $\size{\Agentset^{q\to q'}_{\config,\mu(d)}} \leq n$, then we select $\size{\Agentset^{q\to q'}_{\config,\mu(d)}}$ agents from $\Agentset^{q\to}_{\chi,d}$ and let $\nu$ form a bijection between them and $\Agentset^{q\to q'}_{\config,\mu(d)}$.
    Otherwise, if $\size{\Agentset^{q\to q'}_{\config,\mu(d)}} > n$, then we select $n$ agents from it, add all those which are "observed" in $\run$,
    and we let $\nu$ be surjective on this set by picking an agent from $\Agentset^{q\to }_{\chi,d}$ for every one of them.

    As there are at most $\size{\prot}^3$ "observed" agents with datum $d$ in $\run$ (by \cref{cor:run-few-observed-data-and-agents}, as mentioned at the beginning of the proof), we have selected at most $n + \size{\prot}^3$ agents per state $q'$.
    Hence, we have selected at most $\functionf(n)$ agents in total.
    Thus, we can indeed map an agent of $\Agentset^{q\to}_{\chi,d}$ to each one of them.

    As $\size{\Agentset^{q\to}_{\config,\mu(d)}} \geq \functionf(n) \geq n \cdot \size{\prot}$, there exists $q' \in Q$ such that $\size{\Agentset^{q\to q'}_{\config,\mu(d)}} \geq n$. We pick any $a \in \Agentset^{q\to q'}_{\config,\mu(d)}$ and map all remaining agents of $\Agentset^{q\to}_{\chi,d}$ to $a$.

    This concludes the construction of $\nu$.
    It remains to prove that $\nu$ fulfils the requirements of the claim.
    The two first items are immediate from the definition.
    The third item is obtained by observing that
    \begin{itemize}
      \item if $\size{\Agentset^{q \to q'}_{\config,\mu(d)}} \leq n$ for all $q \in Q$, then $\size{\nu^{-1}(\Agentset^{q \to q'}_{\config,\mu(d)})} = \size{\Agentset^{q \to q'}_{\config,\mu(d)}}$ for all $q$,
        and thus $\size{\nu^{-1}(\Agentset^{\to q'}_{\config,\mu(d)})} = \size{\Agentset^{\to q'}_{\config,\mu(d)}}$, and
      \item if $\size{\Agentset^{q \to q'}_{\config,\mu(d)}} \geq n$ for some $q \in Q$, then $\size{\nu^{-1}(\Agentset^{q \to q'}_{\config,\mu(d)})} \geq n$, and thus both $\size{\nu^{-1}(\Agentset^{\to q'}_{\config,\mu(d)})}$ and $\size{\Agentset^{\to q'}_{\config,\mu(d)}}$ are at least $n$.
    \end{itemize}
    This concludes the proof of the claim.
  \end{claimproof}

  In the remainder of the proof, we construct a "run" $\pi \colon \chi_{start} \runto \chi_{end}$ with $\chi_{end} \metaequiv{n}{M} \config_{end}$ as required in the statement of the lemma.
  For that, we decompose $\run$ into "steps" $\run = \config_0 \to \config_1 \to \cdots \to \config_m$, with $\config_0 = \config_{start}$ and $\config_m = \config_{end}$.
  By induction, we construct a sequence of "configurations" $\chi_0, \ldots, \chi_m$ and runs $\pi_1, \ldots, \pi_m$ such that, for all $i \in \nset{0}{m}$,
  we have $\config_i\bigl(\nu(a)\bigr) = \chi_i(a)$ for all $a \in \Agentset_\chi$ and, for all $j \in \nset{1}{m}$, it holds that $\pi_j \colon \chi_{j-1} \runto \chi_j$.

  We start by setting $\chi_0 \deff \chi_{start}$.
  By \cref{item:equiv-nu-preserves-state} of \cref{claim:equiv-nu}, we have $\config_0\bigl(\nu(a)\bigr) = \chi_0(a)$ for all $a$.
  Now assume we constructed $\chi_i$ and $\pi_i$ for all $i \leq k$ for some $k<m$.
  We construct $\chi_{k+1}$ and $\pi_{k+1}$ as follows.
  Consider the "step" $\config_k \to \config_{k+1}$,
  and let $\delta = q \trans{q_o}{\sometest} q'$ be the used transition,
  let $a^\config$ be the "observing" agent, and let $a^\config_o$ be the "observed" agent in that step.
Let $\set{a^\chi_1, \ldots, a^\chi_p} \deff \nu^{-1}(a^\config)$. Note that $\mu(\dataof(a^\chi_j)) = \dataof(a^\config)$ for all those agents.
\begin{itemize}
    \item If $\sometest$ is $=$, then $\dataof(a^\config_o) = \dataof(a^\config)$.
      Let $\set{d_1, \ldots, d_r} \deff \mu^{-1}(\dataof(a^\config_o))$.
      By \cref{item:equiv-nu-observed-antecedent} of \cref{claim:equiv-nu}, for all $\ell \in \nset{1}{r}$, since $a^\config_o$ is observed in $\run$,
      there exists $a^\chi_{o,\ell}$ such that $\nu(a^\chi_{o,\ell}) = a^\config_o$ and $\dataof(a^\chi_{o,\ell}) = d_\ell$.
      By the induction hypothesis, for all $\ell$, we have $\chi_k(a^\chi_{o,\ell}) = \config_k(a^\config_o)$, and for all $j$, we have $\chi_k(a^\chi_j) = \config_k(a^\config)$.
      Thus, for every $\ell \in \nset{1}{r}$ and sequentially for every $j \in \nset{1}{d}$, we can execute a "step" $\step{\sometest}{a^\chi_j}{a^\chi_{o,\ell}}$ using transition $\delta$.
      We call this sequence of "steps" $\pi_{k+1}$, and we call the reached "configuration" $\chi_{k+1}$.
    \item If $\sometest$ is $\neq$, then $\dataof(a^\config_o)$ is "externally observed" in $\run$,
      and by \cref{item:equiv-mu-observed-have-antecedent} of \cref{claim:equiv-mu}, there exists some datum $d$ such that $\mu(d)=\dataof(a^\config_o)$.
      Furthermore, as $a^\config_o$ is "observed" in $\run$, by \cref{item:equiv-nu-observed-antecedent} of \cref{claim:equiv-nu},
      there exists $a^\chi_o$ such that $\nu(a^\chi_o) = a^\config_o$ and $\dataof(a^\chi_o) = \mu(\dataof(a^\config_o)) = d$.
      Moreover, for all $j$, it holds that $\mu(\dataof(a^\chi_j)) = \dataof(a^\config) \neq \dataof(a^\config_o) = \mu(\dataof(a^\chi_o))$.
      By the induction hypothesis, we have $\chi_k(a^\chi_o) = \config_k(a^\config_o)$ and,
      for all $j$, it holds that $\chi_k(a^\chi_j) = \config_k(a^\config)$.
      Thus, sequentially for every $j \in \nset{1}{p}$, we can execute a "step" $\step{\sometest}{a^\chi_j}{a^\chi_o}$ using transition $\delta$.
      We call this sequence of "steps" $\pi_{k+1}$, and we call the reached "configuration" $\chi_{k+1}$.
  \end{itemize}
It is easy to check that $\config_{k+1}(\nu(a)) = \chi_{k+1}(a)$ for all $a$ in both cases. This concludes our induction.
  As a result, we obtain a "configuration" $\chi_{end} = \chi_m$ and a "run" $\pi \colon \chi_{start} \runto \chi_{end}$ such that $\config_i(\nu(a)) = \chi_i(a)$ for all $a$ "appearing" in $\chi_{init}$ and $i \in \nset{1}{m}$.

  All that is left to do is establish that $\chi_{end} \metaequiv{n}{M} \config_{end}$.
  As a direct consequence of \cref{item:equiv-nu-number-states} of \cref{claim:equiv-nu}, it holds that for all $d$, we have $(\config_{end}, \mu(d)) \cubeequiv{n} (\chi_{end}, d)$.
  Using \cref{item:equiv-mu-number-classes} of \cref{claim:equiv-mu}, we obtain that $\config_{end} \metaequiv{n}{M} \chi_{end}$, proving the lemma.
\end{proof}

We obtain the following analogous statement for $\poststarsymb$.

\begin{corollary}
  \label{cor:eq-classes-reach}
  For all $n, M \in \nats$ and all "configurations" $\config_{start}, \config_{end}, \chi_{end} \in \configset$, if there is a run $\run \colon \config_{start} \runto \config_{end}$ and $\config_{end} \equiv_{\functionf(n),\functiong(n,M)} \chi_{end}$, then there is a "configuration" $\chi_{start} \in \configset$ with $\chi_{start} \equiv_{n,M} \config_{start}$ and a run $\pi \colon \chi_{start} \runto \chi_{end}$.
\end{corollary}

\begin{proof}
  This result is obtained by reversing every transition in the "IOPPUD", \ie, replacing every $q_2 \trans{q_1}{\sometest} q_3$ with $q_3 \trans{q_1}{\sometest} q_2$.
  We can then apply \cref{lem:eq-classes-coreach} on the reversed system.
\end{proof}

\subsection{Proof of Proposition~\ref{prop:bound-coefs-GRE}}
\label{app-bound-coefs-GRE}

\boundCoefsGRE*

\begin{proof}
  In order to improve the readability of this proof,
  instead of directly using the bounds in terms of $\functionf(n)$ and $\functiong(n,M)$,
  we use two polynomial functions $\poly_1, \poly_2$.
  Although the degree of the polynomials will be larger than necessary, the bounds will suffice to prove \exps\ membership.

  We let $\poly_1, \poly_2 \colon \nats \to \nats$ be polynomial functions such that, for all $n,M \in \nats$ with $n >0$,
  it holds that $\functionf(n) \leq n \cdot \poly_1(\size{\prot})$ and $\functiong(n,M)\leq M \cdot n^{\poly_2(\size{\prot})}$.
  Note that such functions exist by the definition of $\functionf$ and $\functiong$.
  Let $\poly \colon \nats \to \nats, n \mapsto \poly_1(n) \cdot \poly_2(n)$.
  Moreover, let $\alpha$ and $\beta$ be functions with $\alpha(\prot, E, F) \deff \norm{E} \cdot \poly_1(\size{\prot})^{\size{F}}$ and $\beta(\prot, E, F) \deff \norm{E}^{\poly(\size{\prot}) \cdot \size{F}^2}$
  for all protocols $\prot$ and all "GRE" $E, F$.

  For every sub-expression $F$ of $E$, we prove by induction on $\size{F}$ that $\semanticsEP{F}{\prot}$ can be described as a union of $(\alpha(\prot, E, F), \beta(\prot, E, F))$-"containers".
\begin{itemize}

    \item If $F$ is an "interval predicate" then, by \cref{lem-pred-to-cont}, $\semanticsEP{F}{\prot}$ is a union of $(\norm{F},\norm{F})$-"containers",
      and thus, by \cref{lem:metacube-monotonicity}, it is also a union of $(\alpha(\prot, E, F),\beta(\prot, E, F))$-"containers".

    \item Suppose $F = \setcomplement{G}$ for some "GRE" $G$.
      By the induction hypothesis, $\semanticsEP{G}{\prot}$ is a union of $(\alpha(\prot, E, G),\beta(\prot, E, G))$-"containers",
      and thus also of $(\alpha(\prot, E, F),\beta(\prot, E, F))$-"containers" by \cref{lem:metacube-monotonicity}.
      Since $(\alpha(\prot, E, F),\beta(\prot, E, F))$-"containers" form a partition of the set of configurations,
      this shows that $\semanticsEP{E}{\prot}$ is also a union of $(\alpha(\prot, E, F),\beta(\prot, E, F))$-"containers".

    \item If $F = G_1 \cup G_2$, then again, by the induction hypothesis and \cref{lem:metacube-monotonicity},
      both $\semanticsEP{G_1}{\prot}$ and $\semanticsEP{G_2}{\prot}$ are unions of $(\alpha(\prot, E, F),\beta(\prot, E, F))$-"containers",
      and so is $\semanticsEP{F}{\prot}$.

    \item Suppose $F = \prestar{G}$ for some "GRE" $G$.
      By the induction hypothesis, $\semanticsEP{G}{\prot}$ is a union of $(\alpha(\prot, E, G),\beta(\prot, E, G))$-"containers".
      We show that for all "configurations" $\config_1, \config_2$, if we have $\config_{1} \equiv_{\alpha(\prot, E, F),\beta(\prot, E, F)} \config_{2}$
      and $\config_{1}\in \semanticsEP{F}{\prot}$, then $\config_{2}\in \semanticsEP{F}{\prot}$.
By \cref{lem:eq-classes-coreach}, it suffices to show that $\config_1 \equiv_{\functionf(\alpha(\prot, E, G)),\functiong(\alpha(\prot, E, G),\beta(\prot, E, G))} \config_{2}$.
      In the following, we assume that the constant term of $\poly_1$ is positive and that $\poly_1$ does not have any negative coefficients.
      Hence, for all $\prot, E, G$, we have $\alpha(\prot, E, G)>0$.
      Then it holds that
      \begin{align*}
        \functionf(\alpha(\prot, E, G))
        &\leq \alpha(\prot, E, G) \cdot \poly_1(\size{\prot})\\
        &= \norm{E} \cdot \poly_1(\size{\prot})^{\size{G}} \cdot \poly_1(\size{\prot}) \\
        &= \norm{E} \cdot \poly_1(\size{\prot})^{\size{G} + 1}\\
        &= \norm{E} \cdot \poly_1(\size{\prot})^{\size{F}} \\
        &= \alpha(\prot, E, F)
      \end{align*}
      and
      \begin{align*}
        \functiong(\alpha(\prot, E, G),\beta(\prot, E, G))
        &\leq \beta(\prot, E, G) \cdot \alpha(\prot, E, G)^{\poly_2(\size{\prot})}\\
        &= \norm{E}^{\poly(\size{\prot}) \cdot \size{G}^2} \cdot (\norm{E} \, \poly_1(\size{\prot}))^{\size{G} \cdot \poly_2(\size{\prot})} \\
        &\leq \norm{E}^{\poly(\size{\prot}) \, \size{G}^2} \, \norm{E}^{\size{G} \, \poly(\size{\prot})} \\
        &\leq\norm{E}^{\poly(\size{\prot}) \, \size{F}^2}\\
        &= \beta(\prot, E, F).
      \end{align*}
      Using \cref{lem:metacube-monotonicity}, we conclude that
      $\config_1 \equiv_{\functionf(\alpha(\prot, E, G)),\functiong(\alpha(\prot, E, G),\beta(\prot, E, G))} \config_{2}$.

    \item If $F = \poststar{G}$ for some "GRE" $G$, we proceed the same way as in the previous case, using \cref{cor:eq-classes-reach} instead of \cref{lem:eq-classes-coreach}.
  \end{itemize}

  All in all, by induction, this proves that $\semanticsEP{E}{\prot}$ is a union of $(\alpha(\prot,E,E), \beta(\prot,E,E))$-"containers",
  which concludes the proof of \cref{prop:bound-coefs-GRE}.
\end{proof}
 	\section{Decidability and Upper Complexity Bounds: Proofs for Section~\ref{sec:complexity}}
\label{app-complexity-upper-bounds}

\ConfigInGREPspace*

We prove the following auxiliary result:
  \begin{lemma}
  \label{ppud-reachability-pspace}
  The following problem is decidable in \pspace:
  given a "PPUD" $\mathcal{P}$ and two "configurations" $\config_1, \config_2$ described "data-explicitly",
  decide if $\config_1 \runto \config_2$.
\end{lemma}
\begin{proof}
  Let $k$ be the number of data appearing in $\config_{1}$, and $m$ the number of agents in $\config_1$.
  All "steps" preserve the number of data and the number of agents.
  As a result, the graph of configurations reachable from $\config_{1}$ contains only configurations with $k$ data and $m$ agents; we can look for a path between $\config_1$ and $\config_{2}$ in this graph in \pspace.
\end{proof}

We now prove \cref{prop:config-in-GRE-pspace}.
\begin{proof}[Proof of \cref{prop:config-in-GRE-pspace}]
  By \cref{ppud-reachability-pspace}, there exists a polynomial function $\poly_{reach}$ such that we can check reachability between two "configurations" represented "data-explicitly" with $k$ data and $m$ agents in deterministic space $\poly_{reach}(\size{\prot}, k, \log(m))$.
  We denote the size of the encoding of a "GRE" $E$ by $\intro*\encodingsize{E}$.

  Let $\psi = \dot\exists \datvar_1, \ldots, \datvar_M, \, \bigwedge_{q \in Q} \bigwedge_{j=1}^M \#(q,\datvar_j) \in [A_{q,j},B_{q,j}]$ be a "simple interval predicate" of "height@@pred" $n$ and "width@@pred" $M$,
  and let $\config$ a "configuration" described "data-explicitly" with $k$ data and $m$ agents.
  One can check whether $\config$ satisfies $\psi$ by enumerating all injective mappings from $\set{\datvar_i \mid 1\leq i\leq M}$ to the data appearing in $\config$, and checking whether one of them satisfies the interval conditions on the number of agents in each state.
  This requires polynomial space in $k+\log(m)+ \encodingsize{\psi}$.

  Given an "interval predicate" $\phi$ and a "configuration" $\config$, it is then straightforward to check which "simple interval predicates" appearing in $\phi$ are satisfied by $\config$ and infer the satisfaction of the "interval predicate" from this, all in polynomial space.
  Let $\poly_{sat}$ be a polynomial such that $\poly_{sat}(k, \log(m), \encodingsize{\psi})$ bounds the required space.
We define $\poly(x, y, z, t) \deff \poly_{sat}(y,z,t) + \poly_{reach}(x,y,z) + xyz$.

  We proceed by induction on $E$ to show that checking if a "configuration" $\config$ is in $\semanticsEP{E}{\prot}$ can be done in space $\bigO\bigl(\size{E} \cdot \poly(\size{\prot}, k, \log(m), \encodingsize{E})\bigr)$,
  where $k$ and $m$ are the numbers of data and agents appearing in $\config$, respectively.

  \begin{itemize}
    \item
    If $E$ is an "interval predicate", then we simply evaluate whether $\config$ satisfies it.
    It is straightforward to get the truth value of each "simple interval predicate" over a "configuration" in space $\poly_{sat}(k, \log(m), \encodingsize{E})$,
    and then infer the satisfaction of the "interval predicate" from it.

    \item If $E = \setcomplement{F}$, then it suffices to check if $\config$ is in $\semanticsEP{F}{\prot}$ and reverse the answer, which requires no extra space.
      Hence, by the induction hypothesis, we need space at most $\size{F}\cdot\poly(\size{\prot},k,\log(m), \encodingsize{E}) \leq \size{E}\cdot\poly(\size{\prot},k,\log(m), \encodingsize{E})$.

    \item If $E = F_1 \cup F_2$, then it suffices to check if $\config$ is in $\semanticsEP{F_1}{\prot}$, accept if it is the case,
      and erase the previous computation and check if $\config$ is in $\semanticsEP{F_2}{\prot}$ otherwise.
      By the induction hypothesis, the required space is then the maximum of the space required by the two sub-computations, which is at most $\size{E}\cdot \poly(\size{\prot},k, \log(m), \encodingsize{E})$.

    \item If $E = \poststar{F}$, then we use the fact that a "step" of a run never changes the number of data or the number of agents in a "configuration".
      Hence, $\config \in \semanticsEP{E}{\prot}$ if and only if there exists $\config'$ with $k$ data and $m$ agents such that $\config'\in \semanticsEP{F}{\prot}$ and $\config' \runto \config$.
      Thus, we can enumerate all configurations with $k$ data and $m$ agents, which requires space $k\cdot \size{\prot} \cdot \log(m)$,
      and check for each one of them whether it is in $\semanticsEP{F}{\prot}$, which requires space $\size{F} \cdot \poly(k, \log(m), \encodingsize{F})$ by the induction hypothesis.
      For every such configuration $\config'$, we check whether $\config$ is reachable from it, which requires space $\poly_{reach}(\size{\prot},k, \log(m))$.
      As a result, we can check whether $\config$ is in $\semanticsEP{E}{\prot}$ in space $\size{F} \cdot \poly(\size{\prot}, k, \log(m), \encodingsize{F}) + k \cdot \size{\prot} \cdot \log(m) + \poly_{reach}(\size{\prot}, k, \log(m)) \leq \size{E} \cdot \poly(\size{\prot}, k, \log(m), \encodingsize{E})$.

    \item The case $E = \prestar{F}$ is analogous to the previous one.
  \end{itemize}

  Our induction shows that we can check if $\config \in \semanticsEP{E}{\prot}$ in deterministic space at most $\size{E} \cdot \poly(k, \size{Q}, \log(m), \encodingsize{E})$.
  As a result, the problem at hand is in \pspace.
\end{proof}
 	\section{Details for Theorem~\ref{nexptime-hardness}}
\label{app:lowerbound}

We proceed by reduction from the problem of tiling an exponentially large grid. 
Recall the following definitions:
A ""tiling instance"" is a tuple $(2^n, \tcolors, \tiles)$ where $n \geq 1$, $\tcolors$ is a finite set of ""colours@@tile"" with special colour $\twhite$ and $\tiles  = \set{t_1, \dots, t_m} \subseteq \tcolors^4$ is a finite set of ""tiles"". 
A "tile" represents a square whose $4$ edges have a colour; we write $\atile \in \tiles$ as $\atile =: (\topcolor{\atile}, \bottomcolor{\atile}, \leftcolor{\atile}, \rightcolor{\atile})$. 
The \emph{size} of the "tiling problem" is $n + \size{\tcolors} + m$ (implicitly, $2^n$ is encoded in binary). 
 Informally, the "tiling problem" asks for the existence of a "tiling" of the $2^n \times 2^n$ grid, \ie, a mapping $\tiling: \nset{0}{2^n{-}1} \times \nset{0}{2^n{-}1} \rightarrow \tiles$ where the colours of neighbouring tiles match and the borders of the grid are white:
\begin{definition}[\cite{tiling_problems}]
\label{def:tiling-problem}
The following ""tiling problem"" is \nexpt-complete.
\\ {\bf Input}: A "tiling instance" $(2^n,\tcolors, \tiles)$,
\\ {\bf Question}: Does there exist a ""tiling"", \ie, a mapping $\tiling: \nset{0}{2^n{-}1} \times \nset{0}{2^n{-}1} \rightarrow \tiles$ such that:
    \begin{enumerate}
    \item \label{item:vertical_white} for all $i \in \nset{0}{2^n-1}$, $\leftcolor{\tiling(i,0)} = \rightcolor{\tiling(i,2^n-1)} = \twhite$,
    \item \label{item:vertical_match} for all $i \in \nset{0}{2^n-1},j \in \nset{0}{2^n-2}$, $\rightcolor{\tiling(i,j)} = \leftcolor{\tiling(i,j+1)}$,
    \item \label{item:horizontal_white} for all $j \in \nset{0}{2^n-1}$, $\topcolor{\tiling(0,j)} = \bottomcolor{\tiling(2^n-1,j)} = \twhite$,
    \item \label{item:horizontal_match} for all $i \in \nset{0}{2^n-2},j \in \nset{0}{2^n-1}$, $\bottomcolor{\tiling(i,j)} = \topcolor{\tiling(i+1,j)}$?
    \end{enumerate}
\end{definition}

This problem is \nexpt-complete \cite{tiling_problems}. We will therefore provide a polynomial-time reduction from this problem to the negation of the "emptiness problem for GRE".

We use the following notations for "interval predicates": 
\begin{itemize}
\item given a state $q$, $\intro*\presenceformula{q} \deff \dot\exists \datvar, \, \#(q,\datvar) > 0$ is the "simple interval predicate" indicating that $q$ is populated,
\item given a set of states $S \subseteq Q$, $\intro*\absenceformula{S}\deff\bigwedge_{q \in S} (\dot \forall \datvar, \, \#(q,\datvar) =0) = \bigwedge_{q \in S} \neg(\dot \exists \datvar, \, \#(q,\datvar) >0)$ is the "interval predicate" indicating that all states in $S$ are empty,
\item given a state $q$, $\intro*\mandatoryformula{q} \deff \bigwedge_{p \in Q}\dot\forall \datvar, \, (\#(p,\datvar) = 0 \lor \#(q,\datvar) >0)$ is the "interval predicate" indicating that all data appearing in the configuration has some agent on $q$. 
\end{itemize}

\subsection{Encoding the Tiling and Checking for Duplicates}
\label{subsec:tiling_encoding}

\cref{fig:tiling_reduction_1} displays how one encodes the tiling into a configuration and how one checks for duplicates. A \emph{duplicate} designates that two values encode the same tile coordinates. \cref{fig:encoding_tiling} corresponds to where $\tiling$ is encoded. 
We only want $\setof{\reductionGRE}$ to contain configurations where no data type has two agents in the tile type part. To do so, we make it so that, if this condition is violated, one will be able to cover $\cheatstate$. For example, from $h_1(0)$, there is a transition to $\cheatstate$ labelled ``$h_1(1), =$'', so that a configuration where a datum has agents on both $h_1(0)$ and $h_1(1)$ will be in
 $\prestar{\presenceformula{\cheatstate}}$. Similarly, we ensure that a datum does not have agents on states $\statet{i}$ and $\statet{j}$ for $i \ne j$. 

The subprotocol of \cref{fig:no_duplication} allows us to enforce that no two data encode the same tile. We will impose that every datum has at least one agent on $D_a$ and one on $D_b$. $\stateover{3}$ can be covered if at least three agents go to the blue part and not to the sink state; therefore, in order to reach a configuration covering $\dupstate$ but from which one cannot cover $\stateover{3}$, one need to send at most two agents to the blue part. Because the left and the right track must observe each other with disequality tests, this is in fact only possible with agents of different data, one in the left track and one in the right track. But then, the two agents can make it all the way to $\dupstate$ if and only if their data have an agent in common on every bit of horizontal and vertical coordinate which, assuming that they have only one agent per bit, means that the two data encode the same coordinates.
 Overall, the predicate that we add to $\reductionGRE$ for the tiling encoding is $\intro*\encodingGRE$, defined as follows:
\AP\begin{align*}
\encodingGRE \deff \setcomplement{\prestar{\presenceformula{\cheatstate}}} \cap \mandatoryformula{D_a} \cap \mandatoryformula{D_b}  \\ \cap \, \absenceformula{Q_{\mathsf{blue}} \cup Q_{\mathsf{red}}} \cap \setcomplement{\prestar{\presenceformula{\dupstate} \cap \setcomplement{\prestar{\presenceformula{\stateover{3}}}}}}
\end{align*}
where $Q_{\mathsf{blue}}$ and $Q_{\mathsf{red}}$ are the states in the blue and red parts, respectively.

\begin{figure}
\centering
\begin{subfigure}{.4\textwidth}
  \centering
  \begin{tikzpicture}[auto, xscale = 0.6, yscale = 0.8]
\tikzset{every node/.style = {font = {\small}}}
\tikzset{every state/.style = {font = {\scriptsize}, inner sep = 1pt, minimum size = 15pt}}
\begin{scope}[xshift = 6cm, yshift = -0.7cm]
\draw[rounded corners=2mm,dashed,fill=black!10] (-3,-1) -| (2.1,-6.8) -| (-3,-6.8) -- cycle;
\node[align = center, color = black!50] at (-5, -4) {Horizontal \\ coordinate};
\node[state] at (-0.5,-2) (b00h) {$h_1(0)$};
\node[state] at (1.2,-2) (b01h) {$h_1(1)$};
\node[state] at (-0.5,-3.5) (b10h) {$h_2(0)$};
\node[state] at (1.2,-3.5) (b11h) {$h_2(1)$};
\node at (0.35, -4.5) {\Large $\vdots$};
\node[state] at (-0.5,-6) (bn0h) {$h_n(0)$};
\node[state] at (1.2,-6) (bn1h) {$h_n(1)$};
\node at (-2.2, -2) {bit $1$};
\node at (-2.2, -3.5) {bit $2$};
\node at (-2.2, -6) {bit $n$};
\begin{scope}[yshift = -6cm]
\draw[rounded corners=2mm,dashed,fill=black!10] (-3,-1) -| (2.1,-6.3) -| (-3,-6.3) -- cycle;
\node[align = center, color = black!50] at (-5, -4) {Vertical \\ coordinate};
\node[state] at (-0.5,-2) (b00h) {$v_1(0)$};
\node[state] at (1.2,-2) (b01h) {$v_1(1)$};
\node[state] at (-0.5,-3.5) (b10h) {$v_2(0)$};
\node[state] at (1.2,-3.5) (b11h) {$v_2(1)$};
\node at (0.35, -4.25) {\Large $\vdots$};
\node[state] at (-0.5,-5.5) (bn0h) {$v_n(0)$};
\node[state] at (1.2,-5.5) (bn1h) {$v_n(1)$};
\node at (-2.2, -2) {bit $1$};
\node at (-2.2, -3.5) {bit $2$};
\node at (-2.2, -5.5) {bit $n$};
\end{scope}
\end{scope}

\begin{scope}[yshift = 3cm, xshift = -1cm]
\draw[rounded corners=2mm,dashed,fill=black!10] (4,-2.5) -| (9.1,-4.5) -| (4,-4.5) -- cycle;
\node[align = center, color = black!50] at (2,-3.5) {Tile type \\ where $\tiles =:$ \\ $\set{t_i \mid i \in \nset{1}{m}}$};
\tikzset{every state/.style = {minimum size = 25pt}}
\node[state] at (5,-3.5)  {$\statet{1}$};
\node at (6.5, -3.5) {$\dots$};
\node[state, inner sep = 0pt] at (8,-3.5) {$\statet{m}$};
\end{scope}

\end{tikzpicture}
   \caption{The part of $\prot$ encoding the value of $\tiling$. Not depicted are the following transition to $\cheatstate$. For every $i \ne j \in \nset{1}{m}$, $\statet{i}$ has a transition to $\cheatstate$ labelled $\statet{j}, =$. For every $i \in \nset{1}{n}$ and $b \in \set{0,1}$, $h_i(b)$ has a transition to $\cheatstate$ labelled $h_i(1-b),=$; and similarly for the vertical counter.}
  \label{fig:encoding_tiling}
\end{subfigure}
\hfill
\begin{subfigure}{.55\textwidth}
  \centering
  \begin{tikzpicture}[auto, xscale = 0.7, yscale = 0.5]
\draw[rounded corners=2mm,dashed,fill=blue!10] (10.5,10) -| (21.5,-7.5) -| (10.5,-7.5) -- cycle;
\tikzset{every state/.style = {minimum size = 15pt, inner sep = 1pt}}
\tikzset{every node/.style = {align = center}}
\node[state] at (13.5,-8.5) (entry1) {$D_a$};
\node[state] at (18.5, -8.5) (entry2) {$D_b$};
\node[state] at (16,-8.5) (sinkstate) {$\sinkstate$};
\node[state] at (13.5, -6.5) (duptrack1) {};
\node[state] at (18.5, -6.5) (duptrack2) {};
\node[state] at (12.5, -3.5) (bit1track1zero) {\scriptsize $H^a_1(0)$};
\node[state] at (14.5, -3.5) (bit1track1one) {\scriptsize $H^a_1(1)$};
\node[state] at (17.5, -4) (bit1track2zero) {\scriptsize $H^b_1(0)$};
\node[state] at (19.5, -4) (bit1track2one) {\scriptsize $H^b_1(1)$};
\node[state] at (13.5, -1.5) (bit1track1ok) {};
\node[state] at (18.5, -1.5) (bit1track2ok) {};
\begin{scope}[yshift = -1cm]
\node[state, draw = none] at (12.5, 1.5) (bit2track1zero) {};
\node[state, draw = none] at (14.5, 1.5) (bit2track1one) {};
\node[state, draw = none] at (17.5, 1.5) (bit2track2zero) {};
\node[state, draw = none] at (19.5, 1.5) (bit2track2one) {};
\node[align = center] at (16,2.3) (onsoon) {(repeat test for every bit \\ of horizontal and vertical coordinates)};
\end{scope}
\begin{scope}[yshift = -5cm]
\node[state] at (13.5, 8) (transtrack1) {};
\node[state] at (18.5, 8) (transtrack2) {};
\node[state] at (12.5, 11) (h1bnzero) {\scriptsize $V^a_n(0)$};
\node[state] at (14.5, 11) (h1bnone) {\scriptsize $V^a_n(1)$};
\node[state] at (17.5, 10.5) (h2bnzero) {\scriptsize $V^b_n(0)$};
\node[state] at (19.5, 10.5) (h2bnone) {\scriptsize $V^b_n(1)$};
\node[state] at (13.5, 13) (bitntrack1ok) {};
\node[state] at (18.5, 13) (bitntrack2ok) {};
\node[state] at (16, 14) (qdup) {$\dupstate$};

\path[-stealth] 
(duptrack1) edge node[left, xshift = 1pt, yshift = -2pt] {$h_1(0),=$ \\ $H_1^b(0), \ne$} (bit1track1zero)
(bit1track1zero) edge (bit1track1ok)
(bit1track1one) edge (bit1track1ok)
(duptrack1) edge node[right, xshift = -1pt, yshift = -2pt] {$h_1(1),=$ \\ $H_1^b(1), \ne$} (bit1track1one)
(duptrack2) edge node[left, xshift = 2pt, yshift = -2pt] {$h_1(0),=$} (bit1track2zero)
(duptrack2) edge node[right , xshift = -2pt, yshift = -2pt] {$h_1(1),=$} (bit1track2one)
(bit1track2zero) edge node[left, xshift = 2pt, yshift = 2pt] {$H^a_1(0), \ne$} (bit1track2ok)
(bit1track2one) edge node[right, xshift = -2pt, yshift = 2pt] {$H^a_1(1), \ne$} (bit1track2ok)
(bit1track1ok) edge[dashed] node[left, xshift = 2pt, yshift = -2pt] {} (bit2track1zero)
(bit1track1ok) edge[dashed] node[right, xshift = -2pt, yshift = -2pt] {} (bit2track1one)
(bit1track2ok) edge[dashed] node[left, xshift = 2pt, yshift = -2pt] {} (bit2track2zero)
(bit1track2ok) edge[dashed] node[right, xshift = -2pt, yshift = -2pt] {} (bit2track2one)
(transtrack1) edge node[left, xshift = 1pt, yshift = -2pt] {$v_n(0),=$ \\ $V_n^b(0), \ne$} (h1bnzero)
(transtrack1) edge node[right, xshift = -1pt, yshift = -2pt] {$v_n(1),=$ \\ $V_n^b(1), \ne$} (h1bnone)
(transtrack2) edge node[left, xshift = 2pt, yshift = -2pt] {$v_n(0),=$} (h2bnzero)
(transtrack2) edge node[right, xshift = -2pt, yshift = -2pt] {$v_n(1),=$} (h2bnone)
(h1bnzero) edge node[left, xshift = 2pt, yshift = 2pt] {} (bitntrack1ok)
(h1bnone) edge node[right, xshift = -2pt, yshift = 2pt] {} (bitntrack1ok)
(h2bnzero) edge node[left, xshift = 2pt, yshift = 2pt] {$V^a_n(0), \ne$} (bitntrack2ok)
(h2bnone) edge node[right, xshift = -2pt, yshift = 2pt] {$V^a_n(1), \ne$} (bitntrack2ok)
(bitntrack1ok) edge (qdup)
(bitntrack2ok) edge (qdup)
(entry1) edge (duptrack1)
(entry2) edge (duptrack2) 
(entry1) edge (sinkstate)
(entry2) edge (sinkstate);
\end{scope} 

\draw[rounded corners=2mm,dashed,fill=red!10] (10.5,17) -| (21.5,11) -| (10.5,11) -- cycle;
\node[state] at (16,16) (over3) {$\stateover{3}$};
\node[state] at (16,14) (over2) {$\stateover{2}$};
\node[state] at (16,12) (over1) {$\stateover{1}$};
\draw[-stealth, color = black] (13,10) -- (over1);
\draw[-stealth, color = black] (14,10) -- (over1);
\draw[-stealth, color = black] (15,10) -- (over1);
\draw[-stealth, color = black] (16,10) -- (over1);
\draw[-stealth, color = black] (17,10) -- (over1);
\draw[-stealth, color = black] (18,10) -- (over1);
\draw[-stealth, color = black] (19,10) -- (over1);
\path[-stealth] 
(over1) edge[bend left= 30] node[left] {$\stateover{1}, =$} (over2)
(over1) edge[bend right= 30] node[right] {$\stateover{1}, \ne$} (over2)
(over2) edge[bend left= 30] node[left] {$\stateover{2}, =$} (over3)
(over2) edge[bend right= 30] node[right] {$\stateover{2},\ne$} (over3);
\end{tikzpicture}   \caption{The test that no tile is encoded by two data. Transitions labelled by two observations are a shortcut for two chained transitions, one with each observation. If only two agents with distinct data come in the blue part (one from $D_a$ and one from $D_b$), they will be able to cover $\dupstate$ if and only if their data encode the same tile. However, if at least three agents are in the blue part, then $\stateover{3}$ can be covered. }
  \label{fig:no_duplication}
\end{subfigure} 
\caption{The part of $\prot$ encoding the tiling and checking for duplicates.}
\label{fig:tiling_reduction_1}
\end{figure}

For a configuration $\config$, let $\tilingofconfig{\config} : \nset{0}{2^n-1}^2 \to 2^{\tiles}$ be the function that, to $x, y \in \nset{0}{2^n-1}$, maps the set of every $t_i \in \tiles$ such that there is a datum $\datum$ in $\config$ with an agent on $\statet{i}$ and:
\begin{itemize}
\item if the $i$-th bit of $x$ is $b$ then $\datum$ has an agent in $h_i(b)$ in $\config$, 
\item if the $i$-th bit of $y$ is $b$ then $\datum$ has an agent in $v_i(b)$ in $\config$.
\end{itemize}
\begin{lemma}
\label{encodingGREworks}
For every $\config \in \setof{\encodingGRE}$, for every $i,j \in \nset{0}{2^n-1}$, $\size{\tilingofconfig{\config}(i,j)} \leq 1$.   
\end{lemma}
\begin{proof}
Let $\config \in \semanticsEP{\encodingGRE}{}$. 
Assume by contradiction that we have $i,j$ such that $\set{ti,t_j} \subseteq \tilingofconfig{\config}(i,j)$ with $i \ne j$. 
Let $\datum$ (resp. $\datum'$) the data witnessing that $t \in \tilingofconfig{\config}(i,j)$. 
If $\datum = \datum'$ then there is an agent in $\statet{i}$ and one in $\statet{j}$ with same datum, but then one can cover $\cheatstate$ 
from $\config$ which contradicts that $\config \in \semanticsEP{\setcomplement{\prestar{\presenceformula{\cheatstate}}}}{}$. If $\datum \ne \datum'$, because $\config \in \semanticsEP{\mandatoryformula{D_a} \cap \mandatoryformula{D_b}}{}$, 
$\datum$ has an agent $a$ in $D_a$ and $\datum'$ has an agent $a'$ in $D_b$. By making $a$ and $a'$ evolve in the blue part, taking the values of the bits of $x$ and then $y$ and observing each other, they can cover $\dupstate$; by sending all others agents in $D_a$ and $D_b$ to $\sinkstate$, 
we reach a configuration $\config' \in \semanticsEP{\presenceformula{\dupstate} \cap \setcomplement{\prestar{\presenceformula{\stateover{3}}}}}{}$, a contradiction. 
\end{proof}

\subsection{Verifying that the Encoded Tiling is Valid}

We here describe the horizontal verifier, \ie, the part of $\prot$ responsible for verifying left-right frontiers of the tiling. 
The data of agents in this part are irrelevant. Two variables, $x,y \in \nset{0}{2^n-1}$, are encoded in binary in similar fashion to \cref{fig:encoding_tiling}. 
These two variables will be private to the horizontal verifier. There will be exactly one agent encoding each bit; the novelty with respect to \cref{fig:encoding_tiling} is that these agents will be able to move between values $0$ and $1$, so that the values of $x$ and $y$ can be changed. 
The rest of the verifier will contain a single agent, called the ""main agent@@tiling"". Under the guarantee that there is only one agent per bit, the "main agent@@tiling" can directly read and modify the values of $x$ and $y$. 
For example, to change bit $i$ of $x$ from $0$ to $1$, the "main agent@@tiling" goes to state $\mathsf{x_i(0 \to 1)}$, the agent of the bit observes it and moves to value $1$, the "main agent@@tiling" observes this change and continues. We represent the part of the protocol encoding the behaviour of the "main agent@@tiling" in \cref{fig:horizontal_verifier} where, for simplicity, we abstract bit encoding of $x$ and $y$ into direct access and modification of the values of $x$ and $y$ by the "main agent@@tiling". 
The "main agent@@tiling" will also manipulate a variable $\colorvar \in \tcolors$, which can be directly encoded in the state space. 

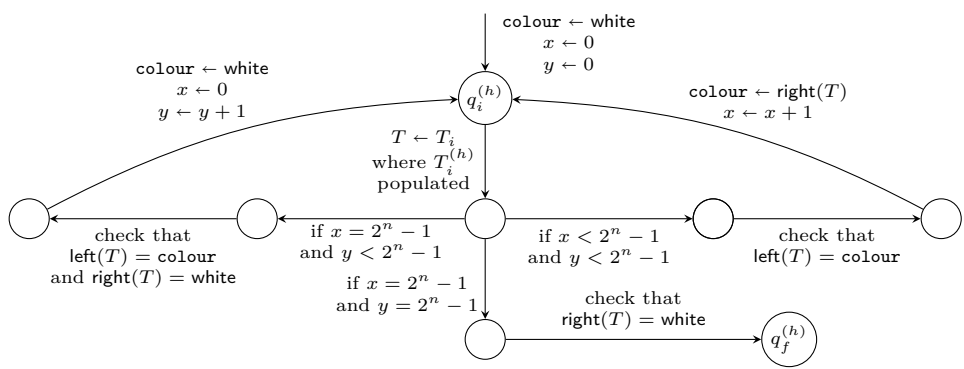
\begin{figure}
\centering
\begin{tikzpicture}[xscale = 1, yscale = 0.8]
\tikzset{every node/.style = {font = {\scriptsize}, align = center}}
\tikzset{every state/.style = {font = {\scriptsize}, minimum size = 15pt}}

\node[state, inner sep = 1pt] at (0,0) (start) {$q_i^{(h)}$};
\draw[-stealth] (0,1.4) -- (start);
\node at (1.1,0.9) {$\colorvar \assignalgo \twhite$ \\ $x \assignalgo 0$ \\ $y \assignalgo 0$};
\node[state] (getT) at (0,-2) {};
\node[state] at (3,-2) (hsmall) {};
\node[state] at (-3,-2) (hbigvsmall) {};
\node[state] at (-6,-2) (hbigvsmallok) {};
\node[state] at (3,-2) (hsmallvsmall) {};
\node[state] at (6,-2) (hsmallvsmallok) {};
\node[state] at (0,-4) (hbigvbig) {};
\node[state, inner sep = 1pt] at (4,-4) (qf) {$q_f^{(h)}$};

\path[-stealth]
(start) edge node[left] 
{$T \assignalgo T_i$ \\ where $T_i^{(h)}$ \\ populated} (getT)
(getT) edge node[below, yshift = 2pt] {if $x = 2^n-1$ \\ and $y < 2^n-1$} (hbigvsmall)
(getT) edge node[below] {if $x < 2^n-1$ \\ and $y < 2^n-1$} (hsmallvsmall)
(getT) edge node[left, yshift = -6pt, xshift = 1pt] {if $x = 2^n-1$ \\ and $y = 2^n-1$} (hbigvbig)
(hsmallvsmall) edge node[below] {check that \\ $\leftcolor{T} = \colorvar$} (hsmallvsmallok)
(hbigvsmall) edge node[below] {check that \\ $\leftcolor{T} = \colorvar$ \\and $\rightcolor{T} = \twhite$} (hbigvsmallok)
(hbigvbig) edge node[above] {check that \\$\rightcolor{T} = \twhite$} (qf)
(hbigvsmallok) edge[bend left = 15] node[above, xshift = -15pt] {$\colorvar \assignalgo \twhite$ \\ $x \assignalgo 0$ \\ $y \assignalgo y+1$} (start)
(hsmallvsmallok) edge[bend right = 15] node[above, xshift = 15pt] {$\colorvar \assignalgo \rightcolor{T}$ \\ $x \assignalgo x+1$} (start)
;

\end{tikzpicture} \caption{The part of the protocol encoding the horizontal verifier. The parts encoding the values of $x$ and $y$ are not represented. Variables $T$ and $\colorvar$ are encoded directly in the state space. The operation ``$T \assignalgo T_i$ where $T_i^{(h)}$ populated'' corresponds to $m = \size{\tiles}$ transitions in parallel, one for each $t_i \in \tiles$; every such transition can be taken if an agent on state $T_i^{(h)}$ is observed. Not depicted is a gadget that allows to go to $\cheatstate$ if at least two agents play the role of horizontal verifier or two agents play the same bit of $x$ or of $y$.}
\label{fig:horizontal_verifier}
\end{figure}

\subsection{Synchronisation with the Verifier}
\label{subsec:synchronization}

Finally,  in the gadget represented in \cref{fig:synchronization_gadget}, each datum will be represented by exactly one agent, called the \emph{reader}, which will initially be in $S^{(h)}$. The reader synchronises with the horizontal verifier in order to let the main agent know what tile type its datum encodes. The goal of synchronisation is to ensure, in relevant runs, that an agent in the middle part of \cref{fig:synchronization_gadget} has a datum value equal to the current value of $z \deff x + 2^n y$. The reader wants to test equality of its own number in $\nset{0}{2^{2n}-1}$ with $z$, this number being equal to its horizontal coordinate plus $2^n$ times the vertical coordinate. However, because the value of $z$ changes throughout the run, it would not suffice that the reader simply tests equality once per bit. However, as we will now see, it suffices that the first synchronisation test tests for equality all bits from most to least significant, and that the second synchronisation test tests for equality all bits from least to most significant.

\begin{lemma}
\label{sync-test}
Consider a run where an agent $a$ starts on $S^{(h)}$ and ends on $F^{(h)}$; assume that its datum faithfully encodes value $k \in \nset{0}{2^N-1}$. Then:
\begin{itemize} 
\item the value of $z$ after the first synchronisation test, which tests bits from most to least significant, is such that $z \geq k$,
\item if we have $z > k$ at the beginning of the second synchronisation test, which tests bits from least to most significant,
 then the agent may not cover $S_f^{(h)}$.
\end{itemize}
\end{lemma}
\begin{proof}
Let $N \deff 2n$.
In the following, a binary number $\ell \in \nset{0}{2^N-1}$ is denoted $\binarynumber{\ell_1 \ell_2 \dots \ell_N}$ where $\ell_i \in \set{0,1}$ is the $i$th most significant bit. 

We first prove the first claim. Let $z^{(m)}$ the value of $z$ after the $m$-th most significant bit is tested in the first synchronisation test. We prove by induction that $z^{(m)} \geq \binarynumber{k_1 \dots k_m 0 \dots 0}$. It is trivially true for $m=0$. Assume that $z^{(m)} \geq \binarynumber{k_1 \dots k_m 0 \dots 0}$: if we have $k_{m+1} = 0$ then we have the property for $m+1$. If $k_{m+1} = 1$, then the first value greater than $z^{(m)}$ whose $(m+1)$-th bit is at least $1$ is greater than $\binarynumber{k_1 \dots k_m k_{m+1} 0 \dots 0}$, concluding the induction. 

We now prove the second claim. Let $z^{(m)}$ now denote the value of $z$ after the test of bit $N-m+1$ in the second synchronisation test; we use convention that $z^{(0)}$ is the value right before the second synchronisation test, so that by hypothesis $z^{(0)} > k$. We prove by induction that there is a bit $j \leq N-m$ such that bit $j$ has value $1$ in $z^{(m)}$ and $0$ in $k$ whereas all bits $i < j$ are equal for $z^{(m)}$ and $k$. This is true for $m=0$ by hypothesis that $z^{(0)} > k$. Assume that it is true for $m$; if $j<N-m-1$ then it is true for $m+1$. Otherwise, we have $j = N-m$, but the test of bit $j$ needs $z^{(m)}$ to have value $0$ at bit $j$; this means that there is a carry between $z^{(m)}$ and $z^{(m+1)}$ that propagates to bit $j-1$ (and possibly further), proving the result. For $m=N$, this yields a contradiction.    
\end{proof}

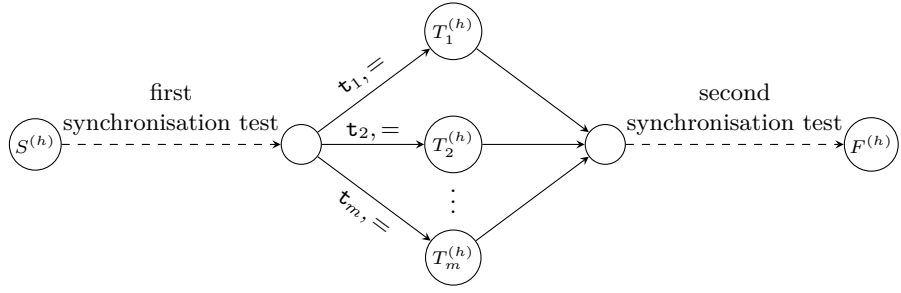
\begin{figure}
\centering
\begin{tikzpicture}[auto]
\tikzset{every node/.style = {font = {\small}}}
\tikzset{every state/.style = {font = {\scriptsize}, inner sep = 1pt, minimum size = 15pt}}

\node[state] (S) at (-1,0) {$S^{(h)}$};
\node[state] (S2) at (2.5,0) {};
\node[state, minimum size = 20pt] (T1) at (4.5,1.5) {$T^{(h)}_1$};
\node[state, minimum size = 20pt] (T2) at (4.5,0) {$T^{(h)}_2$};
\node (Tsusp) at (4.5, -0.65) {$\vdots$};
\node[state, minimum size = 20pt] (TT) at (4.5, -1.5) {$T^{(h)}_{m}$};
\node[state] (F2) at (6.5,0) {};
\node[state] (F) at (10,0) {$F^{(h)}$};

\path[-stealth]
(S) edge[dashed] node[above, align = center] {first \\ synchronisation test} (S2)
(S2) edge node[above, sloped] {$\statet{1}, =$} (T1)
(S2) edge node[above, sloped, yshift = -2pt] {$\statet{2}, =$} (T2)
(S2) edge node[below, sloped] {$\statet{m}, =$} (TT)
(T1) edge (F2)
(T2) edge (F2)
(TT) edge (F2)
(F2) edge[dashed] node[above, align = center] {second \\ synchronisation test} (F)
;

\end{tikzpicture}
 \caption{The synchronisation gadget.}
\label{fig:synchronization_gadget}
\end{figure}

\subsection{Summary of the Reduction}

The protocol built is a combination of all the parts built previously. In \cref{subsec:tiling_encoding}, we have defined a GRE $\encodingGRE$ that, due to \cref{encodingGREworks}, guarantees that, in considered configurations, a given coordinate will have at most one tile type. We now build a larger GRE that expresses that the horizontal and vertical verifiers are able to verifying the entire grid without anyone cheating. Overall, the constructed GRE $\intro*\reductionGRE$ is the following:
\AP
\begin{align*}
\reductionGRE \deff \encodingGRE \cap \setcomplement{\prestar{\presenceformula{\cheatstate}}} \cap \GREinith \cap \GREinitv \cap \prestar{\GREfinalh} \cap \prestar{\GREfinalv} 
\end{align*} 

where:
\begin{itemize}
\item $\intro*\GREinith$ expresses that all states in the horizontal synchronisation gadget except $S^{(h)}$ are empty and that all states in the horizontal verifier gadget are empty except $q_i^{(h)}$ and the states of bits of $x$ and $y$ of value $0$,
\item $\intro*\GREinitv$ is the counterpart of $\GREinith$ but for the vertical verifier and vertical synchronisation,
\item $\intro*\GREfinalh$ expresses that $q_f^{(h)}$ is not empty and that all states in the horizontal synchronisation gadget are empty except $F^{(h)}$, 
\item $\intro*\GREfinalv$ is the counterpart of $\GREfinalh$ but for the vertical verifier and vertical synchronisation.
\end{itemize}

\begin{lemma}
\label{reductionGREimpliestiling}
If $\semanticsEP{\reductionGRE}{}$ is not empty, then the instance of the "tiling problem" is positive.
\end{lemma}
\begin{proof}
Let $\config \in \semanticsEP{\reductionGRE}{}$. Due to \cref{encodingGREworks}, we have that $\size{\tilingofconfig{\config}(x,y)} \leq 1$ for every $x$ and $y$; let $\tiling: \nset{0}{2^n-1}^2 \to \tiles \cup \set{\undefsymb}$ be the function that to $(x,y)$ maps the only element in $\tilingofconfig{\config}(x,y)$ if it exists.

We will prove that $\tiling$ is never equal to $\undefsymb$ and that it is in fact a "tiling". We will prove that the horizontal colours match, \ie, that $\tiling$ satisfies \cref{item:horizontal_white} and \cref{item:horizontal_match} of \cref{def:tiling-problem}; the vertical case is very similar. 

Let $\config'$ such that $\config \runto \config'$ and $\config' \in \semanticsEP{\GREfinalh}{}$. 
By definition of $\GREfinalh$, there is in $\config'$ one agent $a$ in $q_f^{(h)}$. In fact, $a$ is the only agent in the main part of the horizontal verifier in both $\config$ and $\config'$, as otherwise $\config$ would be in $\semanticsEP{\setcomplement{\prestar{\presenceformula{\cheatstate}}}}{}$. We know that $a$ is in state $q_i^{(h)}$ in $\config$ as all other states are forbidden by $\GREinith$. There is at least one agent encoding each bit of $x$ and $y$, otherwise $a$ would not be able to go from $q_i^{(h)}$ to $q_f^{(h)}$; there is also at most one agent per bit, as otherwise one would be able to cover $\cheatstate$ from $\config$. Therefore, one can consider that $x$ and $y$ are encoded reliably; to enforce that, initially, $x = y = 0$, we can either encode it into $\GREinith$ or make the main agent perform an initial sequence of transitions setting the variables to the right values.

In the run from $\config$ to $\config'$, $z \deff x + 2^n y$ has to go incrementally from $0$ to $2^{2n}-1$. Moreover, whenever the transition updated the value of $T$ is taken, we claim that the value $T_i$ taken is such that $\tiling(x,y) = t_i$ (with $x$ and $y$ the values when the transition is taken). Indeed, let $x, y \in \nset{0}{2^n-1}$ and assume that some agent $\syncagent$ is observed on $T_i^{(h)}$ by $a$ with these values of $x$ and $y$.
Because of $\GREinith$ and $\GREfinalh$, we know that $\syncagent$ is in $S^{(h)}$ in $\config$ and in $F^{(h)}$ in $\config'$. 
Let $z \deff x + 2^n y$. Let $\datum$ denote the datum of $\syncagent$ in $\config$.  Because one cannot cover $\cheatstate$ from $\config$, $\datum$ has at most one agent in every bit of the horizontal and vertical coordinates. Moreover, since $\syncagent$ manages to pass the synchronisation tests, there is exactly agent with datum $\adatum$ one per bit. Let $x_\datum$ and $y_\datum$ the values encoded by agents of $\datum$; let $z_\datum \deff x_\datum + y_\datum$. 
Because $\syncagent$ passes the two synchronisation test, by \cref{sync-test} we must have $z = z_\datum$ hence $\tiling(x,y) = t_i$. 

We have proven that, when the verifier takes the transition observing state $T_i^{(h)}$, the value observed indeed corresponds to $\tiling(x,y) \in \tiles$. We conclude that, because the main agent of the horizontal verifier ends on $q_f^{(h)}$, $\tiling$ satisfies \cref{item:horizontal_white} and \cref{item:horizontal_match} of \cref{def:tiling-problem}. We can similarly prove that $\tiling$ also satisfies \cref{item:vertical_white} and \cref{item:vertical_match}, proving that $\tiling$ is a witness that our instance of the "tiling problem" is positive. 
\end{proof}

\begin{lemma}
\label{tilingimpliesnonemptiness}
If the instance of the "tiling problem" is positive, then $\semanticsEP{\reductionGRE}{} \ne \emptyset$. 
\end{lemma}
\begin{proof}
Let $\tiling$ be a witness that the instance is positive. Let $\config$ be a configuration whose appearing data are in the set $\set{\datum_{i,j} \mid i,j \in \nset{0}{2^n-1}}$ and, for every $i,j \in \nset{0}{2^n-1}$:
\begin{itemize} 
\item in the states of \cref{fig:encoding_tiling}, $\datum_{i,j}$ has exactly ${2n}+1$ agents: $n$ agents encoding the value of $i$ in the ``horizontal coordinate'' part, $n$ agents encoding the value of $j$ in the ``vertical coordinate'' part and one agent in the ``tile type'' part, on the state corresponding to $\tiling(i,j)$;
\item in the states of \cref{fig:no_duplication}, $\datum_{i,j}$ has two agents: one on $D_a$ and one on $D_b$;
\item $\datum_{i,j}$ has one agent on $S^{(h)}$ and one agent on $S^{(v)}$, and none in the rest of the horizontal and vertical synchronisation gadgets;
\item $\datum_{i,j}$ has agents in the horizontal and vertical verifier if and only if $i=j = 0$; $\datum_{0,0}$ has one agent on $q_i^{(h)}$, one agent on $q_i^{(v)}$ and $4n$ agents on bits of value $0$ of the four variables of the verifiers (variables $x$ and $y$ of the horizontal verifier and variables $x$ and $y$ of the vertical verifier, which are distinct although, for ease of notation, we never distinguished them above). 
\end{itemize}
Because $\config$ does not have two data with the same coordinates, one cannot cover $\dupstate$ from $\config$ without putting at least three agents in the blue part, which would in turn allow reaching $\stateover{3}$. This proves that $\config \in
\setcomplement{\prestar{\presenceformula{\dupstate} \cap \setcomplement{\prestar{\presenceformula{\stateover{3}}}}}}$. It is quite easy to prove that $\config \in \encodingGRE$ and also that $\config \in \setcomplement{\prestar{\presenceformula{\cheatstate}}} \cap \GREinith \cap \GREinitv$. It finally remains to show that $\config \in \prestar{\GREfinalh \cap \GREfinalv}$. To do so, we consider the execution where the horizontal and vertical verifiers go through all tiles one by one, synchronising with the right agents every time; this execution exists because $\tiling$ satisfies all four conditions from \cref{def:tiling-problem}. Overall, we have $\config \in \semanticsEP{\reductionGRE}{}$ which proves $\semanticsEP{\reductionGRE}{} \ne \emptyset$. 
\end{proof}

With \cref{reductionGREimpliestiling,tilingimpliesnonemptiness}, we have built a polynomial-time reduction from the "tiling problem" to the negation of the emptiness problem for "GRE", which shows that the emptiness problem is $\conexpt$-hard. 

\subsection{Proof of \cref{well-spec-exp}}
\label{app:proof-well-spec-exp}
We formalise \cref{well-spec-exp} into the following proposition:
\begin{proposition}
\label{exp_data_lowerbound_wellspec}
For every $n \in \nats$, there is a protocol $\prot_n$ of size polynomial in $n$ that is not "well-specified" but in which all "fair runs" from configurations with less that $2^n$ data  "stabilise" to $0$.
\end{proposition}
\begin{proof}
For every $n$, we build a protocol $\prot_n$ with states $\cheatstate$ and $q_f$ such that:
\begin{itemize}
\item from every "initial configuration", $\prot_n$ has a run "stabilising" to $\bot$;
\item $Q \setminus \set{q_f}$ has output $\bot$, $q_f$ has output $\top$,
\item an agent on $\cheatstate$ attracts agents on any other state, 
\item an agent on $q_f$ attracts agents on any other state except $\cheatstate$.
\end{itemize}
Here, we say that an agent on state $q$ \emph{attracts} agents in state $q'$ when there is a transition from $q'$ to $q$ that observes the agent on $q$. This implies that, in a fair execution, if $q$ is populated infinitely often then eventually all agents move from $q'$ to $q$. 
The protocol $\prot_n$ is a simplified version of the protocol from \cref{nexptime-hardness}. There are $2n$ \emph{encoding states}, allowing to encode in binary a number in $\nset{0}{2^n-1}$. Each datum must have at most one agent per bit, otherwise by fairness $\cheatstate$ is eventually covered. In a separate gadget, a verifier is able to cover $q_f$ only after checking that, for every value in $\nset{0}{2^n-1}$, there is a datum encoding this value. This uses the same two ideas as in \cref{nexptime-hardness}. First, the verifier manipulates a variable $x \in \nset{0}{2^n-1}$ implemented using another gadget and $n$ agents. Second, there is a synchronisation gadget that makes sure that the verifier only interacts with a datum encoding the current value of $x$. For this second point, there is a non-guarded transition from $q$ to $\cheatstate$ for every state $q$ of the synchronisation gadget except its last state. This guarantees that, eventually, all agents in the synchronisation gadget have successfully passed both synchronisation tests. Again, if two agents play the same role, $\cheatstate$ is eventually covered. In order to stabilise to $\top$, one needs one datum for each value in $\nset{0}{2^n-1}$, hence $2^n$ data.
\end{proof}
 \end{document}